\documentclass[british]{article}
\usepackage{libertineRoman}
\usepackage{biolinum}
\usepackage{libertineMono}
\usepackage[T1]{fontenc}
\usepackage[latin9]{inputenc}
\usepackage{geometry}
\usepackage{color}
\usepackage{array}
\usepackage{prettyref}
\usepackage{refstyle}
\usepackage{float}
\usepackage{amsmath}
\usepackage{amsthm}
\usepackage{amssymb}
\usepackage{cancel}
\usepackage{graphicx}
\usepackage{varwidth}
\PassOptionsToPackage{normalem}{ulem}
\usepackage{ulem}
\usepackage{braket}
\usepackage{comment}

\usepackage{authblk}

\usepackage{qcircuit}

\usepackage[ruled,vlined]{algorithm2e}
\usepackage{enumerate}

\makeatletter

\AtBeginDocument{\providecommand\Subsecref[1]{\ref{Subsec:#1}}}
\AtBeginDocument{\providecommand\Eqref[1]{\ref{Eq:#1}}}
\AtBeginDocument{\providecommand\Figref[1]{\ref{Fig:#1}}}
\AtBeginDocument{\providecommand\Exaref[1]{\ref{Exa:#1}}}
\AtBeginDocument{\providecommand\Defref[1]{\ref{Def:#1}}}
\AtBeginDocument{\providecommand\Algref[1]{\ref{Alg:#1}}}
\AtBeginDocument{\providecommand\Remref[1]{\ref{Rem:#1}}}
\AtBeginDocument{\providecommand\Claimref[1]{\ref{Claim:#1}}}
\AtBeginDocument{\providecommand\Propref[1]{\ref{Prop:#1}}}
\AtBeginDocument{\providecommand\Thmref[1]{\ref{Thm:#1}}}
\AtBeginDocument{\providecommand\Lemref[1]{\ref{Lem:#1}}}
\AtBeginDocument{}
\providecommand{\tabularnewline}{\\}
\newenvironment{cellvarwidth}[1][t]
    {\begin{varwidth}[#1]{\linewidth}}
    {\@finalstrut\@arstrutbox\end{varwidth}}

\floatstyle{ruled}
\newfloat{algorithm}{tbp}{loa}
\providecommand{\algorithmname}{Algorithm}
\floatname{algorithm}{\protect\algorithmname}
\RS@ifundefined{subsecref}
  {\newref{subsec}{name = \RSsectxt}}
  {}
\RS@ifundefined{thmref}
  {\def\RSthmtxt{theorem~}\newref{thm}{name = \RSthmtxt}}
  {}
\RS@ifundefined{lemref}
  {\def\RSlemtxt{lemma~}\newref{lem}{name = \RSlemtxt}}
  {}

\theoremstyle{plain}
\newtheorem{thm}{\protect\theoremname}
\theoremstyle{definition}
\newtheorem{example}{\protect\examplename}
\theoremstyle{plain}
\newtheorem*{thm*}{\protect\theoremname}
\theoremstyle{plain}
\newtheorem{defn}[thm]{\protect\definitionname}
\theoremstyle{plain}

\theoremstyle{remark}

\theoremstyle{remark}
\newtheorem{rem}[thm]{\protect\remarkname}
\theoremstyle{plain}
\newtheorem{lem}[thm]{\protect\lemmaname}
\theoremstyle{remark}
\newtheorem{claim}[thm]{\protect\claimname}
\theoremstyle{plain}
\newtheorem{prop}[thm]{\protect\propositionname}
\theoremstyle{remark}
\newtheorem*{claim*}{\protect\claimname}
\theoremstyle{plain}

\theoremstyle{plain}

\usepackage{hyperref}

\hypersetup{
    colorlinks=true,
    linkcolor=blue,
    filecolor=magenta,      
    urlcolor=cyan,
	citecolor=blue,
    }

\hypersetup{hidelinks,
backref=true,
pagebackref=true,
hyperindex=true,
breaklinks=true,
colorlinks=true,%
urlcolor=blue,
bookmarks=true,
bookmarksopen=false}

\RS@ifundefined{thm}{}{\newref{thm}{name=theorem~,Name=Theorem~,names=theorems~,Names=Theorems~}}
\RS@ifundefined{def}{}{\newref{def}{name=definition~,Name=Definition~,names=definitions~,Names=Definitions~}}
\RS@ifundefined{cor}{}{\newref{cor}{name=corollary~,Name=Corollary~,names=corollaries~,Names=Corollaries~}}
\RS@ifundefined{lem}{}{\newref{lem}{name=lemma~,Name=Lemma~,names=lemmas~,Names=Lemmas~}}
\RS@ifundefined{claim}{}{\newref{claim}{name=claim~,Name=Claim~,names=claims~,Names=Claims~}}
\RS@ifundefined{section}{}{\newref{sec}{name=section~,Name=Section~,names=sections~,Names=Sections~}}
\RS@ifundefined{subsection}{}{\newref{subsec}{name=section~,Name=Section~,names=sections~,Names=Sections~}}
\RS@ifundefined{prop}{}{\newref{prop}{name=proposition~,Name=Proposition~,names=propositions~,Names=Propositions~}}
\RS@ifundefined{conj}{}{\newref{conj}{name=conjecture~,Name=Conjecture~,names=conjectures~,Names=Conjectures~}}
\RS@ifundefined{assu}{}{\newref{assu}{name=assumption~,Name=Assumption~,names=assumptions~,Names=Assumptions~}}
\RS@ifundefined{rem}{}{\newref{rem}{name=remark~,Name=Remark~,names=remarks~,Names=Remarks~}}
\RS@ifundefined{algorithm}{}{\newref{alg}{name=algorithm~,Name=Algorithm~,names=algorithms~,Names=Algorithms~}}
\RS@ifundefined{note}{}{\newref{note}{name=note~,Name=Note~,names=notes~,Names=Notes~}}
\RS@ifundefined{nota}{}{\newref{nota}{name=notation~,Name=Notation~,names=notations~,Names=Notations~}}
\RS@ifundefined{example}{}{\newref{exa}{name=example~,Name=Example~,names=examples~,Names=Examples~}}
\RS@ifundefined{tab}{}{\newref{tab}{name=table~,Name=Table~,names=tables~,Names=Tables~}}
\RS@ifundefined{fact}{}{\newref{fact}{name=fact~,Name=Fact~,names=facts~,Names=Facts~}}
\RS@ifundefined{part}{}{\newref{part}{name=part~,Name=Part~,names=parts~,Names=Parts~}}
\RS@ifundefined{prob}{}{\newref{prob}{name=problem~,Name=Problem~,names=problems~,Names=Problems~}}
\RS@ifundefined{eq}{}{\newref{eq}{name=equation~,Name=Equation~,names=equations~,Names=Equations~}}

\usepackage{color}
\definecolor{purple}{RGB}{0,0,0} %

\definecolor{KB}{rgb}{0.4,0.3,0.9}

\newif\ifcommentfootnote %

\ifcommentfootnote
  \newcommand{\anote}[1]{\footnote{\color{purple}[Andrea: #1]}}

\else
  \newcommand{\anote}[1]{{\color{purple}[Andrea: #1]}}

  \fi

\newcommand{\atul}[1]{#1}

\newcommand{\newblue}[1]{#1}

\usepackage{tikzsymbols}

\usepackage[bottom]{footmisc}

\@ifundefined{showcaptionsetup}{}{%
 \PassOptionsToPackage{caption=false}{subfig}}
\usepackage{subfig}
\makeatother

\usepackage{babel}

\providecommand{\assumptionname}{Assumption}
\providecommand{\claimname}{Claim}
\providecommand{\conjecturename}{Conjecture}
\providecommand{\definitionname}{Definition}
\providecommand{\examplename}{Example}
\providecommand{\lemmaname}{Lemma}
\providecommand{\notationname}{Notation}
\providecommand{\propositionname}{Proposition}
\providecommand{\remarkname}{Remark}
\providecommand{\theoremname}{Theorem}

\newenvironment{construction}[1][htb]
  {\renewcommand{\algorithmname}{Construction}%
   \begin{algorithm}[#1]%
  }{\end{algorithm}}

\newcommand{\IHPC}{A*STAR Quantum Innovation Centre (Q.InC), Institute of High Performance Computing (IHPC), Agency for Science, Technology and Research (A*STAR), Singapore.\looseness=-1}

\newcommand{\CQuERE}{Centre for Quantum Engineering, Research and Education, TCG CREST,  India.\looseness=-1}

\begin{document}

\title{A computational test of quantum contextuality, \\and even simpler proofs of quantumness}

\author{
Atul Singh Arora\footnote{Joint Center for Quantum Information and Computer Science (QuICS), University of Maryland. IQIM, Caltech.}, 
Kishor Bharti\footnote{\IHPC~\CQuERE}, 
Alexandru Cojocaru\footnote{School of Informatics, University of Edinburgh, UK. QuICS, University of Maryland.}, 
Andrea Coladangelo\footnote{Paul G. Allen School of Computer Science and Engineering, University of Washington, USA.}
}

\date{~}

\maketitle
\begin{abstract}
  Bell non-locality is a fundamental feature of quantum mechanics whereby measurements performed on ``spatially separated'' quantum systems can exhibit correlations that cannot be understood as revealing predetermined values. This is a special case of the more general phenomenon of ``quantum contextuality'', which says that such correlations can occur even when the 
  measurements are not necessarily on separate quantum systems, but are merely ``compatible'' (i.e.\ commuting). Crucially, while any non-local game yields an experiment that demonstrates quantum advantage by leveraging the ``spatial separation'' of two or more devices (and in fact several such demonstrations have been conducted successfully in recent years), the same is not true for quantum contextuality: finding the contextuality analogue of such an experiment is arguably one of the central open questions in the foundations of quantum mechanics.

  In this work, we show that an arbitrary contextuality game can be compiled into an operational ``test of contextuality'' involving a single quantum device, 
  by only making the assumption that the device is computationally bounded. Our work is inspired by the recent work of Kalai et al. (STOC '23) that converts any non-local game into a classical test of quantum advantage with a single device. The central idea in their work is to use cryptography to enforce spatial separation within subsystems of a single quantum device. Our work can be seen as using cryptography to enforce ``temporal separation'', i.e.\ to restrict communication between sequential measurements.

  Beyond contextuality, we employ our ideas to design a ``proof of quantumness'' that, to the best of our knowledge, is arguably even simpler than the ones proposed in the literature so far. 
\end{abstract}
\global\long\def\polylog{{\rm poly}({\rm log}(n))}%
\global\long\def\poly{{\rm poly}(n)}%
\global\long\def\ply#1{{\rm poly}(#1)}%
\global\long\def\dQC{{\rm QC}_{d}}%
\global\long\def\dCQ{{\rm CQ}_{d}}%
\global\long\def\tr{{\rm tr}}%
\global\long\def\perm#1#2{\!_{#1}P_{#2}}%
\global\long\def\comb#1#2{\,{}_{#1}C_{#2}}%
\global\long\def\paths{{\rm paths}}%
\global\long\def\parts{{\rm parts}}%
\global\long\def\mat{{\rm mat}}%
\global\long\def\td{{\rm TD}}%
\global\long\def\f{\mathcal{F}}%
\global\long\def\g{\mathcal{G}}%
\global\long\def\negl{\mathsf{negl}}%
\global\long\def\ngl#1{{\mathsf{negl}}(#1)}%
\global\long\def\CQd{\mathsf{CQ_{d}}}%
\global\long\def\QCd{\mathsf{QC_{d}}}%
\global\long\def\BQP{\mathsf{BQP}}%
\global\long\def\BPP{\mathsf{BPP}}%
\global\long\def\QNC{\mathsf{QNC}}%
\global\long\def\CQ#1{\mathsf{CQ_{#1}}}%
\global\long\def\QC#1{\mathsf{QC_{#1}}}%
\global\long\def\CQdp{\mathsf{CQ_{d'}}}%
\global\long\def\QCdp{\mathsf{QC_{d'}}}%
\global\long\def\CH#1{#1\text{-}\mathsf{CodeHashing}}%
\global\long\def\CQC#1{\mathsf{CQC}_{#1}}%
\global\long\def\BQNC{\mathsf{BQNC}}%
\global\long\def\TD{{\rm TD}}%

\global\long\def\NP{\mathsf{NP}}%

\global\long\def\ig{\mathsf{InstanceGen}}%

\global\long\def\qua{\mathsf{qu}}%
\global\long\def\cla{\mathsf{cl}}%

\global\long\def\Gen{\mathsf{Gen}}%
\global\long\def\gen{\mathsf{Gen}}%

\global\long\def\verify{\mathsf{Verify}}%

\global\long\def\prove{\mathsf{Prove}}%

\global\long\def\pk{\mathsf{pk}}%

\global\long\def\sk{\mathsf{sk}}%

\global\long\def\classCQC#1{\BPP^{\QNC_{#1}^{\BPP}}}%

\global\long\def\classQC#1{\QNC_{#1}^{\BPP}}%

\global\long\def\classCQ#1{\BPP^{\QNC_{#1}}}%

\global\long\def\Pivalid{\Pi_{{\rm valid}}}%

\global\long\def\bbF{\mathbb{F}}%

\global\long\def\bit{\{0,1\}}%

\global\long\def\DualLabels{{\rm D.DualLabels}}%

\global\long\def\StdLabels{{\rm D.StdLabels}}%

\global\long\def\AllLabels{{\rm D.AllLabels}}%

\global\long\def\red{{\rm red}}%

\global\long\def\blue{{\rm blue}}%

\global\long\def\colour{\mathsf{O}}%

\global\long\def\QMA{\mathsf{QMA}}%

\global\long\def\QCMA{\mathsf{QCMA}}%

\global\long\def\G{\mathsf{G}}%

\global\long\def\Call{C^{{\rm all}}}%

\global\long\def\spectr{{\mathsf{Spectr}}}%

\global\long\def\coninstance{({\cal H},\left|\psi\right\rangle ,\mathbf{O})}%

\global\long\def\usym{{\cal U}\text{-}{\rm Sym}}%

\global\long\def\relabel{{\rm relabel}}%

\global\long\def\QFHE{\mathsf{QFHE}}%
\global\long\def\FHE{\mathsf{FHE}}%

\global\long\def\Enc{\mathsf{Enc}}%

\global\long\def\Eval{\mathsf{Eval}}%

\global\long\def\cEval{\mathsf{cEval}}%

\global\long\def\calA{{\cal A}}%

\global\long\def\inc{{\rm incnst}}%

\global\long\def\valNC{{\rm valNC}}%

\global\long\def\valQu{{\rm valQu}}%

\global\long\def\inv{\mathsf{Inv}}%

\global\long\def\supp{\mathsf{Supp}}%

\global\long\def\chk{\mathsf{Chk}}%

\global\long\def\samp{\mathsf{Samp}}%

\global\long\def\calD{\mathcal{D}}%

\global\long\def\calL{\mathcal{L}}%

\global\long\def\calM{\mathcal{M}}%

\global\long\def\calY{{\cal Y}}%

\global\long\def\calX{{\cal X}}%

\global\long\def\calF{{\cal F}}%

\global\long\def\calC{{\cal C}}%

\global\long\def\calD{{\cal D}}%

\global\long\def\calP{{\cal P}}%

\global\long\def\calU{{\cal U}}%

\global\long\def\phase{\mathsf{phase}}%

\global\long\def\enc{\mathsf{Enc}}%

\global\long\def\Enc{\mathsf{Enc}}%

\global\long\def\dec{\mathsf{Dec}}%

\global\long\def\Dec{\mathsf{Dec}}%

\global\long\def\OP{\mathsf{OPad}}%

\global\long\def\OPad{\mathsf{OPad}}%

\global\long\def\opad{\mathsf{OPad}}%

\global\long\def\truthtable{\mathsf{TruthTable}}%

\global\long\def\accept{\mathsf{accept}}%

\global\long\def\nonnegl{\mathsf{nonnegl}}%

\global\long\def\pred{\mathsf{pred}}%

\global\long\def\C{\mathsf{context}}%

\global\long\def\PPT{\mathsf{PPT}}%

\global\long\def\QPT{\mathsf{QPT}}%

\global\long\def\const{\mathsf{const}}%

\global\long\def\skp{\mathsf{skip}}%

\global\long\def\Cprimeall{\mathsf{C^{\prime all}}}%

\global\long\def\init{\mathsf{Init}}%

\global\long\def\test{\mathsf{Test}}%

\global\long\def\apply{\mathsf{Apply}}%

\global\long\def\clm{\mathsf{challenger\text{-}lm}}%

\global\long\def\plm{\mathsf{prover\text{-}lm}}%

\global\long\def\prover{\mathsf{Prover}}%

\global\long\def\challenger{\mathsf{Challenger}}%

\global\long\def\conv{\mathsf{conv}}%

\global\long\def\circuit{\mathsf{Circuit}}%

\global\long\def\valClassical{\mathsf{valClassical}}%

\global\long\def\valQuantum{\mathsf{valQuantum}}%

\global\long\def\qubitEnc{\mathsf{qubitEnc}}%
\global\long\def\qubitDec{\mathsf{qubitDec}}%

\global\long\def\sSamp{\mathsf{sSamp}}%
\global\long\def\qProver{\mathsf{qProver}}%

\global\long\def\calS{\mathcal{S}}%

\global\long\def\qstrat{\mathsf{qstrat}}

\global\long\def\calO{\mathcal O}%
\global\long\def\calH{\mathcal H}%

\pagebreak{}

\tableofcontents{}

\newpage{}

\pagenumbering{arabic}

\section{Introduction}\label{subsec:intro}
One of the most intriguing features of quantum mechanics is that, in general, observable properties of a quantum system, usually referred to as ``observables'', do not seem to hold a precise value until they are measured. In technical jargon, quantum mechanics is not a ``local hidden variable theory''. While this feature is now well-understood, Einstein, Podolski, and Rosen, in their paper \cite{epr1935quantum}, originally conjectured that a local hidden variable explanation of quantum mechanics should exist. It was only years later that Bell~\cite{bell1964einstein}, and subsequently Clauser, Horne, Shimony, and Holt (CHSH)~\cite{clauser1969proposed}, in their seminal works, proposed an operational test, i.e.\ an experiment, capable of ruling out a local hidden variable explanation of quantum mechanics. More precisely, they showed that there exists an experiment involving %
measurements on ``spatially separated'' quantum systems such that the outcomes exhibit correlations that cannot be explained by a local hidden variable theory. Such an experiment, usually referred to as a Bell test or a \emph{non-local game}, has been performed convincingly multiple times \cite{hensen2015loophole,giustina2015significant,shalm2015strong,li2018test,rosenfeld2017event,storz2023loophole}. Crucially, a Bell test rules out a local hidden variable theory under the assumption that the devices involved in the experiment are non-communicating (which is usually enforced through ``spatial separation''). %

Bell non-locality may be viewed as a special case of the more general phenomenon of quantum contextuality~\cite{Specker1960-bg,Koc}, which says that such correlations can occur even when the %
measurements are not necessarily on separate quantum systems, but are merely ``compatible'' (i.e.\ commuting). Contextuality has a long tradition in the foundations of quantum mechanics~\cite{budroni2022kochen}. %
However, unlike a non-local game, a more general \emph{contextuality game}\footnote{We choose to use the term contextuality ``game'' here to preserve the analogy with a non-local game. However, in the contextuality literature, the more commonly used term is contextuality ``scenario''. We refer the reader to the technical overview (\Subsecref{tech-overview-contextuality}) for details.} %
does not in general have a corresponding operational test. By an operational test, we mean a test that can be carried out on a device by interacting with it classically, and importantly, without having to make bespoke assumptions about its inner workings. Thus, even though some contextuality games are even simpler than non-local games~\cite{klyachko2008simple}, a satisfactory approach to compiling arbitrary contextuality games into operational ``tests of contextuality'' is missing. This is not for lack of trying---numerous attempts have been made that inevitably have to either resort to strong assumptions about the quantum hardware or assumptions that are hard to enforce in practice, such as the device being essentially \emph{memoryless}.\footnote{These assumptions are typically referred to as ``loopholes'' in the literature:~\cite{lapkiewicz2011experimental,um2013experimental,jerger2016contextuality,zhan2017experimental,malinowski2018probing,leupold2018sustained,zhang2019experimental,um2020randomness,wang2022significant,hu2023self,liu2023experimental}.}
Thus, one of the central open questions in the foundations of quantum mechanics is:

\begin{center}
  {\em Is there a way to compile an arbitrary contextuality game into an ``operational test of contextuality''? %
  }
\end{center}

Beyond demonstrating the presence of genuine quantum behaviour, non-local games can be employed to achieve a much more fine-grained control over the behaviour of untrusted quantum devices: for example, they allow a classical user to verify the correctness of full-fledged quantum computations, by interacting with two non-communicating quantum devices \cite{reichardt2012classical, coladangelo2019verifier}. Of course, any guarantee obtained via non-local games hinges on the physical (and non-falsifiable) assumption that the devices involved are non-communicating. To circumvent the need for this assumption, a lot of the attention in recent years has shifted to the computational setting. This exploration was kick-started by Mahadev's seminal work \cite{mahadev2018classical} showing, via cryptographic techniques, that the verification of quantum computations can be achieved with a single quantum device, under the assumption that the device is computationally bounded. She and her collaborators~\cite{BCMVV21} later proposed what can be thought of as the analogue of a Bell/CHSH experiment with a single device---they proposed a simple test that an efficient quantum device can pass, but that an efficient classical device cannot. Since then, various works have proposed increasingly efficient ``proofs of quantumness'' in this setting~\cite{BKVV20, alagic2020non, MCVY22, KLVY22, alnawakhtha2022lattice,BGK+23}. %
The goal of this line of work is to simplify these tests to the point that they can be implemented on current quantum devices. An experimental realisation of such a proof of quantumness would be a milestone for the field of quantum computation, as it would constitute the first \emph{efficiently verifiable} demonstration of quantum advantage. Towards this goal, the second question that we consider in this work is:

\begin{center}
  {\em Can contextuality help realise simpler proofs of quantumness?}
\end{center}

\subsection{Our results}
\label{sec:results}

Our first contribution is a positive answer to the first question. We show that, using cryptographic techniques, an arbitrary contextuality game can be compiled into an ``operational test of contextuality'' involving a single quantum device, where the only assumption is that the device is computationally bounded.\footnote{{In the contextuality jargon, one might say that all ``loopholes'' are being replaced by a computational assumption.}}
A \emph{contextuality game} involves a single player and a referee (unlike non-local games that always involve more than one player). In an execution of the game, the referee asks the player to measure a \emph{context}---a set of commuting observables---and the player wins if the measurement outcomes satisfy certain constraints. Importantly, a given observable may appear in multiple contexts. If the player uses a strategy where the values of these observables are ``predetermined'' (which is the analogue of a ``local hidden variable'' strategy in the non-local game setting), then the highest winning probability achievable is referred to as the \emph{non-contextual value} of the game, denoted by $\valNC$.\footnote{We emphasise that a non-contextual strategy is such that, if an ``observable'' appears in multiple contexts, then the \emph{same} predetermined value should be returned by the player for that observable in \emph{all} contexts in which it appears. This is precisely the difficulty in realising an ``operational test of contextuality'': how does the referee enforce that the player is consistent across contexts?} On the other hand, if the player uses a quantum strategy, then the highest winning probability achievable is the \emph{quantum value}, denoted by $\valQu$ (which is strictly greater than $\valNC$ for contextuality games of interest).\footnote{A quantum strategy is such that, if an observable appears in multiple contexts, then the \emph{same} observable should be measured to obtain the answer in all such contexts.} 
\newblue{
We formalise what we mean by an ``operational test of contextuality'' by  introducing a property called \emph{faithfulness}. Intuitively, this requires that the test has a strong correspondence to some contextuality game.}
We defer more precise definitions to \Subsecref{criteria}. We show the following. %

\begin{thm}[informal, simplified]\label{thm:infMain1}
Any contextuality game $\G$ can be compiled into a
  single prover, {\mbox{2-round} (i.e.\ \mbox{4-message})}
  \newblue{operational test of contextuality, under standard cryptographic assumptions. More precisely, the test is faithful to $\G$ and satisfies the following:}
   \begin{itemize}
  \item (Completeness) There is a quantum polynomial-time (QPT) prover that wins with probability at least
  \[
    \frac{1}{2}(1 + \valQu) -\negl \,.
  \]
  \item (Soundness) Any probabilistic polynomial-time (PPT) prover wins with probability at most
  \[
    \frac{1}{2}(1 + \valNC) +\negl \,.
  \]
  \end{itemize}
  Here $\negl$ are (possibly different) negligible functions of a security parameter. 
\end{thm}

The bounds in the statement above are for games with contexts of size two (which subsume %
``$2$-player'' non-local games, see \Figref{whatcompileswhat}). The general bounds are similar but depend on the size of the contexts and are stated later. %
Notably, the number of messages, even in the general case (which subsumes non-local games with any number of players), remains constant (i.e.\ four). The cryptographic assumptions we make are the existence of (1) a quantum homomorphic encryption (QFHE) scheme\footnote{With an assumption on the form of encrypted ciphertexts, which is satisfied by both Mahadev's and Brakerski's QFHE schemes, \cite{brakerski18,mahadev2020classical}; see \Subsecref{techoverviewQFHE}.} and (2) a new primitive, the \emph{oblivious pad}, which we introduce below. QFHE can be realised assuming the quantum hardness of the Learning With Errors problem (LWE)~\cite{LWE}. We show how to construct the oblivious pad under the same assumption in the quantum random oracle model---which in turn can be heuristically instantiated using a cryptographic hash function, such as SHA3.\footnote{In the random oracle
model~\cite{BR93} (ROM), a hash function $f$ is modelled as a uniformly random black-box function: parties can evaluate it by sending a query $x$ and receiving $f(x)$ in return. In the \emph{quantum} random oracle model (QROM), such queries can also be made in superposition. %
A proof of security in this model is taken to be evidence for security of the protocol when the black-box is replaced by, for example, a suitable hash function $f$. This is because, informally, any attack on the resulting protocol must necessarily exploit the structure of $f$. %
} 

\begin{figure}
  \begin{centering}
    \includegraphics{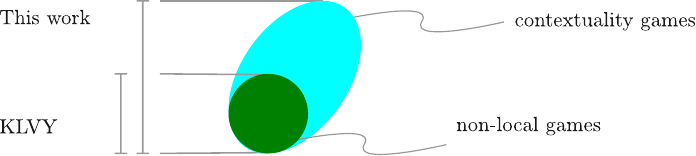} %
    \par\end{centering}
  \caption{The compiler in this work compiles a much larger set of games compared to the one in \cite{KLVY22}.}\label{fig:whatcompileswhat}
\end{figure}

Our result is motivated by the recent work of Kalai et al.\ \cite{KLVY22} (KLVY from here on) that converts any non-local game into a test of quantumness with a single device. %
Consider a non-local game with two players, Alice and Bob. The central idea behind the KLVY compiler is to use
cryptography to enforce spatial separation within a \emph{single} quantum device, i.e.\ to enforce that Alice and Bob's measurements occur on separate subsystems. The KLVY compiler relies on the following mechanism to cryptographically enforce spatial separation: Alice's question and answer are encrypted (using a quantum fully homomorphic encryption (QFHE) scheme), while Bob's question and answer are in the clear. The referee, who holds the decryption key, can then test the correlation across the two. Unfortunately, this approach does not extend to contextuality, at least not in any direct way, since in a contextuality game there is no notion of Alice and Bob.

In contrast, our work can be seen as using cryptography to enforce ``temporal separation'', 
i.e.\ to restrict communication between sequential measurements.
In a nutshell, our idea is to ask the first question under a homomorphic encryption and the second question in the clear, as in KLVY, but with the following important difference. In KLVY, the quantum device prepares an \emph{entangled} state, whose first half is encrypted, and used to homomorphically answer the first question, while the second half is not, and is used to answer the second question in the clear. In our protocol, there are no %
separate subsystems between the two rounds: instead, the encrypted post-measurement state from the first round (which results from the measurement performed to obtain the encrypted answer) is \emph{re-used} in the second round. The technical barrier is, of course, the following: how can the post-measurement state be re-used if it is still encrypted? The two most natural approaches do not work:
\begin{itemize}
\item[(i)] Providing the decryption key to the quantum device is clearly insecure as it allows the device to learn the first question in the clear.
\item[(ii)] Homomorphically encrypting the second question does not work either because the quantum device can correlate its answers to the two questions ``under the hood of the homomorphic encryption''.
\end{itemize}
We circumvent this barrier by introducing a procedure that allows the prover to \emph{obliviously} ``re-encrypt'' the post-measurement state non-interactively. This re-encryption procedure is such that it allows the verifier to achieve the following: the verifier can now reveal some information that allows a quantum prover to access the post-measurement state in the clear, while a PPT prover \emph{does not learn the first question}. More precisely, the verifier does not directly reveal the original decryption key, which would clearly be insecure, as pointed out earlier. Instead, the verifier expects the prover to ``re-encrypt'' its state using the new procedure, and then the verifier is able to safely reveal the resulting ``overall'' decryption key. %
Crucially, while a classical prover is unable to make use of the additional information to beat the classical value $\valNC$, we show that there is an efficient quantum prover that can access the post-measurement state in the clear, and proceed to achieve the quantum value $\valQu$. 

The main technical tool that we introduce to formalise this idea, which may find applications elsewhere, is a primitive that we call {\emph{oblivious (Pauli) pad}}. The oblivious pad takes as input a state $\ket{\psi} $ and a public key $\pk$ and produces a Pauli-padded state $X^xZ^z \ket{\psi}$ together with a string $s$ (which can be thought of as encrypting $x$ and $z$ using $\pk$). 
The string $s$ can be used to recover $x, z$ given the corresponding secret key $\sk$. 
The security requirement is modelled as the following distinguishing game between a PPT prover and a challenger: 
\begin{itemize}
    \item The challenger generates public and secret keys $(\pk,\sk)$ and sends $\pk$ to the prover.
    \item The prover produces a string $s$ and sends it to the challenger.
    \item The challenger either returns $(x, z)$ (as decoded using $s$ and $\sk$), or a fresh pad $(\tilde{x}, \tilde{z})$ sampled uniformly at random.
\end{itemize}
We require that no PPT prover can distinguish between the two cases with non-negligible advantage.\footnote{Note that there is a QPT prover that \emph{can} distinguish these two cases with non-negligible advantage.}
We show that an oblivious pad can be realised based on ideas from~\cite{BCMVV21} in the quantum random oracle model.
Crucially, while in KLVY the second phase of the protocol happens on a \emph{different} subsystem, the oblivious pad is what allows us to carry out the second phase on the \emph{same} subsystem, {as depicted in Figure~\ref{fig:schematicprot}}.

\begin{figure}
  \begin{centering}
    \includegraphics{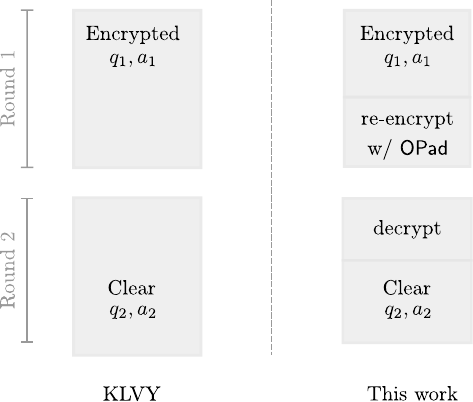}
    \par\end{centering}
  \caption{A schematic comparing the non-local game compiler~\cite{KLVY22} with our contextuality game compiler. %
  The key idea in KLVY is to ask the first question of a non-local game under a homomorphic encryption and the second one in the clear, with the prover using two entangled subsystems (one that is encrypted, and one in the clear). In our compiler, the oblivious pad ($\opad$) allows the prover to ``re-encrypt'' its post-measurement state, just before Round 2. Upon obtaining information about the ``re-encryption'' that took place, the verifier can then safely reveal the ``overall'' decryption key in Round 2, allowing the prover to proceed with the next measurement in the clear.} %
  \label{fig:schematicprot}
\end{figure}

Our second contribution streamlines the ideas introduced to prove \Thmref{infMain1} in order to obtain a 2-round proof of quantumness relying on the classical hardness of the Learning-with-Errors (LWE) problem~\cite{LWE}. Our construction has the main advantage of being simpler than existing ones in the literature, in the sense explained below.
Our proof of quantumness makes use of Noisy Trapdoor Claw-Free functions (NTCFs), introduced in \cite{BCMVV21} (but it does not require the NTCFs to have an ``adaptive hardcore bit'' property). It relies on the particular structure of the ``encrypted CNOT operation'' introduced in Mahadev's QFHE scheme \cite{mahadev2020classical}. We informally state our result, but we defer the construction to the technical overview. %
\begin{thm}[Informal]
Assuming the classical hardness of LWE, there exists a 2-round (i.e.\ 4-message) proof of quantumness with the following properties:
\begin{itemize}
    \item It requires only one coherent evaluation of an NTCF, and one layer of single-qubit Hadamard gates.
    \item The quantum device only needs to maintain a single qubit coherent in-between the two rounds.
\end{itemize}
\end{thm}

Our proof of quantumness can be seen as combining ideas from \cite{KLVY22} and \cite{MCVY22, alnawakhtha2022lattice}. It is simpler than existing proofs of quantumness in the following ways:
\begin{itemize}
  \item \emph{Single encrypted CNOT operation.} The 2-round proof of quantumness of KLVY is based on a QFHE scheme. Concretely, an implementation of their proof of quantumness based on Mahadev's QFHE scheme requires performing a homomorphic controlled-Hadamard gate, which requires three sequential applications of the ``encrypted CNOT operation''.
        Crucially, these three operations require computing the NTCF three times in superposition while maintaining coherence all along. Our 2-round proof of quantumness only requires a single application of the ``encrypted CNOT'' operation. Moreover, in KLVY, the encrypted subsystem, on which Alice's operations are applied homomorphically needs to remain entangled with a second subsystem, which is used to perform Bob's operations in the clear. In contrast, in our proof of quantumness, the encrypted CNOT operation and the subsequent operations in the clear happen on a {single} system (of the same size as Alice's in KLVY).
  \item \emph{Single qubit coherent across rounds and 2 rounds of interaction.}
  Compared to the 3-round proof of quantumness of \cite{MCVY22}, ours requires one less round of interaction. However, more importantly than the number of rounds, the protocol from \cite{MCVY22} requires the quantum device to keep a superposition over preimages of the NTCF %
  coherent in-between rounds, %
  while waiting for the next message. In contrast, both our proof of quantumness and that of KLVY only require the quantum device to keep a \emph{single} qubit coherent in-between rounds.
  \item \emph{Simple quantum operations.}
  The 2-round proof of quantumness of \cite{alnawakhtha2022lattice} matches ours in that it requires a single coherent application of an NTCF based on LWE, and the quantum device only needs to keep a single qubit coherent in between rounds. However, their protocol requires the prover to perform single-qubit measurements in ``rotated'' bases coming from a set of a size that scales linearly with the LWE modulus (for which typical parameters are $\approx 10^2$ to $10^3$).
  Their completeness-soundness gap also suffers a loss compared to ours that comes from the use of the ``rotated'' basis measurements.
\end{itemize}
To the best of our knowledge, the only aspect in which our proof of quantumness compares unfavourably with existing ones, e.g.\ \cite{MCVY22}, is that, based on current knowledge of constructions of NTCFs, our proof of quantumness requires a construction from LWE (in order to implement the ``encrypted CNOT'' procedure), whereas \cite{MCVY22} has the flexibility that it can be instantiated using any TCF, e.g.\ based on Diffie-Hellman or Rabin's function. While the latter are more efficient to implement, they also generally require a larger security parameter (inverting Rabin's function is as hard as factoring, whereas breaking LWE is as hard as worst-case lattice problems). Hence, it is currently still unclear which route is closer to a realisation at scale. In particular, we note that instantiating the construction of \cite{MCVY22} with the $x^2 \bmod N$ function by relying on a novel multiplication algorithm leads to a very efficient quantum circuit for the prover as shown in \cite{KY24}. It is plausible that future improvements to our proof of quantumness might yield constructions from a broader class of hardness assumptions than LWE (e.g.\ Ring-LWE), which would likely yield a further improvement in concrete efficiency\footnote{An upcoming work by one of the authors in fact improves and instantiates our template using Ring-LWE.}. %

Some existing proofs of quantumness are non-interactive (i.e.\ 1-round), with security in the random oracle model~\cite{BKVV20, alagic2020non, YZ22}. 
Likewise, our proof of quantumness can also be made non-interactive by using the Fiat-Shamir transformation~\cite{FS87} (where the computation of the hash function is classical, and does not increase the complexity of the actual quantum computation). The proof of quantumness in \cite{YZ22} has the additional desirable property of being publicly-verifiable, although it currently seems to be more demanding than the others in terms of quantum resources.

\subsection{Future Directions} 

\begin{itemize}
    \item \emph{Better contextuality compilers.} Focusing on compilers for contextuality, the following important aspect remains to be strengthened. The current compiler does not in general achieve \emph{quantum soundness} (in the sense that there are contextuality games in which a QPT prover can do much better than the ``compiled'' quantum value from \Thmref{infMain1}). Interestingly, recent works show that KLVY does satisfy quantum soundness {for certain families of non-local games}~\cite{NZ23,BGK+23,CMM+24}. 
    \item \emph{Oblivious Pauli pad.} We think that the new functionality of the oblivious Pauli pad (or some variant of it) has the potential to be useful elsewhere, and we leave this exploration to future work. A related question is to realise the oblivious Pauli pad in the plain model. We note that the \emph{oblivious pad} can be constructed from the hardness of factoring if one does not require the \newblue{``classical range sampling'' property} (see \Defref{Oblivious-U-pad}), %
    but we do not know whether one can have a plain model construction that satisfies \Defref{Oblivious-U-pad} %
    in full (we discuss this in more detail in \Secref{opad}). %
    \item \emph{More efficient ``encrypted CNOT''.} Turning to proofs of quantumness, the broad goal is of course to simplify these even further towards experimental implementations. One concrete direction in which our proof of quantumness could be simplified is the following. Currently, we use an NTCF that supports the ``encrypted CNOT'' operation from \cite{mahadev2020classical}. However, our only requirement is that the NTCF satisfies the potentially weaker property that it hides a bit in the xor of the first bit of a pair of preimages. This is because we only need to perform the ``encrypted CNOT'' operation once, and not repeatedly as part of a full-fledged homomorphic computation. It would be interesting to see if the weaker requirement could be achieved by simpler claw-free functions (from LWE or other assumptions, like DDH or factoring).
\item \emph{Testing other sources of quantumness.} Many other sources of quantumness have been identified in the literature, such as generalised contextuality (which allows arbitrary experimental procedures, not just projective measurements, and a broad class of ontological models, not just deterministic ones)~\cite{spekkens2005contextuality} and the Leggett-Garg experiment (which is the time analogue of Bell's experiment)~\cite{leggett1985quantum,budroni2013bounding}. These all suffer from the same limitation as contextuality, and our results raise the following question: can one construct analogous operational single-device tests for these sources of quantumness as well? 
    \item \emph{Testing indefinite causal order.} More ambitiously, one could try to separate quantum mechanics from more general theories such as those with indefinite causal order~\cite{chiribella2013quantum,oreshkov2012quantum} (i.e.\ theories that obey causality only locally and that may not admit a definite causal order globally). %
    For instance, a recent result gives evidence (in the black-box model) that indefinite causal order does not yield any relevant advantage over quantum mechanics~\cite{abbott2023quantum}. Perhaps one can obtain a clear separation under cryptographic assumptions?
\end{itemize}

\subsection*{Acknowledgements}

We are thankful to an anonymous QCrypt reviewer for their observations about the relation between the oblivious pad and non-interactive proofs of quantumness, and for pointing out the limitations of not formalising what it means for a test to be an operational test of contextuality. %

We thank Ulysse Chabaud, Mauro Morales and Thomas Vidick for their feedback on early drafts of this work. We thank Yusuf Alnawakhtha, Alexandru Gheorghiu, Manasi Mangesh Shingane, and Urmila Mahadev for various helpful discussions. %
ASA acknowledges support from the U.S. Department of Defense through a QuICS Hartree Fellowship, IQIM, an NSF Physics Frontier Center (GBMF-1250002) and MURI grant FA9550-18-1-0161. Part of the work was carried out while ASA was visiting the Simons Institute for the Theory of Computing.
AC acknowledges support from the National Science Foundation grant CCF-1813814, from the AFOSR under Award Number FA9550-20-1-0108 and from the Engineering and Physical Sciences Research Council through the Hub in Quantum Computing and Simulation grant (EP/T001062/1) and the Quantum Advantage Pathfinder (QAP) research programme (EP/X026167/1). KB is supported by A*STAR C230917003.

\section{Technical Overview\label{sec:tech-overview}}

In this overview, we start by briefly recalling non-local games, and introducing the notion of contextuality with some examples. 
We then build up towards our compiler for contextuality games, by first introducing the ideas behind the KLVY compiler, and then describing why new ideas are needed to achieve a compiler in the contextuality setting. We then describe our main novel technical tool, the oblivious pad, and how we use it to realise a compiler for contextuality games. We also discuss the main ideas in the proof. Finally, we describe how some of the new ideas can be streamlined to obtain a potentially simpler proof of quantumness. The latter can be understood without any reference to contextuality, and the interested reader may wish to skip directly to it (Section \ref{sec:tech-overview-poq}).

\subsection{Non-local Games\label{subsec:Non-local-Game}}

Let $A, B, X, Y$ be finite sets. A $2$-player non-local game is specified by a predicate $\pred:A\times B\times X\times Y\to\{0,1\}$ %
which indicates whether the players win or not, and a probability
distribution ${\cal D}_{{\rm questions}}$ over the
questions, which specifies $\Pr(x,y)$ for $(x,y)\in X\times Y$.
The game consists of a referee and two players, Alice and Bob, who
can agree on a strategy before the game starts, but cannote communicate once the game starts. The game proceeds as follows: the referee samples questions
$(x,y) \leftarrow {\cal D}_{{\rm questions}}$, sends $x$ to Alice and $y$ to Bob, and receives
their answers $a\in A$ and $b \in B$ respectively. Their success probability
can be written as
\begin{equation}
  \Pr({\rm win})=\sum_{x,y,a,b}\pred(a,b,x,y)\Pr(a,b|  x,y)\Pr(x,y)\label{eq:PrWin_nonlocal} \,,
\end{equation}
where the strategy used by Alice and Bob specifies $\Pr(a,b|x,y)$.

\paragraph*{Local Hidden Variable Strategy} 

The ``local hidden variable'' model allows Alice and
Bob to share a classical (random) variable $r$, and then have their answers be arbitrary, but fixed, functions of their respective questions,
i.e.
\[
  \Pr(a,b|x,y)=\sum_{r} P_{A}(a|x, r)P_{B}(b|y, r)P_{R}(r)
\]
where  $P_{A}$, $P_{B}$ and $P_{R}$ are probability
distributions that specify Alice and Bob's strategy. We refer to the optimal winning probability achievable by local hidden variable strategies as the \emph{classical value} of the game.

\paragraph*{Quantum Strategy}
A quantum strategy allows Alice and Bob to share a state $\left|\psi\right\rangle _{AB}$,
and use local measurements $M^{A}_x = \{M_{a|x}^{A}\}_{a}$ and $M_y^B = \{M_{b|y}^{B}\}_b$
to produce their answers (where the measurements can be taken to be projective without loss of generality), so that
\[
  \Pr(a,b|x,y)=\left\langle \psi\right|M_{a|x}^{A}\otimes M_{b|y}^{B}\left|\psi\right\rangle .
\]
We refer to the optimal winning probability achievable by quantum strategies as the \emph{quantum value} of the game.

It is well-known that there exist non-local games (like the CHSH game) where the quantum value exceeds the classical value. However, is it necessary to consider spatial separation (i.e.\ a tensor product structure) to observe such a ``quantum advantage''? A partial answer is no: one can consider contextuality.

\begin{figure}[H]
  \begin{centering}
    \includegraphics[scale=1.2]{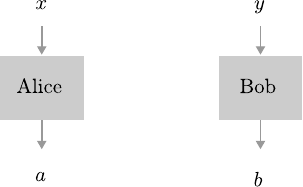}
    \par\end{centering}
  \caption{A two-player non-local game. Alice and Bob get $(x,y)$  from the referee and respond with $(a,b)$. They cannot communicate once the game starts.}

\end{figure}

\subsection{Contextuality}
\label{subsec:tech-overview-contextuality}

We start with a slightly informal definition of a contextuality game (see \Secref{Contextuality-Games} for a formal treatment). 
Let $Q$ and $A$ be finite sets. Let $\Call$ be a set of subsets of $Q$. We refer to the elements of $\Call$ as \emph{contexts}. Suppose for simplicity that all subsets $C \in \Call$ have the same size $k$. A \emph{contextuality game} is specified by a predicate $\pred:A^k \times \Call \to \{0,1\}$, and a probability distribution ${\cal D}_{{\rm contexts}}$ over contexts, which specifies $\Pr(C)$ for $C \in \Call$. The game involves a referee and a \emph{single} player, and proceeds as follows: 
\begin{itemize}
\item The referee samples $C = \{q_1, \ldots, q_k\} \in \Call$ according to ${\cal D}_{{\rm contexts}}$, and sends $C$ to the player.
\item The player responds with answers $\{a_1,\ldots, a_k\} \in A^k$. %
\end{itemize}
The success probability is
\begin{equation}
  \Pr({\rm win})= \sum_{a_1,\ldots,a_k,C} \pred(a_1,\ldots,a_k,C) \Pr(a_1,\ldots,a_k | C) \Pr(C) \label{eq:PrWin_contextual} \,,
\end{equation}
where the strategy used by the player specifies $\Pr(a_1,\ldots,a_k|C)$.

\paragraph*{Non-contextual strategy} This is the analogue of a ``local hidden variable'' strategy. A \emph{deterministic assignment} %
maps each question $q \in Q$ to a \emph{fixed} answer $a_q$. This represents the following strategy: upon receiving the context $C = \{q_1, \ldots, q_k\} \in \Call$, return $(a_{q_1}, \ldots, a_{q_k})$ as the answer. A strategy that can be expressed as a convex combination of deterministic assignments is referred to as \emph{non-contextual}, i.e.\ the answer to a question $q$ is independent of the context in which $q$ is being asked. %

\paragraph*{Quantum strategy} A quantum strategy is specified by a quantum state $\ket{\psi}$, and a collection of observables $\mathbf{O}=\{O_q\}_{q\in Q}$, such that, for any context $C =  \{q_1, \ldots, q_k\}\in \Call$, the observables $O_{q_1}, \ldots, O_{q_k}$ are compatible (i.e\ commuting). The strategy is the following: upon receiving context $C =\{q_1, \ldots, q_k\} \in \Call$, measure observables $O_{q_1}, \ldots, O_{q_k}$ on state $\ket{\psi}$, and return the respective outcomes $a_{1}, \ldots, a_{k}$.

\vspace{2mm}
Quantum mechanics is \emph{contextual} in the sense there are examples of games for which a quantum strategy can achieve a higher winning probability than the best non-contextual strategy. However, crucially, unlike a non-local game, a contextuality game does not directly yield an ``operational test'' of contextuality. The issue is that there is no clear way for the referee to enforce that the player's answer to question $q$ is consistent across the different contexts in which $q$ appears!

We informally describe three simple examples: the magic square game (Peres-Mermin)~\cite{peres1990incompatible,mermin1990simple,cabello2001bell, aravind2004quantum}, %
non-local games, and the KCBS experiment (the contextuality analogue of the Bell/CHSH experiment)~\cite{klyachko2008simple}. {We do this by directly specifying a quantum strategy first and then ``deriving'' the corresponding game (where the observables are just labels for the questions). In doing so, we abuse the notation slightly and identify questions with observables.}

\begin{example}[Peres-Mermin (Magic Square) \cite{peres1990incompatible,mermin1990simple}]
  \label{exa:MagicSquare} Consider the following set of observables
  \[
    \begin{array}{c}
      \\
      \mathbf{O}:= \\
      \\
      \\
      \\
    \end{array}\begin{array}{ccccc}
      \{X\otimes \mathbb{I},\  & \mathbb{I}\otimes Z, & X\otimes Z,    &  & \mathbb{I} \\
      \mathbb{I}\otimes X,     & Z\otimes \mathbb{I}, & Z\otimes X,    &  & \mathbb{I} \\
      X\otimes X,     & Z\otimes Z, & \ Y\otimes Y\} &  & \mathbb{I} \\
      \\
      \mathbb{I}                       & \mathbb{I}                   & -\mathbb{I}
    \end{array}
  \]
  (where $X,Y,Z$ are Pauli matrices). They satisfy the following properties: (a) they take $\pm1$ values (i.e.\ they have $\pm1$ eigenvalues), (b) operators along any row or column commute, and (c) the product of observables along any row or column equals $\mathbb{I}$, except along
  ${\rm col}_{3}$, where it equals $-\mathbb{I}$. 
  If we define the set of contexts by $C^{{\rm all}}:=\{{\rm row}_{1},{\rm row}_{2},{\rm row}_{3},{\rm col}_{1},{\rm col}_{2},{\rm col}_{3}\}$, it is not difficult to see that no deterministic assignment can be such that the condition on the products is satisfied
  along each row and column. For instance, the assignment
  \[
    \begin{array}{ccc}
      1 & -1 & -1 \\
      1 & -1 & -1 \\
      1 & 1  & ?
    \end{array}
  \]
  satisfies all constraints except one: if the question mark is $1$, then the condition
  along ${\rm col}_{3}$ fails, and if it is $-1$ then the condition
  along ${\rm row}_{3}$ fails. To satisfy both, somehow the value assigned
  to the last ``observable'' has to depend on the context, ${\rm row}_{3}$
  or ${\rm col}_{3}$, in which it appears---the assignment has to be ``contextual''. In quantum mechanics, all of the constraints can be satisfied using the observables described. One can in fact measure these observables on an arbitrary state $\left|\psi\right\rangle $ to win with probability $1$.
  One thus concludes that quantum mechanics is ``contextual'' in this sense.
\end{example}

2-player non-local games can be viewed as a special case of contextuality games as follows (and similarly for non-local games with more players).
\begin{example}[2-player non-local games]
Given a quantum strategy for a non-local game (as in \Subsecref{Non-local-Game}), we identify the measurements $M^A_x$ and $M^B_y$ with corresponding observables. Then, define $\mathbf{O}:=\{M_{x}^{A}\otimes I\}_{x}\cup\{I \otimes M_{y}^{B}\}_{y}$
and $\Call:=\{\{M_{x}^{A}\otimes I, I \otimes M_{y}^{B}\}\}_{x,y}$. %
It is not hard to see that the set of non-contextual strategies is the same as the set of local hidden variable strategies. %
\end{example}

Some contextuality games can yield a separation between non-contextual and quantum strategies with even smaller quantum systems than what is possible for non-local games. The following example yields a separation using just a qutrit---a single $3$-dimensional system (in contrast, non-local games require at least two qubits, i.e.\ dimension $4$, as in the CHSH game). For contextuality $3$ dimensions are necessary and sufficient~\cite{Koc}.\footnote{There are generalisations of contextuality~\cite{spekkens2005contextuality} that can give a separation with dimension $2$, but we do not consider these here.}
  \begin{example}[{KCBS \cite{klyachko2008simple}}]\label{exa:KCBSintro}
  Consider a $3$-dimensional vector space spanned by $\{\ket{0},\ket{1},\ket{2}\}$. Let $\ket{\psi}=\ket{0}$ and define five vectors $$\ket{v_q}:=\cos \theta \ket{0} + \sin \theta \sin \phi_{q} \ket{1} + \sin \theta \cos \phi_q \ket{2},$$ indexed by $q\in\{1, \dots, 5\}$ where $\phi_q=4\pi q/5$ and $\cos^2 \theta = \cos(\pi/5) / (1 + \cos (\pi/5))$. The heads of these vectors form a pentagon with $\ket{\psi}$ at the centre, and vectors indexed consecutively are orthogonal, i.e.\ $\left\langle v_{q}|v_{q+1}\right\rangle =0$ (where we take the indices to be periodic) as illustrated in \Figref{KCBS}.
  Define $\mathbf{O}:=\{\Pi_{q}\}_{q\in1\dots5}$ where $\Pi_{q}:=\left|v_{q}\right\rangle \left\langle v_{q}\right|$
  and $\Call:=\left\{ \{\Pi_{1},\Pi_{2}\},\{\Pi_{2},\Pi_{3}\}\dots\{\Pi_{5},\Pi_{1}\}\right\} $.
  The referee asks a context $C\leftarrow\Call$ uniformly at random
  from the set of all contexts and the player wins if the answer
  is either $(0,1)$ or $(1,0)$ (i.e.\ neighbouring assignments should be distinct).
  It is not hard to check that a non-contextual strategy
 wins at most with probability $4/5=0.8$, while the quantum strategy described above wins with probability $\frac{2}{\sqrt{5}}\approx 0.8944$.
\end{example}
\begin{figure}
  \begin{centering}
    \includegraphics[scale=1.25]{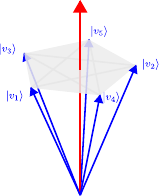}
    \par\end{centering}
  \caption{\label{fig:KCBS} An illustration of the optimal quantum strategy corresponding to the KCBS Game defined in \Exaref{KCBSintro}. Here the red vector denotes the quantum state $\ket{\psi}$ and the blue ones $\ket{v_q}$ correspond to projective measurements, $\Pi_q = \ket{v_q}\bra{v_q}$. Consecutively indexed blue vectors, i.e.\ $\ket{v_q},\ket{v_{q+1}}$, are orthogonal (indexing is periodic). 
  }

\end{figure}

One route towards obtaining an operational test of contextuality is to find a way to enforce that measurements on a single system happen ``sequentially'', i.e.\ they are separated in
``time'' (as opposed to being separated in ``space'', which is the case for non-local games). We describe one folklore attempt at constructing an operational test, which assumes that the device is ``memoryless''. This example is not essential to understanding our results, and may be skipped at first read.

\begin{center}
    \begin{figure}
      \hfill{}\subfloat[\label{fig:memoryless}]{\centering{}\includegraphics[scale=1.0]{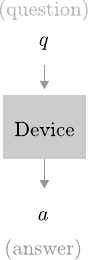}}\hfill{}\subfloat[\label{fig:MagicSquareMemoryless}]{\centering{}\includegraphics[scale=0.8]{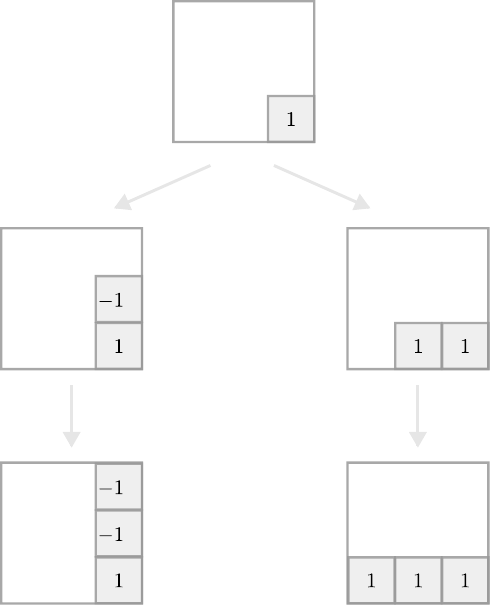}}\hfill{}
      \caption{The folklore memoryless interpretation of contextuality: a memoryless device (left) and a memoryless operational test, illustrated using the Peres-Mermin magic square (right).}
    \end{figure}
  \end{center}

\paragraph*{The ``memoryless'' attempt at an operational test}
As denoted in \Figref{memoryless}, consider a device that takes as
input a question $q$, produces an answer $a$, and then forgets the
question. The referee in this case proceeds as follows:
\begin{enumerate}
  \item Samples a context $C$ from $\Call$ with probability $\Pr(C)$. %
  \item Sequentially asks all the questions $q_1\dots q_k$ in the context $C$.
  \item Accepts the answers $a_1\dots a_k$ if the constraint corresponding to $C$ holds, i.e.\ if $\pred(a_1\dots a_k,C)$ is true.
\end{enumerate}
Note that the most general deterministic model for the device is one
that encodes a ``truth table'' $\tau:Q\to A$ that maps questions
to answers. Since there is no memory, no previous question can affect
the way the device answers the current question. The most general
quantum device, on the other hand, starts with an initial $\left|\psi\right\rangle $
and to each question, assigns an observable $O_{q}$, which is measured
to obtain an answer $a$ to the question $q$.  %

Let us work out an example to illustrate
the key point. Consider again the Magic Square game from \Exaref{MagicSquare}.
Suppose the first question the referee asks is the value of the bottom
right box of the magic square (see \Figref{MagicSquareMemoryless}).
It can then either ask questions completing the corresponding row
or column. For a deterministic \emph{memoryless} device, since a single
truth table $\tau$ is being used, it is impossible to satisfy all
the constraints of the magic square. However, a deterministic
device \emph{with memory} can satisfy all the constraints. This is
because before answering the last question, it can learn whether the
third column is being asked or the third row and thus, it can answer the
last question to satisfy the constraint. Thus ``classical memory'',
allows for classical simulation of contextuality in this test, without
using any quantum effects. In fact, \cite{kleinmann2011memory,cabello2018optimal} even quantifies the amount
of classical memory needed to simulate contextuality. %

Evidently, the glaring limitation of this attempt is that there is no operational
way of ensuring that a device is ``memoryless''. %
Thus, this has remained a barrier despite the numerous attempts
\cite{lapkiewicz2011experimental,um2013experimental,jerger2016contextuality,zhan2017experimental,malinowski2018probing,leupold2018sustained,zhang2019experimental,um2020randomness,wang2022significant,hu2023self,liu2023experimental,bharti2019robust,bharti2019local,xu2024certifying,saha2020sum,budroni2022kochen} %
 and the fact that contextuality has, in general, been a bustling area of investigation \cite{budroni2022kochen,wang2022significant,hu2023self,liu2023experimental,liu2023exploring,xu2024certifying}.

Our construction follows the same general approach of enforcing separation in ``time''
instead of in ``space'', but is a radical departure from the attempt
described above. \newblue{In the subsequent sections, we only assume that the device is \emph{computationally bounded}, and use cryptographic techniques to construct an ``operational test of contextuality''.}
\begin{rem}[Criteria for being an Operational Test of Contextuality]\label{rem:criteria}
We have not yet formally defined what we mean by an ``operational test of contextuality''. %
Consider any test that serves as a proof of quantumness. %
Intuitively, we require that this test, additionally, is \emph{faithtful to some contextuality game $\G$}, i.e.\ it satisfies the following two properties: 
\begin{itemize}
  \item The test involves asking the prover questions corresponding to $\G$, possibly under some ``encoding''. The test is, however, allowed to involve other messages unrelated to $\G$. 
  \item Whether the test passes or fails is determined solely by evaluating the predicate corresponding to $\G$, using the ``decoded'' questions and answers. %
\end{itemize}
Formalising these requirements is a bit more involved and is deferred to \Subsecref{criteria}. However, this intuitive notion suffices for now, as will be clear from our construction in the subsequent sections. %

\end{rem}
\global\long\def\map{\mathsf{map}}%
\global\long\def\aux{\mathsf{aux}}%

\global\long\def\NC{\mathsf{NC}}%
\global\long\def\qu{\mathsf{qu}}%
\global\long\def\complete{\mathsf{complete}}%
\global\long\def\transcript{\mathsf{transcript}}%
\global\long\def\chall{\mathsf{chall}}%
\global\long\def\poq{\mathsf{\pi}}%
\global\long\def\PoQ{\mathsf{PoQ}}%
\global\long\def\valid{\mathsf{valid}}%
\global\long\def\cert{\mathsf{Cert}}%
\global\long\def\challenge{\mathsf{chall}}%

\subsection{Quantum Fully Homomorphic Encryption (QFHE)}\label{subsec:techoverviewQFHE}

We informally introduce fully homomorphic encryption. We start with the classical notion.

\paragraph*{Fully Homomorphic Encryption (FHE)} A homomorphic encryption scheme is specified
by four algorithms, $(\gen,\enc,\dec,\Eval)$, as follows:
\begin{itemize}
\item $\gen$ takes as input a security parameter $1^{\lambda}$, and outputs a secret key $\sk$.
  \item $\enc$ takes as input a secret key $\sk$ and a message $s$, and outputs a ciphertext $c$. We use the notation $\enc_{\sk} = \enc(\sk, \cdot)$.
  \item $\dec$ takes as input a secret key $sk$, a ciphertext $c$ and outputs the corresponding plaintext $s$. We use the notation $\dec_{\sk} = \dec(\sk, \cdot)$.
\end{itemize}
The ``homomorphic'' property says that a circuit $\circuit$ can be applied on an encrypted input to obtain an encrypted output, i.e.\
\begin{itemize}
  \item $\Eval$ takes as input a circuit $\circuit$, a ciphertext $c$, and an auxiliary input $\mathsf{aux}$, and outputs a ciphertext $c'$. The following is satisfied. Let $c \leftarrow \mathsf{Enc}_{\sk}(s)$ and $c' \leftarrow\Eval(\circuit, c, \mathsf{aux})$, then $\dec_{\sk}(c')=\circuit(s, \mathsf{aux})$.
\end{itemize}
Crucially, note that the secret key $\sk$ is not needed to apply
the $\Eval$ algorithm. The security condition is the usual one: that
an encryption of $s$ should be indistinguishable from an encryption of $s'\neq s$. If $\Eval$ supports evaluation for the class of all polynomial-size circuits (in the security parameter), then the scheme is said to be \emph{fully} homomorphic.\footnote{While homomorphic schemes for restricted families of circuits have been known for some time, a
``fully homomorphic scheme'' was only discovered somewhat recently in \cite{Gentry09}.} 

\paragraph*{Quantum Fully Homomorphic Encryption (QFHE)}

For this overview, it suffices to take a QFHE scheme to be the same as an FHE scheme, except that it allows messages and auxiliary inputs to be quantum states, and circuits to be quantum circuits. The QFHE schemes that are relevant to our work \cite{mahadev2020classical, brakerski18} satisfy the following
additional properties.
\begin{enumerate}
  \item \emph{Classical encryption/decryption}. \\
        $\gen$ is a classical algorithm, while $\Enc,\Dec$ become classical algorithms when their inputs are classical. %
        In particular, this means
        that classical inputs are encrypted into classical ciphertexts. This
        property is essential when the scheme is deployed in protocols involving
        classical parties, as will be the case here.\footnote{The first schemes to satisfy this property appeared in
        the breakthrough works \cite{mahadev2020classical,brakerski18}.}
  \item \emph{Locality is preserved}. \\
        {Consider an arbitrary bipartite state $\ket{\psi}_{AB}$ and let $M^A$ and $M^B$ be circuits acting on registers $A$ and $B$ respectively with a measurement at the end. The property requires that correlations between the measurement outcomes from $M^A$ and $M^B$ should be the same in the following two cases:\\
        (i) Register $A$ is encrypted, $M_A$ is applied using $\Eval$, and the result is decrypted. Let $a$ be the decrypted outcome. $M^B$ is applied to register $B$. Let $b$ be the outcome.\\
        (ii) Apply $M^A$ on register $A$ and $M^B$ on register $B$. Let $a$ and $b$ respectively be the outcomes.}
  \item\label{item:3QFHEproperty} \emph{Form of Encryption}.\\
        Encryption of a state $\left|\psi\right\rangle $
        takes the form $(X^{x}Z^{z}\left|\psi\right\rangle ,\widehat{xz})\leftarrow\Enc_{\sk}(\left|\psi\right\rangle )$,
        where $X^x$ applies a Pauli $X$ to the $i$-th qubit based on the value of $x_i$, and $Z^z$ is defined similarly. $\widehat{xz}$ is a classical $\FHE$ encryption of the pad $xz$.
\end{enumerate}
All known constructions of QFHE schemes satisfying property 1 also
satisfy properties 2 and 3. The KLVY compiler, as one might guess, relies
on property 2. Our compiler relies on property 3.\footnote{
Strictly speaking, we require the $\FHE$ scheme in property 3 to be public-key. This is used to establish the faithfulness condition (detailed in \Defref{OperationalTestOfContextuality} and \Subsecref{FHE-Scheme}). We neglect this here for readability.
}

\subsection{The KLVY Compiler}

Consider a two-player non-local game specified by question and answer sets $X,Y,A,B$, a predicate
$\pred:A\times B\times X\times Y\to\{0,1\}$ and a distribution over the
questions ${\cal D}_{{\rm questions}}$ that specifies $\Pr(x,y)$.
The KLVY compiler takes this non-local game as input and produces the following single-prover game, where the verifier proceeds as follows:
\begin{itemize}
    \item {\textbf{Round 1.}} Sample $(x,y) \leftarrow {\cal D}_{{\rm questions}}$, and a secret key $\sk$ for a $\QFHE$ scheme. Send an encryption $c_x$ of ``Alice's question'', and get an encrypted answer $c_a$.
    \item {\textbf{Round 2.}} Send ``Bob's question'' $y$, and get his answer $b$ in the clear. Decrypt Alice's response using $\sk$, and accept if $\pred(a,b,x,y) = 1$.
\end{itemize}
The honest prover prepares the entangled state $\ket{\psi}_{AB}$ corresponding to the optimal quantum strategy for the non-local game. It uses $\QFHE$'s $\Eval$ algorithm on subsystem $A$ to answer question $x$ according to Alice's optimal strategy, and answers question $y$ in the clear using Bob's optimal strategy on subsystem $B$. More formally, they proceed as in \Figref{KLVY}. %
\begin{figure*}

\begin{center}
  \begin{tabular}{|>{\raggedright}p{4cm}|>{\centering}p{3cm}|>{\raggedright}m{5cm}|}
    \multicolumn{1}{>{\raggedright}p{4cm}}{(Honest Quantum) Prover} & \multicolumn{1}{>{\centering}p{3cm}}{} & \multicolumn{1}{>{\raggedright}m{5cm}}{Verifier}\tabularnewline
    \cline{1-1} \cline{3-3}
    Prepare $\left|\psi\right\rangle _{AB}$ corresponding to the non-local
    game's optimal quantum strategy.                                        &                                        & ~\\$\sk\leftarrow\gen(1^{\lambda})$\\
    $(x,y)\leftarrow{\cal D}_{{\rm questions}}$
    $c_{x}\leftarrow\enc_{\sk}(x)$                                  \tabularnewline
                                                                &                                        & \tabularnewline
                                                       & {\Large $\overset{c_{x}}{\longleftarrow}$}     & \tabularnewline
 Homomorphically compute the answer to $x$ using ``Alice's strategy'' on system $A$. Let $c_{a}$ be the encrypted answer. &  & \tabularnewline
    & {\Large $\overset{c_{a}}{\longrightarrow}$}                                  & \tabularnewline
                                                                    &                                        & \tabularnewline
                                  & {\Large $\overset{y}{\longleftarrow}$} & \tabularnewline

Compute the answer to $y$ using ``Bob's strategy'' (in the clear) on system $B$. Let the answer be $b$. &        &  \tabularnewline
   & {\Large $\overset{b}{\longrightarrow}$}                                      & \tabularnewline
                                                                    &                                        & $a=\dec_{\sk}(c_{a})$

    Accept if $\pred(a,b,x,y)=1$.\\~\tabularnewline
    \cline{1-1} \cline{3-3}
  \end{tabular}
  \par
\end{center}
\caption{The KLVY compiler~\cite{KLVY22} is illustrated above: it takes any two-player non-local game and converts it into a proof of quantumness. \label{fig:KLVY}}
\end{figure*}

\begin{thm*}[\cite{KLVY22}, informal]
Consider a two-player non-local game with classical and quantum values $\omega_c$ and $\omega_q$ respectively. The corresponding KLVY-compiled single-prover game satisfies the following:
\begin{itemize}
  \item (Completeness) The QPT prover described above makes the verifier accept with probability $\omega_q-\negl(\lambda)$.
  
  \item (Soundness) Every PPT prover makes the verifier above accept with probability
  at most $\omega_c+\negl(\lambda)$.
\end{itemize}
Here $\negl$ are (possibly different) negligible functions.
\end{thm*}
For completeness, we briefly outline the general compiler for non-local games with $k>2$ players. To
compile games with $k$ players, KLVY generalises the procedure above
as follows: the compiled game consists of $k$ rounds--each
round consisting of a question and an answer (so $2k$ messages in total). The first $k-1$ questions
are encrypted (using $k-1$ independent random $\QFHE$ keys) and the last
question is asked in the clear.

Let us take a moment and build some intuition about the KLVY compiler.
{At first, one might think that it would be even more secure to also encrypt $y$. However, as we remarked in the introduction, this is not the case: since the prover can compute any desired answer $b$ using $x,a$ and $y$ under the hood of the QFHE, it can ensure that $b$ is such that $\pred(a,b,x,y)=1$
(for any non-trivial choice of $\pred$).} Instead, the key observation of KLVY
is that the verifier is testing the correlation between encrypted
answers and answers in the clear, and it turns out that a quantum prover can produce stronger correlations
than any PPT prover can. How is this intuition formalised? The key idea is that, since a PPT prover is classical, one can rewind the PPT prover to obtain answers to \emph{all} possible second round questions. This is equivalent to obtaining Bob's entire assignment of answers to questions (corresponding to a fixed encrypted question and answer for Alice from the first round). Suppose for a contradiction that such a PPT prover wins with probability non-negligibly greater than the classical value. Then, it must be that Bob's entire assignment is non-trivially correlated to ``Alice's encrypted question''. This information can thus be used to obtain a non-negligible advantage in guessing the encrypted question, breaking security of the QFHE scheme. %

Now that we understand the KLVY construction and the key insight behind
their proof, we will discuss  barriers to extending these ideas to contextuality, and our approach to circumvent these.

\subsection{Contribution 1 | A Computational Test of Contextuality\label{subsec:Contributions}} %

For simplicity, let us first restrict to contextuality games with contexts of size $2$.

\paragraph*{Attempts at extending KLVY to contextuality.}

Let us consider compilers for contextuality games where the verifier proceeds in an analogous way as the verifier in the KLVY compiler:
\begin{itemize}
    \item {\textbf{Round 1.}} Sample a context $C=(q_1,q_2) \leftarrow {\cal D}_{{\rm contexts}}$. Send a $\QFHE$ encryption $c_{q_1}$ of $q_1$ and receive as a response an encryption $c_{a_1}$ of $a_1$.
    \item {\textbf{Round 2.}} \emph{There is no clear way to proceed. Three natural approaches are listed below. }
\end{itemize}
The honest prover's state after Round 1 is encrypted and has the form $(X^x Z^z \ket{\psi_{(q_1,a_1)}}, \widehat{xz})$ where $\widehat{xz}$ denotes a classical encryption of the strings $xz$ (using Property 3 of the $\QFHE$ scheme), as detailed below \atul{in Figure~\ref{figure:klvy_context}.}
\begin{figure*}

  \begin{center}
  \begin{tabular}{|>{\raggedright}p{4cm}|>{\centering}p{3cm}|>{\raggedright}m{5cm}|}
    \multicolumn{1}{>{\raggedright}p{4cm}}{(Honest Quantum) Prover}                                                          & \multicolumn{1}{>{\centering}p{3cm}}{} & \multicolumn{1}{>{\raggedright}m{5cm}}{Verifier}\tabularnewline
    \cline{1-1} \cline{3-3}
    & & \tabularnewline
    Prepare $\left|\psi\right\rangle $ corresponding to the quantum strategy.                                                &                                        & $\sk\leftarrow\gen(1^{\lambda})$

    $C\leftarrow {\cal D}_{{\rm contexts}}$, where $C=(q_{1},q_{2})$

    $c_{q_{1}}\leftarrow\enc_{\sk}(q_{1})$\tabularnewline
                                                                                                                             &                                        & \tabularnewline
     & {\Large $\overset{c_{q_{1}}}{\longleftarrow}$} & \tabularnewline
     Evaluate the answer to $q_{1}$ under QFHE to get an encrypted answer
    $c_{a_{1}}$. The post-measurement state is now $\left(X^{x}Z^{z}\left|\psi_{q_{1}a_{1}}\right\rangle ,\widehat{xz}\right).$ &   & \tabularnewline

    & {\Large $\overset{c_{a_{1}}}{\longrightarrow}$}                                                                                       & \tabularnewline
                                                                                                                             &                                        & \tabularnewline
                                                                                                                             &  {\Large $\overset{?}{\longleftarrow}$}

    {\Large $\longrightarrow$}
    & \tabularnewline
    \cline{1-1} \cline{3-3}
  \end{tabular}
  \par\end{center}
\caption{The figure illustrates that, a priori, it is unclear how to generalise KLVY to compile contextuality games.}
\label{figure:klvy_context}
\end{figure*}

Here are three natural approaches for how to proceed, and how they fail.
\begin{enumerate}
  \item \emph{Proceed just as KLVY: Ask question $q_{2}$ in the clear.}
        This does not work because, unlike in the original KLVY setup,
        there is no analogue of system $B$ which is left in the clear. Here the prover holds an encrypted state so it is unclear how $q_{2}$ can be answered with any non-trivial dependence on the state under the encryption. %
  \item \emph{Ask the second question also under encryption.}
        If the same key is used for encryption, then, as we argued for KLVY, the predicate of the game can be satisfied by computing everything homomorphically. If the keys are independent, then we essentially return to the problem in item 1.
  \item \emph{Reveal the value of the (classically) encrypted pad $\widehat{xz}$ and ask question
        $q_{2}$ in the clear.}
        This has a serious issue: the prover
        can simply ask for the encrypted pad corresponding to $c_{q_{1}}$
        and thereby learn $q_{1}$ (or at least some bits of $q_{1}$). Once $q_{1}$ is learned, again, the predicate of the game can be trivially satisfied.\label{item3}
\end{enumerate}

\paragraph*{The Oblivious Pauli Pad}

As mentioned in Section \ref{subsec:intro}, our compiler relies on
a new cryptographic primitive, that we introduce to circumvent the barriers described above. Here, it
helps to be a bit more formal. Let $\mathbf{U}:=\{U_{k}\}_{k\in K}$ be a group of unitaries acting
on the Hilbert space ${\cal H}$. We take this group to
be the set of Paulis $\{X^{x}Z^{z}\}_{xz}$, as this makes the primitive compatible with the the form of the QFHE scheme that we will employ later in our compiler. Nonetheless, we use the general notation
$\{U_{k}\}_{k\in K}$, as it simplifies the presentation. We define the \emph{oblivious} $\mathbf{U}$ \emph{pad} as follows.

The \emph{oblivious} $\mathbf{U}$ \emph{pad} is a tuple of algorithms $(\gen,\Enc,\Dec)$ where $\gen$ and $\dec$ are PPT.\footnote{The formal definition of the oblivious pad involves a fourth PPT algorithm $\samp$ which for clarity is deferred to \Secref{opad}.}
Let $(\pk,\sk)\leftarrow\gen(1^{\lambda})$ be the public and the secret keys
generated by $\gen$. Encryption takes the form $(U_{k}\left|\psi\right\rangle ,s)\leftarrow\Enc_{\pk}(\left|\psi\right\rangle )$, where $k=\dec_{\sk}(s)$. 
Notice the similarity with the post-measurement state in the discussion
above (we will return to this in a moment). The security requirement is
that no PPT algorithm can win the following security game with probability
non-negligibly greater than $1/2$. We depict the oblivious pad primitive in Figure~\ref{fig:opad}.

\begin{figure*}
\begin{center}
  \begin{tabular}{|>{\raggedright}p{4cm}|>{\centering}p{3cm}|>{\raggedright}m{5cm}|}
    \multicolumn{1}{>{\raggedright}p{4cm}}{Prover} & \multicolumn{1}{>{\centering}p{3cm}}{} & \multicolumn{1}{>{\raggedright}m{5cm}}{Verifier}\tabularnewline
    \cline{1-1} \cline{3-3}
                                                   &                                        & ~\\$(\pk,\sk)\leftarrow\gen(1^{\lambda})$\tabularnewline
                                                   & {\Large $\overset{\pk}{\longleftarrow}$}

    {\Large $\overset{s}{\longrightarrow}$}                     & \tabularnewline
                                                   &                                        & \tabularnewline
                                                   &         & $b\leftarrow\{0,1\}$

    $k_{0}\leftarrow K$

    $k_{1}=\Dec_{\sk}(s)$\tabularnewline
    & {\Large $\overset{k_{b}}{\longleftarrow}$}  & \tabularnewline
                                                   & {\Large $\overset{b'}{\longrightarrow}$}            & \tabularnewline
    &  & Accept if $b=b'$ \\~ \tabularnewline
    \cline{1-1} \cline{3-3}
  \end{tabular}
  \par
\end{center}
\caption{The security game for the oblivious pad.}
\label{fig:opad}
\end{figure*}

The security game formalises the intuition that no PPT prover can
distinguish between the correct ``key'' $k_{1}$ and a uniformly random
``key'' $k_{0}$. We emphasize, in words, the two distinctive features of this primitive:
\begin{itemize}
    \item By running $\enc$, a QPT prover can obtain, given a state $\ket{\psi}$, an encryption of the form $(U_{k}\left|\psi\right\rangle ,s)$, where $k = \dec_{\sk}(s)$.
    \item There is no way for a PPT prover, given $\pk$, to produce an ``encryption'' $s$, for which it has non-negligible advantage at guessing $\dec_{\sk}(s)$.
\end{itemize}

We describe informally how to instantiate the primitive in the random oracle model
assuming noisy trapdoor claw-free functions (the detailed description is in \Algref{ObliviousPauliPad}). %
The construction builds on ideas from \cite{BKVV20}.

The key idea is the following. Let $f_0,f_1$ be a Trapdoor Claw-Free function pair. We take $\pk = (f_0, f_1)$, and $\sk$ to be the corresponding trapdoor. Then, $\enc_{\pk}$ is as follows:
\begin{itemize}
\item[(i)] On input a qubit state $\ket{\psi} = \alpha \ket{0} + \beta \ket{1}$, evaluate $f_0$ and $f_1$ in superposition, controlled on the first qubit, and measure the output register. This results in some outcome $y$, and the leftover state $(\alpha \ket{0} \ket{x_0} + \beta \ket{1}\ket{x_1})$, where $f(x_0) = f(x_1) = y$.
\item[(ii)] Compute the random oracle ``in the phase'', to obtain $((-1)^{H(x_0)}\alpha \ket{0} \ket{x_0} + (-1)^{H(x_1)} \beta \ket{1}\ket{x_1})$. Measure the second register in the Hadamard basis. This results in a string $d$, and the leftover qubit state
$$ \ket{\psi_Z} = Z^{d\cdot(x_0 \oplus x_1) + H(x_0) + H(x_1)} \ket{\psi} \,.$$
\item[(iii)] Repeat steps (i) and (ii) on $\ket{\psi_{Z}}$, but \emph{in the Hadamard basis}! This results in strings $y'$ and $d'$, as well as a leftover qubit state $\ket{\psi_{XZ}} = $
$$X^{d'\cdot(x_0' \oplus x_1') + H(x_0') + H(x_1')}Z^{d\cdot(x_0 \oplus x_1) + H(x_0) + H(x_1)} \ket{\psi}\,,$$
where $x_0'$ and $x_1'$ are the pre-images of $y'$.
\end{itemize}
Notice that the leftover qubit state $\ket{\psi_{XZ}}$ is of the form $X^x Z^z \ket{\psi}$ where $x,z$ have the following two properties: (a) a verifier in possession of the TCF trapdoor can learn $z$ and $x$ given respectively $y, d$ and $y',d'$, and (b) no PPT prover can produce strings $y,d$ as well as predict the corresponding bit $z$ with non-negligible advantage (and similarly for $x$). Intuitively, this holds because a PPT prover that can predict $z$ with non-negligible advantage must be querying the random oracle at \emph{both} $x_0$ and $x_1$ with non-negligible probability. By simulating the random oracle (by lazy sampling, for instance), one can thus extract a claw $x_0, x_1$ with non-negligible probability, breaking the claw-free property.

\paragraph*{Our Compiler}
We finally describe our contextuality game compiler. As mentioned in the introduction, our strategy is still to ask the first question under $\QFHE$ encryption and the second question in the clear, with the following crucial difference: the prover is first asked to ``re-encrypt'' the post-measurement state using the \emph{oblivious pad} functionality (from here one referred to as $\opad$), and only \emph{after that} the verifier reveals to the prover how to ``decrypt'' the state, in order to proceed to round 2. %

The key idea is easy to state, once the notation is clear. To this end, recall that Property~\ref{item:3QFHEproperty} of a $\QFHE$ scheme ensures that the encryption of a quantum state $\ket{\psi''}$ has the form 
\begin{equation}
    (U_{k''} \ket{\psi''}, \hat{k}'') \label{eq:Property3Form}
\end{equation} 
where $\hat{k}''$ denotes a classical encryption of $k''$ (the reason why we use double primes will become clear shortly), and $U_{k''}$ is an element of the Pauli group. Note that using the secret key of the $\QFHE$ scheme, one can recover $k''$ from $\hat{k}''$. Further, let the optimal quantum strategy for the underlying contextuality game consist of state $\ket{\psi}$ and observables $\{O_q\}$. Finally, denote by $\ket{\psi_{a,q}}$ the post-measurement state arising from measuring $\ket{\psi}$ using $O_q$ and obtaining outcome $a$.

We are now ready to describe our compiler. We first explain it in words and subsequently give a more formal description for clarity. In both cases, we highlight the conceptually new parts in {\color{blue} blue}.%
\begin{itemize}
    \item \textbf{Round 1} \emph{Ask the encrypted question, and have the prover re-encrypt its post-measurement state using $\opad$.} \\
    The verifier samples keys for the $\QFHE$ scheme {\color{blue}and for the $\opad$.} It samples a context $C  \leftarrow \cal{D}_{\rm contexts} $ and then uniformly samples questions $q_1,q_2$ from the context $C$ (note that $q_1=q_2$ with probability $1/2$ since, for simplicity, we are considering contexts of size $2$.)
        \begin{itemize}
            \item \textbf{Message 1} The verifier sends the $\QFHE$ encryption $c_{q_1}$ of the first question $q_1$, together with {\color{blue} the public key of the $\opad$.}
        \end{itemize}
    The honest quantum prover obtains the encrypted answer $c_{a_1}$ by measuring, under the $\QFHE$ encryption, the state $\ket{\psi}$ using the observable $O_{q_1}$. %
    It now holds a state of the form in \Eqref{Property3Form}, with $\ket{\psi''}=\ket{\psi_{a_1,q_1}}$. {\color{blue}Using the public key of the $\opad$, the prover applies $\opad.\Enc$ to the encrypted post-measurement state, $U_{k''}\ket{\psi_{a_1,q_1}}$, to obtain the ``re-encrypted'' quantum state, $U_{k'}U_{k''} \ket{\psi''}$, where $U_{k'}$ was applied by the $\opad$, together with a classical string $s'$ that encodes $k'$.} This step is critical to the security of the protocol and is discussed in more detail shortly. Crucially, note that both $U_{k'}$ (coming from the $\opad$) and $U_{k''}$ (coming from the $\QFHE$) are Paulis.
        \begin{itemize}
            \item \textbf{Message 2} The prover sends the $\QFHE$ encrypted answer $c_{a_1}$, together with the two strings $\hat k''$ and $s'$.
        \end{itemize}
    \item \textbf{Round 2} \emph{Remove the overall encryption, and proceed in the clear.} {\color{blue} The verifier recovers $k'$ from $s'$ (using the secret key of the $\opad$) and $k''$ from $\hat{k}''$ (using the secret key of the $\QFHE$ scheme). It then computes $k$ satisfying $U_k=U_{k'}U_{k''}$.}
        \begin{itemize}
            \item \textbf{Message 3} The verifier sends the second question $q_2$ {\color{blue} together with $k$ as computed above.}
        \end{itemize}
    The prover measures its {\color{blue} quantum state $U_{k'}U_{k''}\ket{\psi''} = U_{k}\ket{\psi_{a_1,q_1}}$, using observable $U_k O_{q_2} U_k^\dagger$} to obtain an outcome $a_2$.
        \begin{itemize}
            \item \textbf{Message 4} The prover sends $a_2$. 
        \end{itemize}
    The verifier decrypts $c_{a_1}$ to recover $a_1$ using the secret key of the $\QFHE$ scheme. {\color{blue} If $q_1=q_2$, it accepts if the answers match, i.e.\ $a_1=a_2$.} If $q_1\neq q_2$, it accepts if the predicate is true, i.e.\ $\pred(a_1,a_2,q_1,q_2)=1$.
\end{itemize}

We summarise our compiler in Figure~\ref{fig:context_compiler}. Since we now have ${\gen,\enc,\dec}$ algorithms for both $\opad$ and $\QFHE$, we use prefixes such as $\opad.\enc$ to refer to $\enc$ associated with $\opad$ to avoid confusion.

\begin{figure*}
\begin{center}
  \begin{tabular}{|>{\raggedright}p{6cm}|>{\centering}p{3cm}|>{\raggedright}m{6cm}|}
    \multicolumn{1}{>{\raggedright}p{6cm}}{(Honest Quantum) Prover}                           & \multicolumn{1}{>{\centering}p{3cm}}{}                              & \multicolumn{1}{>{\raggedright}m{6cm}}{Verifier}\tabularnewline
    \cline{1-1} \cline{3-3}
    Prepare $\left|\psi\right\rangle $ corresponding to the quantum strategy.                 &                                                                     & ~\\ $\sk\leftarrow\QFHE.\gen(1^{\lambda})$

    $C\leftarrow \cal{D}_{\rm contexts}$%

    \textcolor{blue}{$q_{1}\leftarrow C$}

    \textcolor{blue}{$q_{2}\leftarrow C$}

    $c_{q_{1}}\leftarrow\QFHE.\enc_{\sk}(q_{1})$

    \textcolor{blue}{$(\opad.\sk,\opad.\pk)\leftarrow\opad.\gen(1^{\lambda})$}\tabularnewline
                                                                                              &                                                                     & \tabularnewline
     & {\Large \textcolor{blue}{$\xleftarrow{(c_{q_{1}}, \opad.\pk)}$}}     & \tabularnewline
     Evaluate the answer to $q_{1}$ under QFHE to get an encrypted answer
    $c_{a_{1}}$. The post-measurement state is now $\left(U_{k''}\left|\psi_{q_{1}a_{1}}\right\rangle ,\widehat{k}''\right)$
    \textcolor{blue}{and now apply the $\opad$ to the post-measurement state,
    i.e.\ $(U_{k'}U_{k''}\left|\psi_{q_{1}a_{1}}\right\rangle ,s')\leftarrow$}

    \textcolor{blue}{$\opad.\enc_{\opad.\pk}(U_{k''}\left|\psi_{q_{1}a_{1}}\right\rangle )$.} & & \tabularnewline
                                                                                              & {\Large \textcolor{blue}{$\xrightarrow{(c_{a_{1}}, \hat{k}'', s')}$}} & \tabularnewline
            &   & \textcolor{blue}{Using the secret keys, find $k$ such that $U_{k}=U_{k'}U_{k''}$}   \tabularnewline
                                                                                              & {\Large \textcolor{blue}{$\xleftarrow{(q_{2},k)}$}}                 & \tabularnewline
                                    \textcolor{blue}{Measure $U_{k}O_{q_{2}}U_{k}^{\dagger}$} & & \tabularnewline
                                       & {\Large \textcolor{blue}{$\overset{a_{2}}{\longrightarrow}$}}                    & \tabularnewline
                                                                                              &                                                                     & {Decrypt $c_{a_{1}}$ to learn $a_{1}$.}

    \textcolor{blue}{If $q_{1}=q_{2}$, accept if $a_{1}=a_{2}$.}

    {If $q_{1}\neq q_{2}$, accept if $\pred(a_{1},a_{2},q_{1},q_{2})=1$.}\tabularnewline
    & & \tabularnewline
    \cline{1-1} \cline{3-3}
  \end{tabular}
  \par\end{center}
  \caption{Our contextuality game compiler. It takes a contextuality game given by $(Q,A,\Call,\pred,\calD_{\rm contexts})$, a $\QFHE$ scheme and an $\OPad$ and returns an operational test of contextuality.} %
  \label{fig:context_compiler}
\end{figure*}

Our compiler satisfies the following, assuming the underlying $\QFHE$ and oblivious pad are secure. We first state a special case of our general result (which is stated later in Theorem \ref{thm:general}).
\begin{thm*}[restatement of \Thmref{infMain1}]
Consider a contextuality game \atul{$\G$} with contexts of size $2$. Let $\valNC$ and $\valQu$ be its non-contextual and quantum values respectively. The compiled game (as described above) is faithful to $\G$. In particular, it satisfies the following: 
   \begin{itemize}
    \item (Completeness) The QPT prover described above wins with probability $\frac{1}{2}(1 + \valQu) -\negl(\lambda) \,.$
    \item (Soundness) Any PPT prover wins with probability at most $\frac{1}{2}(1 + \valNC) +\negl(\lambda) \,.$
  \end{itemize}
Here $\negl$ denote (possibly different) negligible functions. %
\end{thm*}

\emph{Proof sketch:} \newblue{The faithfulness condition as discussed in \Remref{criteria} is evidently satisfied by the compiled game.\footnote{The formal notion, as stated later in \Subsecref{criteria}, is also satisfied but the details are deferred.} } %
Suppose ${\cal A}$ is a PPT algorithm
that wins with probability non-negligibly greater than $\frac{1}{2}(1+\valNC)$.
Observe that one can associate a ``deterministic assignment'' corresponding
to ${\cal A}$, conditioned on some fixed first round messages, as follows:
simply rewind ${\cal A}$ to learn answers to all possible second round questions, obtaining
an assignment $\tau:Q\to A$, mapping questions to answers. Let us write
$\tau_{q_{1}}$ to make the dependence of the assignment on the first question more
explicit (note that the assignment depends on the encrypted question $c_{q_1}$ as well as the encrypted answer $c_{a_1}$). For the purpose of this overview, suppose also that ${\cal A}$
is consistent, i.e.\ if $q_{1}=q_{2}$, then $a_{1}=a_{2}$ (note that
this in particular ensures that, when $q_{1}=q_{2}$, ${\cal A}$ wins
with probability 1. One can show that an adversary that is not consistent can be turned into an adversary that is consistent and wins with at least the same probability). Now, to win with probability more than $\frac{1}{2}(1+\valNC)$, it must be that the $\tau_{q_1}$'s are different for different $q_{1}$'s. Otherwise,
${\cal A}$'s strategy is just a convex combination of deterministic assignments
and this by definition cannot do better than $\valNC$ when $q_{1}\neq q_{2}$.
But if the distribution over $\tau_{q_{1}}$s and $\tau_{q_{1}'}$s
is different for at least some $q_{1}\neq q_{1}'$, then one is able
to distinguish $\QFHE$ encryptions of $q_{1}$ from those of $q_{1}'$.
Thus, as long as the $\QFHE$ scheme is secure, no PPT algorithm can win
with probability non-negligibly greater than $\frac{1}{2}(1+\valNC)$.

In the above sketch, we glossed over the very important subtlety that, in order to obtain the truth table $\tau$, the reduction
needs to provide as input to $\cal A$ the ``correct'' decryption key $k$ (as the verifier does in the third message of our compiled game, where $k$ is such that $U_k = U_{k'}U_{k''}$). However, the reduction only sees \emph{encryptions} of $k'$ and $k''$. So, how does it compute $k$ without the secret keys?  Crucially, this is where the $\opad$ comes into play---it allows the reduction to instead use an independent uniformly random
$k$ (not necessarily the ``correct'' one) when constructing the reduction that breaks the security of the
$\QFHE$ scheme. The fact that such a $k$ is computationally indistinguishable from the correct one (from the point of view of the prover $\mathcal{A}$) follows precisely from the security of the $\opad$. %

\paragraph*{General Compilers}
The compiler we described earlier handles contextuality games with contexts of size $2$. How does one generalise it to contexts of arbitrary size? Unlike for KLVY, it is not entirely clear what the ``correct'' way is here.

We design two compilers (which seem incomparable). The first compiler applies universally to all contextuality games. The second applies primarily to contextuality games where the quantum value is $1$ (for instance, it works for the magic
square but not for KCBS). Notably, both compilers are 4-message protocols. 
\begin{itemize}
  \item $(|C|-1,1)$ compiler:
        \begin{itemize}
          \item Round 1: Ask $|C|-1$ questions under $\QFHE$.
          \item Round 2: Ask $1$ uniformly random question in the clear. If the question was already asked in Round $1$, check consistency. Otherwise, check the predicate.
        \end{itemize}
  \item $(|C|,1)$ compiler
        \begin{itemize}
          \item Round 1: Ask all $|C|$ questions under $\QFHE$.
          \item Round 2: Ask $1$ uniformly random question in the clear, and check consistency with the questions asked in Round $1$.
        \end{itemize}
\end{itemize}
By now, one would be wary of guessing that asking more questions under
$\QFHE$ is going to improve the security of the protocol. Indeed, the $(|C|-1,1)$
compiler (which reduces to the one we discussed above for $|C|=2$)
is the universal one. We show the following.
\begin{thm}[informal] \label{thm:general}
  Consider an arbitrary contextuality game $\G$, and let $\valNC$ and $\valQu$ be its non-contextual and quantum values respectively. The compiled game, obtained via the $(|C|-1,1)$ compiler, \newblue{is faithful to $\G$} and satisfies the following: 
  \begin{itemize}
  \item (Completeness) There is a QPT prover that wins with probability
    $1-\frac{1}{|C|}+\frac{\valQu}{|C|}-\negl$.
  \item (Soundness) PPT provers win with probability at most
    $1-\frac{1}{|C|}+\frac{\valNC}{|C|}+\negl$.
  \end{itemize}
  The compiled game, obtained via the $(|C|,1)$ compiler, \newblue{is also faithful to $\G$, and} satisfies the following:
  \begin{itemize}
  \item (Completeness) There is a QPT prover that wins with probability $\valQu$.
  \item (Soundness) PPT provers win with probability at most $1-\const_{1}+\negl$, where $\const_{1}=\min_{C\in\Call}\frac{\Pr(C)}{|C|}$ (this is constant in the sense that it is independent of the security parameter), and $\Pr(C)$ denotes the probability of sampling the context $C$.
  \end{itemize}
Here, $\negl$ are (possibly different) negligible functions.
\end{thm}
We make some brief remarks about the two compilers and defer the details to the main
text.
\begin{itemize}
  \item The $(|C|,1)$ compiler is not universal because, for instance, when
        applied to KCBS, there is no gap between the PPT and the QPT prover's
        winning probabilities. In fact, there is a PPT algorithm\footnote{Assuming that $\Eval$ is PPT if the circuit and input are classical. } that does better than
        the honest quantum strategy. Yet, the compiler does apply to the magic square game,
        for instance, because $\const_{1}<1$ and $\valQu=1$. In fact, for the magic square game, this compiler gives \emph{a better completeness-soundness gap} than the $(|C|-1,1)$ compiler.
  \item The $(|C|-1,1)$ compiler is universal in the sense that, when applied to any contextuality game with a gap between non-contextual and quantum value, the compiled game will have a constant gap between completeness and soundness. However, the resulting gap is sometimes smaller compared to the previous compiler. %
        It is unclear if one can do better than this, with or without
        increasing the number of rounds, while preserving universality.
\end{itemize}
%

%
%

%

%

%
%

%
%

%
%
%
%
%
%
    
%
%
%
%
%
%
%
%
%
%
%
%
%
%
%
%

%
%
%
%
%
%
%
%
%
%

%
%

%
%
%
%
%
%
%
%
%
%
%
%
%
%
%
%
%
%
%
%
%
%

%

%
%
%
%

%
%
%
%
%
%
%
%
%
%
%

%
%
%
%
%
%

%
%
%
%
%

%
%
%

%
%
%

%
%
%
%

%

\subsection{Contribution 2 | An Even Simpler Proof of Quantumness} 
\label{sec:tech-overview-poq}

Our proof of quantumness, like many of the existing ones in the literature, is based on the use of Trapdoor Claw-Free Functions (TCF). In our protocol, these are used to realize an ``encrypted CNOT'' functionality, which is the central building block of Mahadev's QFHE scheme \cite{mahadev2020classical}. The ``encrypted CNOT'' functionality allows a prover to homomorphically apply the gate $\mathsf{CNOT}^a$, while holding a (classical) encryption of the bit $a$. Formally, our protocol uses Noisy Trapdoor Claw-Free Functions (NTCF, defined formally in Definition \ref{def:NTCF}), but here we describe our scheme using regular TCFs for simplicity. 

\paragraph*{The proof of quantumness}
Our 2-round proof of quantumness is conceptually very simple. It can be viewed as %
combining and distilling ideas from the proofs of quantumness in \cite{KLVY22}, \cite{MCVY22} and our contextuality compiler. %
We provide an informal description here, and we defer a formal description to \Partref{poq}. At a high level, it can be understood as follows:
\begin{itemize}
    \item \textbf{Round 1}: \emph{Delegate the preparation of a uniformly random state in $\{\ket{0}, \ket{1}, \ket{+}, \ket{-}\}$, unknown to the prover.}
    
    \vspace{0.5mm}
    \noindent The verifier samples a bit $a$ uniformly at random. 
    \begin{itemize}
    \item \textbf{Message 1}: The verifier sends an appropriate encryption of $a$ to the prover (and holds on to the corresponding secret key).
    \end{itemize}
    The honest prover prepares the two-qubit state $\ket{+}\ket{0}$, along with auxiliary registers required to perform an ``encrypted CNOT'' operation. It then performs an ``encrypted CNOT'' operation (from the first qubit to the second), i.e.\ homomorphically applies $\mathsf{CNOT}^a$, followed by a measurement of the second (logical) qubit.
    \begin{itemize}
    \item \textbf{Message 2}: The prover sends all measurement outcomes to the verifier. 
    \end{itemize}
   Since a CNOT gate can be thought of as a deferred measurement in the standard basis, we can equivalently think of the prover's operations as performing an ``encrypted measurement'' of the first qubit, where the first qubit is being measured or not based on the value of $a$. Note that, after having performed these operations, the prover holds a \emph{single} qubit. Thanks to the specific structure of the ``encrypted CNOT'' operation from \cite{mahadev2020classical}, the resulting ``post-measurement'' qubit state is encrypted with a Quantum One-Time Pad, and is either:
    \begin{itemize}
        \item $\ket{+}$ or $\ket{-}$, if $a = 0$ (i.e.\ no logical CNOT was performed)
        \item $\ket{0}$ or $\ket{1}$, if $a = 1$ (i.e.\ a logical CNOT was performed)
    \end{itemize}
    All in all, at the end of round 1, the honest prover holds a uniformly random state in $\{\ket{0}, \ket{1}, \ket{+}, \ket{-}\}$, i.e.\ a BB84 state. This state is known to the verifier, who possesses $a$ and the secret key. From here on, the protocol no longer uses any encryption, and everything happens ``in the clear''.
    \item \textbf{Round 2}: \emph{Ask the prover to perform ``Bob's CHSH measurement''}. 

    \vspace{0.5mm} The astute reader may notice that the qubit held by the prover after Round 1 is distributed identically to ``Bob's qubit'' in a CHSH game where Alice and Bob perform the optimal CHSH strategy. More precisely, if one imagines that Alice has received her question and performed her corresponding optimal CHSH measurement (which is either in the standard or Hadamard basis), the leftover state of Bob's qubit is a uniformly random state in $\{\ket{0}, \ket{1}, \ket{+}, \ket{-}\}$, where the randomness comes both from the verifier's question (which in our protocol corresponds to the bit $a$), and Alice's measurement outcome.
    \begin{itemize}
        \item \textbf{Message 3}: The verifier sends a uniformly random bit $b$ to the prover (corresponding to Bob's question in a CHSH game).
    \end{itemize}
    The prover performs Bob's optimal CHSH measurement corresponding to question $b$.
    \begin{itemize}
        \item \textbf{Message 4}: The prover returns the measurement outcome to the verifier.
    \end{itemize}
    The verifier checks that the corresponding CHSH game is won.
\end{itemize}

\vspace{1mm}
\noindent In \Figref{prover_circuit}, we show the circuit for the honest quantum prover in our proof of quantumness.

   \captionsetup
  [figure]%
  {%
    name      = Fig.,
    labelfont = normal,
    labelsep  = colon
  }
\begin{figure*}%
        \begin{equation*}
    \Qcircuit{
    & \lstick{\ket{+}} & \multigate{1}{U_{f_0, f_1}} & \ustick{\ket{\psi} \in \text{ BB84}} \qw &  \gate{R_{b}}  & \meter \\
    & \lstick{\sum_{x \in \{0,1\}^n} \ket{x}} &  \ghost{U_{f_0, f_1}} & \gate{I \otimes H^{\otimes n-1}} & \meter 
    }
    \end{equation*}
   \caption{The prover's circuit. Here, $(f_0, f_1)$ is a pair of trapdoor claw-free functions (with inputs of size $n$). $U_{f_0, f_1}$ denotes the $(n+1)$-qubit unitary that coherently computes $f_0$ in the last $n$ qubits if the first qubit is $\ket{0}$, and computes $f_1$ otherwise. The circuit starts by preparing $\ket{+}$ in the first qubit, and a uniform superposition over all inputs in the next $n$ qubits. The circuit then applies $U_{f_0, f_1}$ (note that we are omitting auxiliary work registers that are required to compute $U_{f_0, f_1}$), followed by a layer of Hadamard gates on the last $n-1$ qubits. Then, the last $n$ qubits are measured. As a result, the leftover qubit $\ket{\psi}$ in the first register is now a BB84 state (which one it is depends on $f_0$, $f_1$, and the measurement outcome). As its second message, the verifier sends a bit $b$, and the prover applies the rotation $R_b$ %
  defined as follows: 
   $ \ket{0} \overset{R_b}{\mapsto}  \cos((-1)^b\pi/8)\ket{0} + \sin((-1)^b\pi/8)\ket{1}$ 
    and $\ket{1} \overset{R_b}{\mapsto} -\sin((-1)^b\pi/8)\ket{0} + \cos((-1)^b\pi/8)\ket{1}$. Finally, the prover measures the qubit in the standard basis.\protect\footnotemark
   }\label{fig:prover_circuit}
\end{figure*}

\footnotetext{The simplified description of the proof of quantumness before this figure is slightly inaccurate: it states that the prover prepares starts by preparing the two-qubit state $\ket{+}\ket{0}$. Technically, the prover only needs to prepare $\ket{+}$, and the role of the second qubit is performed by the first qubit of the pre-image register (which is initialised as a uniform superposition).} %

\paragraph*{Soundness}
An efficient quantum prover can efficiently pass this test with probability $\cos^2(\frac{\pi}{8}) \approx 0.85$, while an efficient classical prover can pass this test with probability at most $3/4$. The proof of classical soundness is fairly straightforward. %
In essence, %
a classical prover can be rewound to obtain answers to \emph{both} of the verifier's possible questions in Message 3. If the classical prover passes the test with probability $3/4+\delta$ (which is on average over the two possible questions), then the answers to both questions together must reveal some information about the encrypted bit $a$ (this is a simple consequence of how the CHSH winning conditions is defined). In particular, such a classical prover can be used to guess $a$ with probability $\frac12 + 2\delta$. This breaks the security of the encryption, as long as $\delta$ is non-negligible. We defer the reader to \Secref{poq-soundness} for more details. %

\paragraph*{Putting the ideas in perspective}%
We have already discussed in Section \ref{sec:results} how our proof of quantumness compares to existing ones in terms of efficiency. Here, we focus on how our proof of quantumness compares conceptually to \cite{KLVY22} and \cite{MCVY22}:
\begin{itemize}
\item In \cite{KLVY22}, the prover is asked to create an entangled EPR pair, of which the first half is encrypted, and the second half is in the clear. Then, the prover is asked to perform Alice's ideal CHSH measurement homomorphically on the first half, and Bob's CHSH measurement in the clear on the second half. Our proof of quantumness departs from this thanks to two observations:
\begin{itemize}
\item By leveraging the structure of the ``encrypted CNOT operation'' from \cite{mahadev2020classical}, the post-measurement state from Alice's homomorphic measurement can be re-used \emph{in the clear} (precisely because the verifier knows what the state is, but the prover does not). So the initial entanglement is not needed. This idea is also the starting point for our contextuality compiler from Section \ref{subsec:Contributions}, although for the latter we take this idea much further: we find a way to give the prover the ability to decrypt the leftover state without giving up on soundness. Our proof of quantumness is a baby version of this idea: it leverages the fact that the leftover encrypted state has a special form, namely it is a BB84 state.
\item In order to setup a ``CHSH-like correlation'' between the verifier and the leftover qubit used by the prover in Round 2, one does not need to compile the CHSH game in its entirety. This compilation, even for the simple CHSH game requires the prover to perform an ``encrypted controlled-Hadamard'' operation (because Alice's ideal CHSH measurements are in the standard and Hadamard bases). The latter requires three sequential ``encrypted CNOT'' operations. Instead, our observation is that one can setup this CHSH-like correlation more directly, as we do in Round 1 of our proof of quantumness.
\end{itemize}
\item From a different point of view, our proof of quantumness can also be viewed as a simplified version of \cite{MCVY22}. Indeed, observation (ii) is inspired by the proof of quantumness in \cite{MCVY22}, which introduces the idea of a ``computational'' CHSH test. One can interpret \cite{MCVY22} as setting up an ``encrypted classical operation'', akin to an ``encrypted CNOT'', that either entangles two registers or does not, ultimately having the effect of performing an ``encrypted measurement''. This is achieved via an additional round of interaction. Our proof of quantumness can be thought of as zooming in on this interpretation, and finding a direct way to achieve this without the additional interaction.
\end{itemize}

\newpage

\section{Preliminaries}
{
\Subsecref{notation} sets up some notation, \Subsecref{FHE-Scheme} formally introduces QFHE, and \Subsecref{NTCF} formally introduces trapdoor-claw free functions.
}

\subsection{Notation}\label{subsec:notation}
\begin{itemize}
  \item For mixed state $\rho$ and $\sigma$, we write $\rho\approx_{\epsilon}\sigma$ to mean that $\rho$ and $\sigma$ are
        at most $\epsilon$-far in trace distance, i.e. $\frac{1}{2} {\rm tr}(|\rho - \sigma|) \le \epsilon $.
  \item Let $X$ be the Pauli $X$ matrix. For a string $x \in \{0,1\}^n$, we write $X^{x}$ to mean $\otimes_{i=1}^n X^{x_{i}}$. We use a similar notation for Pauli $Z$.
  \item We denote the spectrum of an observable $O$ by $\spectr(O)$  and the support of a function $f$ by $\supp(f)$.
  \item We use the abbreviations PPT and QPT for probabilistic polynomial time and quantum polynomial time algorithms respectively.
 \item \emph{Vector/list indexed by a set}.\\
    Let $S$ and $V$ be finite sets. Let $\mathbf{v}$ be a vector/list with entries in $V$, indexed by $S$, i.e.\ $\mathbf{v}$
    contains a value in $V$ for each $s\in S$. We denote by $\mathbf{v}[s]\in V$ the value corresponding to $s$.\\
    For a subset $C\subseteq S$, we write $\mathbf{v}[C]$ for the vector
    obtained by restricting the indices of $\mathbf{v}$ to the set $C$. For example, if
    $\mathbf{v'}$ is a vector indexed by $C$, then, by $\mathbf{v}[C]=\mathbf{v'}$
    we mean that $\mathbf{v}[s]=\mathbf{v'}[s]$ for all $s\in C$. %
\end{itemize}

\begin{itemize}
    \item \emph{Asymptotic notation.} %
    \begin{itemize}
        \item \emph{Big-O.} Let $f,g:\mathbb{N}\to\mathbb{R}$. We write $f\le O(g)$ if $\exists c,n_{0}$ such that for all $n\ge n_{0}$, $f(n)\le cg(n)$. 
        \item \emph{Big-O on $\Lambda$.} Take $\Lambda\subseteq\mathbb{N}$ to be an infinite subset of $\mathbb{N}$. Then, by $f\le O(g)$ on $\Lambda$, we mean that there are $c,n_{0}$ such that for all $n\ge n_{0}$ in $\Lambda$, $f(n)\le cg(n)$.
    \end{itemize}
    We define Big-$\Omega$ on $\Lambda$ as a similar generalisation of the Big-$\Omega$ notation. 
\end{itemize}

\subsection{QFHE Scheme\label{subsec:FHE-Scheme}}

We formally define the notion of Quantum Homomorphic Encryption that we will employ in the rest of the paper. Note that, as in \cite{KLVY22} we only require security to hold against PPT provers.

\begin{defn}[Quantum Homomorphic Encryption]
  \label{def:QFHEscheme} A Quantum Homomorphic Encryption scheme for a class of circuits $\mathcal{C}$ is a tuple of algorithms $(\gen,\Enc,\Dec,\Eval)$ with the following syntax:
  \begin{itemize}
    \item $\gen$ is PPT. It takes a unary input $1^{n}$, and outputs a secret key $\sk$.
    \item $\Enc$ is QPT. It takes as input a secret key $\sk$ and a quantum state $\left|\psi\right\rangle$, and outputs a ciphertext $\ket{c}$. We require that if $\ket{\psi}$ is classical, i.e.\ a standard basis state, then $\Enc$ becomes a PPT algorithm. In particular, $\left|c\right\rangle$ is also classical.
    \item $\Eval$ is QPT. It takes as input a tuple $(C, \ket{\xi}, \ket{c})$, where 
    \begin{enumerate}
    \item $C: \mathcal{H} \times(\mathbb{C}_{2})^{\otimes n} \to (\mathbb{C}_{2})^{\otimes m}$ is a quantum circuit in $\mathcal{C}$, 
    \item $\ket{\xi} \in \mathcal{H}$ is a quantum state,
    \item $\ket{c}$ is a ciphertext encrypting an $n$-qubit state.
    \end{enumerate}
    $\Eval$ computes a quantum circuit $\Eval_C(\ket{\xi}, \ket{c})$, which outputs a ciphertext $\ket{c_{\mathsf{out}}}$. We require that, if $C$ has classical output, then $\Eval_C$ also has classical output. 
    \item $\Dec$ is QPT. It takes as input  a secret key $\sk$ and a ciphertext $\ket{c}$, and outputs a state $\left|\phi\right\rangle$. If $\ket{c}$ is classical, then $\ket{\phi}$ is also classical.
  \end{itemize}
  The correctness and security conditions are as follows:
  \begin{enumerate}
    \item \textbf{Correctness:} For any quantum circuit $C: \mathcal{H} \times(\mathbb{C}_{2})^{\otimes n} \to (\mathbb{C}_{2})^{\otimes m}$ in $\mathcal{C}$, there is a negligible function $\negl$, such that the following holds. For any quantum state $\ket{\xi} \in \mathcal{H}$, any $n$-qubit state $\ket{\psi}$, any $\lambda \in \mathbb{N}$, any $\sk \leftarrow \gen(1^{\lambda})$, and any $\ket{c} \leftarrow \enc_{\sk}(\ket{\psi})$, the following two states are $\negl(\lambda)$-close in trace distance.
          \begin{itemize}
            \item $C(\left|\xi\right\rangle \otimes\left|\psi\right\rangle)$,
            \item $\Dec_{\sk}(\Eval_{C}(\left|\xi\right\rangle ,\ket{c}))$.
          \end{itemize}
    \item \textbf{Security:} For all PPT distinguishers $D$, for all polynomial functions $\mathsf{poly}$, there exists a negligible function $\negl$ such that, for any two strings $x_{0},x_{1}\in\{0,1\}^{\mathsf{poly}(\lambda)}$, for all $\lambda$, 
          \begin{equation} \left|\Pr\left[D(\ket{c_{0}})=1:\begin{array}{c}
                \sk\leftarrow\gen(1^{\lambda}) \\
            \ket{c_0} \leftarrow\Enc_{\sk}(x_{0})
              \end{array}\right]-\Pr\left[D(\ket{c_{1}})=1:\begin{array}{c}
                \sk\leftarrow\gen(1^{\lambda}) \\
                \ket{c_1} \leftarrow\Enc_{\sk}(x_{1})
              \end{array}\right]\right|\le\negl(\lambda) \,.\label{eq:distinguisher-1}
          \end{equation}
  \end{enumerate}
\end{defn}

\addvspace{\baselineskip}  

A Quantum \emph{Fully} Homomorphic Encryption scheme ($\QFHE$) is a Quantum Homomorphic Encryption scheme for the class of all poly-size quantum circuits.

\paragraph{Form of encryption.}
For the rest of this work, we restrict to \atul{(public key)} QFHE schemes where encryption takes the following form. For an $n$-qubit state $\ket{\psi}$,  %
\begin{equation}
  (U_{k}\ket{\psi},\hat{k}) \leftarrow \QFHE.\Enc_{\sk}(\psi) \,,\label{eq:MBform}
\end{equation}
where $\{ U_{k} \}_{k\in K}$ forms a group (potentially up to global phases) and $K(n)$ is a finite set. Furthermore, $\hat k$ is a classical fully homomorphic  encryption of $k$, which can be decrypted using a PPT algorithm (specified by the $\QFHE$ scheme) given $\sk$. %

\addvspace{\baselineskip}

We show that such $\QFHE$ schemes also satisfy the following property.

\paragraph{Classical evaluation of classical circuits.} There is a PPT procedure $\cEval$ that is the classical version of $\Eval$. More precisely, we require that the following properties hold:
\begin{itemize}
    \item For any classical ciphertext $c$ produced by running $\Enc$ on a classical message $m$, and any classical circuit $C$, it holds that $c'\leftarrow \cEval(C,c)$ is such that $C(m)=\Dec(\sk,c)$.
    \item $\cEval$ and $\Eval$ produce identical distributions when run on identical classical inputs (i.e.\ ciphertexts and circuits). %
\end{itemize}

\begin{claim}[Classical Eval] \label{claim:classicalEval}
Consider any (public-key) $\QFHE$ scheme whose ciphertext is of the form in Equation~\ref{eq:MBform}. Then, there exists a PPT procedure $\cEval$, satisfying the classical evaluation of classical circuits property stated above.
\end{claim}

All known $\QFHE$ schemes satisfying Definition \ref{def:QFHEscheme}~\cite{mahadev2020classical, brakerski18} are public key schemes that satisfy both properties listed above. In particular, the first property above (form of encryption) is satisfied with $U_k$ being a Pauli pad and $k=(x,z)\in\{0,1\}^{2n}$ encoding which Pauli operator is applied, i.e.\ $U_{k}=X^{x}Z^{z}$. We end by outlining how the second property (classical evaluation of classical circuits) is proved generically as stated in \Claimref{classicalEval}.

\begin{proof}[Proof sketch for \Claimref{classicalEval}]

We will show how the classical evaluation works by looking at homomorphic evaluation of the "xor" (addition) and "and" (multiplication) operations, that are universal for classical circuits.
Consider two ciphertexts $\tilde{c}_0$ and $\tilde{c}_1$ corresponding to two messages $m_0$ and $m_1$, i.e. $\tilde{c}_0$ and $\tilde{c}_1$ are of the form $\tilde{c}_0 = (m_0 \oplus k_0, c_{k_0})$ and $\tilde{c}_1 = (m_1 \oplus k_1, c_{k_1})$ 
where $c_{k_0}\leftarrow \FHE.\Enc_{\pk}(k_0)$ and $c_{k_1}\leftarrow \FHE.\Enc_{\pk}(k_1)$. Then to homomorphically evaluate the addition, $\cEval$ will just add $\tilde{c}_0$ and $\tilde{c}_1$, namely compute $\tilde{c}_{m_0 \oplus m_1} = (m_0 \oplus k_0 \oplus m_1 \oplus k_1, c_{k_0} + c_{k_1}) = ((m_0 \oplus m_1) \oplus (k_0 \oplus k_1), c_{k_0 \oplus k_1})$, where $ c_{k_0 \oplus k_1}$ is the $\FHE$ encryption of $k_0 \oplus k_1$.

For the homomorphic multiplication operation, we will first compute: $(m_0 \oplus k_0) \cdot (m_1 \oplus k_1) = (m_0 \cdot m_1) \oplus (k_0 \cdot m_1 \oplus k_1 \cdot m_0 \oplus k_0 \cdot k_1)$. As a result, to complete the homomorphic evaluation, it suffices to show how to compute a $\FHE$ encryption of the pad $k_0 \cdot m_1 \oplus k_1 \cdot m_0 \oplus k_0 \cdot k_1$. This can be done as follows. First, compute the $\FHE$ encryptions $c_{m_0 \oplus k_0} = \FHE.\Enc_{\pk}(m_0 \oplus k_0)$ and $c_{m_1 \oplus k_1} = \FHE.\Enc_{\pk}(m_1 \oplus k_1)$.
Then given $c_{k_0}$ and $c_{m_0 \oplus k_0}$, the algorithm will homomorphically compute the xor of the two messages (under encryption), resulting in obtaining $c_{m_0} = \FHE.\Enc_{\pk}(m_0)$ (and analogously $c_{m_1} = \FHE.\Enc_{\pk}(m_1)$). Next, given $c_{k_1}$ and $c_{m_0}$ one can homomorphically compute 
$c_{m_0 \oplus k_1} = \FHE.\Enc_{\pk}(m_0 \oplus k_1)$ and similarly, $c_{m_1 \oplus k_0} = \FHE.\Enc_{\pk}(m_1 \oplus k_0)$. \\
The procedure will then sample a uniformly random $r$ and compute $c_r = \FHE.\Enc_{\pk}(r)$.

Finally, by homomorphically adding all the pieces together we get the final ciphertext of the multiplication:
$\tilde{c}_{m_0 \cdot m_1} := ((m_0 \cdot m_1) \oplus (k_0 \cdot m_1 \oplus k_1 \cdot m_0 \oplus k_0 \cdot k_1 \oplus r), \FHE.\Enc_{\pk}(k_0 \cdot m_1 \oplus k_1 \cdot m_0 \oplus k_0 \cdot k_1 \oplus r))$. It is clear to see that the ciphertext is in the right $\QFHE$ form, and moreover, due to the uniform pad $r$ we also have that the corresponding key is uniform, hence, ensuring that the output of $\cEval$ is identically distributed to the output of $\Eval$.

\end{proof}

\subsection{Noisy Trapdoor Claw-Free Functions}\label{subsec:NTCF}

Our work requires trapdoor claw-free functions for two purposes:
\begin{itemize}
\item They are used to instantiate Mahadev's QFHE scheme \cite{mahadev2020classical}.
\item They are used to construct the \emph{oblivious pad}.
\end{itemize}

However, the only known constructions of TCFs that satisfy the properties needed to instantiate QFHE (e.g.\ the property of Definition \ref{def:ntcf-2}) are ``noisy'' versions (NTCFs), and they are based on the hardness of LWE. We define NTCFs next. Before doing so, we point out that, while an NTCF is needed to instantiate QFHE based on current knowledge, any TCF (even based on quantum insecure assumptions, like Diffie-Hellman) suffices to build the \emph{oblivious pad} (however, for simplicity, we still use NTCFs in Section \ref{sec:opad}).

For readers familiar with the area, note that the ``adaptive hardcore bit'' property is not required for any of the constructions in this work.

\begin{defn}[NTCF family; paraphrased from \cite{BCMVV21}.]
  \label{def:NTCF}
  Let $\lambda$ be a security parameter. Let $\mathcal{X}$ and $\mathcal{Y}$ be finite sets
  and $\mathcal{K}$ be a finite set of keys (these sets implicitly depend on $\lambda$). A family of functions
  \[
    \mathcal{F} = \{f_{\pk,b}:\mathcal{X}\to\mathcal{D}_{\mathcal{Y}}\}_{\pk\in\mathcal{K},b\in\{0,1\}}
  \]
  is called a noisy trapdoor claw free function (NTCF) family
  if the following conditions hold (for each $\lambda$):
  \begin{enumerate}
    \item \textbf{Efficient Function Generation:} There exists an efficient
          probabilistic algorithm $\gen$ that generates a (public)
          key in $\mathcal{K}$ together with a trapdoor (secret key) $\sk$:
          \[
            (\pk,\sk)\leftarrow\gen(1^{\lambda}).
          \]
    \item \textbf{Trapdoor Injective Pair:}
          \begin{enumerate}
            \item \emph{Trapdoor:} There exists an efficient deterministic algorithm
                  $\inv$ such that with overwhelming probability over the
                  choice of $(\pk,\sk)\leftarrow\gen(1^{\lambda})$, the following holds:
                  \[
                    \inv(\sk,b,y)=x
                  \]
                  for all $b\in\{0,1\},x\in\mathcal{X}$ and $y\in\supp(f_{\pk,b}(x))$. %
            \item \emph{Injective pair:} For all keys $\pk\in\mathcal{K}$, there exists
                  a perfect matching $\mathcal{R}_{\pk}\subseteq\mathcal{X}\times\mathcal{X}$
                  such that $f_{\pk,0}(x_{0})=f_{\pk,1}(x_{1})$ if and only if $(x_{0},x_{1})\in\mathcal{R}_{\pk}$. Such a pair $(x_0, x_1)$ is referred to as a ``claw''.
          \end{enumerate}
    \item \textbf{Efficient Range Superposition.} For all keys $\pk\in\mathcal{K}$
          and $b\in\{0,1\}$ there exists a function $f'_{\pk,b}:\mathcal{X}\to\mathcal{D}_{\mathcal{Y}}$
          such that the following holds:
          \begin{enumerate}
            \item $\inv(\sk,b,y)=x_{b}$ and $\inv(\sk,b\oplus1,y)=x_{b\oplus1}$ for
                  all $(x_{0},x_{1})\in\mathcal{R}_{\pk}$ and $y\in\supp(f'_{\pk,b}(x_{b}))$. 
            \item There is an efficient deterministic procedure $\chk$ s.t. for all
                  $b\in\{0,1\}$, $\pk\in\mathcal{K}$, $x\in\mathcal{X}$ and $y\in\mathcal{Y}$,
                  \[
                    \chk(\pk,b,x,y)=\begin{cases}
                      1 & \text{if } y\in\supp(f'_{\pk,b}(x)) \\
                      0 & \text{else}.
                    \end{cases}
                  \]
                  Observe that $\chk$ is not provided the secret trapdoor $\sk$.
            \item For each $\pk\in\mathcal{K}$ and $b\in\{0,1\}$, it holds that
                  \begin{equation}
                  \label{eq:heilinger}
                    \mathbb{E}_{x\leftarrow\mathcal{X}}[H^{2}(f_{\pk,b}(x),f'_{\pk,b}(x))]\le\negl(\lambda)
                  \end{equation}
                  for some negligible function $\negl$, where %
                  $H^{2}$denotes the Hellinger distance.\\
                  Further, there is an efficient procedure $\mathsf{Samp}$ that on
                  input $\pk$ and $b\in\{0,1\}$ prepares the state
                  \[
                    \frac{1}{\sqrt{|\mathcal{X}|}}\sum_{x\in\mathcal{X},y\in\mathcal{Y}}\sqrt{(f'_{\pk,b}(x))(y)}\left|x\right\rangle \left|y\right\rangle .
                  \]
          \end{enumerate}
    \item \textbf{Claw-Free Property.} For any PPT adversary $\mathcal{A}$,
          \[
            \Pr\left[(x_{0},x_{1})\in\mathcal{R}_{\pk}:\begin{array}{c}
              (\pk,\sk)\leftarrow\gen(1^{\lambda}) \\
              (x_{0},x_{1})\leftarrow\mathcal{A}(\pk)
            \end{array}\right]\le\ngl{\lambda}.
          \]
  \end{enumerate}
  When ${\cal F}$ is used with other primitives, we refer to the algorithms
  $(\gen,\inv,\samp,\chk)$ and the various sets $({\cal X},{\cal D}_{{\cal Y}},{\cal Y})$
  with ${\cal F}$ prefixed, e.g. $\gen(1^{\lambda})$ is referred to
  as ${\cal F}.\gen(1^{\lambda})$.
\end{defn}

\newblue{We also implicitly require that the classical version of \emph{efficient range superposition} holds: uniform elements from the set $\calX$ (and therefore $\calD_\calY$ and $\calY$) can be efficiently sampled classically. This is the case for most NTCFs.} 

We define an additional property of NTCFs, which we will invoke directly in the soundness analysis of our proof of quantumness in Section \ref{sec:poq-soundness}.
\begin{defn}[``Hiding a bit in the xor'']
\label{def:ntcf-2}
Let $\mathcal{F} = \{f_{\pk,b}:\mathcal{X}\to\mathcal{D}_{\mathcal{Y}}\}_{\pk\in\mathcal{K},b\in\{0,1\}}$ be an NTCF family (as in Definition \ref{def:NTCF}). We say that $\mathcal{F}$ ``hides a bit in the xor'' if $\gen$ takes an additional bit of input $s$:
\[
            (\pk,\sk)\leftarrow\gen(1^{\lambda}, s) \,,
          \]
such that the following holds (in addition to the properties already satisfied by $\pk$ and $\sk$ in Definition \ref{def:NTCF}). For all $s \in \{0,1\}$, if $(\pk,\sk)$ is in the support of $\gen(1^{\lambda}, s)$, then $s = x_0[1] \oplus x_1[1]$ for all $(x_0, x_1) \in \mathcal{R}_{\pk}$, where $x_0[1]$ and $x_1[1]$ denote the first bits of $x_0$ and $x_1$ respectively, and $\mathcal{R}_{\pk}$ is as in Definition \ref{def:NTCF}. Moreover, for all QPT algorithms $\mathcal{A}$, there exists a negligible function $\negl$ such that, for all $\lambda$,
\begin{equation}
\label{eq:guess-bit}
    \Pr[s \leftarrow \mathcal{A}(\pk): (\pk, \sk) \leftarrow {\cal F}.\gen(1^{\lambda}, s) \,\,,\,\, s \leftarrow \{0,1\}] \leq 1/2 + \negl(\lambda) \,.
\end{equation}
\end{defn}

The following theorem is implicit in \cite{mahadev2020classical}.
\begin{thm}[\cite{mahadev2020classical}]
There exists an NTCF family with the property of Definition \ref{def:ntcf-2}, assuming the quantum hardness of the Learning With Errors (LWE) problem. 
\end{thm}

\pagebreak{}

\part{Even Simpler Proofs of Quantumness}
\label{part:poq}

\section{The Proof of Quantumness}
\label{sec:poq}
For an informal description of our proof of quantumness, we refer the reader to Section \ref{sec:tech-overview-poq}. Here we provide a formal description.
Our proof of quantumness makes use of a NTCF family $\mathcal{F} = \{f_{\pk,b}:\mathcal{X}\to\mathcal{D}_{\mathcal{Y}}\}_{\pk\in\mathcal{K},b\in\{0,1\}}$ (as in Definition~\ref{def:NTCF}) that additionally satisfies the property of Definition~\ref{def:ntcf-2}, i.e.\ a claw-free function pair hides a bit in the xor of the first bit of any claw. 

We parse the domain $\mathcal{X}$ as $\mathcal{X} = \{0,1\} \times \mathcal{V}$. Then, a bit more precisely, the property of Definition \ref{def:ntcf-2} says the following. Let $s \in \{0,1\}$, $(\pk, \sk) \leftarrow \mathcal{F}.\gen(1^{\lambda}, s)$. Let $(\mu_0, v_0)$ and $(\mu_1, v_1)$ be such that $f_{\pk, 0}(\mu_0, v_0) = f_{\pk, 1}(\mu_1, v_1)$. Then $\mu_0 \oplus \mu_1 = s$.

Our proof of quantumness invokes the procedure $\mathcal{F}.\samp$ from the ``Efficient Range Superposition'' property in Definition \ref{def:NTCF}. For simplicity, we describe our proof of quantumness assuming that the ``Efficient Range Superposition'' property holds exactly, i.e.\ the RHS of Equation \eqref{eq:heilinger} is zero. The actual procedure $\mathcal{F}.\samp$ is indistinguishable from the exact one, up to a negligible distinguishing advantage in the security parameter.
\begin{construction}[H]
~\\ Let $\lambda \in \mathbb{N}$ be a security parameter.
    \begin{itemize}
        \item \textbf{Message 1}: The verifier samples $s \leftarrow \{0,1\}$. Then, she samples $(\pk, \sk) \leftarrow \mathcal{F}.\mathsf{Gen}(1^{\lambda}, s)$. The verifier sends $\pk$ to the prover.
        \item \textbf{Message 2:} The prover uses $\mathcal{F}.\samp$ to create the state
        \begin{equation}
        \label{eq:superposition}
        \frac{1}{\sqrt{2 | \mathcal{X}|}}\sum_{b \in \{0,1\}, x \in \mathcal{X}, y \in \mathcal{Y}}\sqrt{f_{\pk}(x)(y)} \ket{b} \ket{x} \ket{y}
        \end{equation}
        This state can be created by preparing a qubit in the state $\ket{+}$ (and a sufficiently large auxiliary register), and then running the $\mathcal{F}.\samp$ controlled on the first qubit.
        
        For $b \in \{0,1\}$, let $(\mu^y_b, v_b^y)$ be such that $y \in \mathsf{Supp}(\mu^y_b, v_b^y)$ (this is a unique element by the properties of the NTCF). Then, we can rewrite the state as:
              \begin{align*}
                    & \frac{1}{\sqrt{2 | \mathcal{X}|}}\sum_{b \in \{0,1\} } \sqrt{f_{\pk}(\mu^y_0, v_0^y)(y)}\ket{b} \ket{\mu^y_b, v_b^y} \ket{y}                 \\
                  = & \frac{1}{\sqrt{2 | \mathcal{X}|}}\sum_{b \in \{0,1\} } \sqrt{f_{\pk}(\mu^y_0, v_0^y)(y)}\ket{b} \ket{\mu^y_0 \oplus b \cdot s, v_b^y} \ket{y}  
              \end{align*}
              The prover applies Hadamard gates to every qubit in the ``$v_b^y$'' register, and then measures all qubits except the first. Let the output be $(\mu, d, y)$. Then, the resulting post-measurement state of the first qubit is (up to global phases):
              \begin{equation} \label{eq:bb84_m2}
                  \begin{cases}
                      \ket{0} + (-1)^{d\cdot (v^y_0 + v^y_1)} \ket{1} \quad \textnormal{ if }  s=0 \\
                      \ket{\mu \oplus \mu_0^y} \,\,\,\,\, \quad \quad \quad \quad \quad \quad \textnormal{ if }  s=1
                  \end{cases}
              \end{equation}
              The prover returns $(\mu,d,y)$ to the verifier.
        \item \textbf{Message 3:} The verifier samples $c \leftarrow \{0,1\}$, and sends $c$ to the prover.
        \item \textbf{Message 4:}
              \begin{itemize}
                  \item If $c=0$, the prover measures the qubit in the basis $$\{\cos(\frac{\pi}{8})\ket{0} + \sin(\frac{\pi}{8}) \ket{1}, -\sin(\frac{\pi}{8}) \ket{0} + \cos(\frac{\pi}{8}) \ket{1}\} \,.$$
                  \item If $c=1$, the prover measures in the basis $$\{\cos(-\frac{\pi}{8})\ket{0} + \sin(-\frac{\pi}{8}) \ket{1}, -\sin(-\frac{\pi}{8}) \ket{0} + \cos(-\frac{\pi}{8}) \ket{1}\} \,.$$
              \end{itemize}
              The prover returns the outcome $b$ to the verifier.
        \item \textbf{Verifier's final computation:}
              \begin{itemize}
                  \item If $s=0$, the verifier runs $(\mu_0^y, v_0^y) \leftarrow \mathcal{F}.\mathsf{Inv}(\sk, 0, y)$ and $(\mu_1^y, v_1^y) \leftarrow \mathcal{F}.\mathsf{Inv}(\sk, 1, y)$. She sets $a = d \cdot (v_0^y \oplus v_1^y)$. Finally, she outputs $\textsf{accept}$ if $a \oplus b = c$, and $\mathsf{reject}$ otherwise.
                  \item If $s=1$, the verifier runs $(\mu_0^y, v_0^y) \leftarrow \mathcal{F}.\mathsf{Inv}(\sk, 0, y) $. She sets $a = \mu \oplus \mu_0^y$. Finally, she outputs $\textsf{accept}$ if $a \oplus b = 0$, and $\mathsf{reject}$ otherwise.
              \end{itemize}
    \end{itemize}
    \caption{(Proof of Quantumness)}\label{alg:simplerPoQ}
\end{construction}

\section{Analysis}
\subsection{Correctness}
Correctness as stated below is straightforward to verify.

\begin{thm}[Correctness] There exists a QPT algorithm $\mathcal{A}$ and a negligible function $\negl$ such that, for all $\lambda$,
    $$\Pr[\mathcal{A} \textnormal{ wins in \Algref{simplerPoQ}}] \geq \cos^2(\pi/8) - \negl(\lambda).$$
\end{thm}
 \begin{proof}
     The QPT algorithm $\mathcal{A}$ follows the steps of the prover in \Algref{simplerPoQ}. Then, correctness follows from the fact that the state in Equation \eqref{eq:bb84_m2} is one of the four BB84 states, and a straightforward calculation. One can also realize that the prover's measurement is Bob's ideal CHSH measurement corresponding to question $c$, and that the verifier is precisely checking the appropriate CHSH winning condition based on $s$. The negligible loss in the correctness probability comes from the fact that, as we mentioned earlier, the procedure $\mathcal{F}.\samp$ actually generates a state that is only negligibly close to that of Equation \eqref{eq:superposition}.
 \end{proof}

\subsection{Soundness}
\label{sec:poq-soundness}

\begin{thm}
    For any PPT prover $\mathcal{A}$, there exists a negligible function $\negl$ such that, for all $\lambda$,
    $$\Pr[\mathcal{A} \textnormal{ wins in \Algref{simplerPoQ}}] \leq \frac34 + \negl(\lambda) \,.$$
\end{thm}

\begin{proof}
    We show this by giving a reduction from a prover $\mathcal{A}$ for \Algref{simplerPoQ} to an adversary $\mathcal{A}'$ breaking the property of Definition \ref{def:ntcf-2} satisfied by $\mathcal{F}$, i.e.\ predicting the bit $s$ ``hidden in the claw-free pair''. Specifically, we show that if $\mathcal{A}$ wins with probability $\frac34 + \delta$ for some $\delta \geq 0$, then $\mathcal{A}'$ can guess the the bit $s$ from Equation \eqref{eq:guess-bit} with probability at least $\frac12 + 2\delta$. $\mathcal{A}'$ proceeds as follows, for security parameter $\lambda$:
              \begin{itemize}
                  \item[(i)] $\mathcal{A}'$ receives $\pk$ (where $\pk$ is sampled according to $(\pk, \sk) \leftarrow \mathcal{F}.\gen(1^{\lambda}, s)$ for some uniformly random $s$).
                  \item[(ii)] $\mathcal{A}'$ runs $\mathcal{A}$ on ``first message'' $\pk$. Let $z$ be the output.
                  \item[(iii)] Continue running $\mathcal{A}$ on ``second message'' $c = 0$. Let $b_0$ be the received output.
                  \item[(iv)] Rewind $\mathcal{A}$ to just before step (iii). Run $\mathcal{A}$ on ``second message'' $c=1$. Let $b_1$ be the received output.
                  \item[(v)] Output the guess $s' = b_0 \oplus b_1$.
              \end{itemize}

%

    We show that $\Pr[\mathcal{A}' \textnormal{ wins}] \geq \frac12 + 2\delta$. Using the notation introduced above, first notice that when $b_0$ is a response that a Verifier from \Algref{simplerPoQ} would accept, we have, by construction, that
    \begin{equation}
        \label{eq:1}
        a \oplus b_0 = 0 \,,
    \end{equation}
    where $a$ is defined as in \Algref{simplerPoQ} (given $\mathcal{A}$'s first response $z$, and $\sk$). We refer to such a $b_0$ as ``valid''. Similarly, if $b_1$ is valid, then, by construction,
    \begin{equation}
        \label{eq:2}
        a \oplus b_1 = s \,.
    \end{equation}
    Summing Equations \eqref{eq:1} and \eqref{eq:2} together gives $b_0 \oplus b_1 = s$ (where $s$ is the bit that was actually sampled in (i)). Since $\mathcal{A'}$ outputs precisely $s' = b_0 \oplus b_1$, we have that
    \begin{equation}
        \Pr[\mathcal{A}' \textnormal{ wins}] = \Pr[s' = s] \geq \Pr[b_0 \textnormal{ and } b_1 \textnormal{ are valid}] \,.
    \end{equation}

    Now, for fixed $\pk$ and $z$, define
    $$p^{\pk, z}_{win, 0} := \Pr[b_0 \textnormal{ is valid} \,|\, \pk, z] \,,$$
    and define $p^{\pk, z}_{win, 1}$ analogously using $b_1$. Then, we have
    \begin{align}
        \Pr[\mathcal{A}' \textnormal{ wins}] & \geq  \Pr[b_0 \textnormal{ and } b_1 \textnormal{ are valid}] \nonumber                  \\
                                             & =\mathbb{E}_{\substack{s \leftarrow \{0,1\} \\(\pk, \sk) \leftarrow \mathcal{F}.\gen(1^{\lambda}, s)           \\ z \leftarrow \mathcal{A}(\pk)}}  \,\, \left[p^{\pk, z}_{win, 0} \cdot  p^{\pk, z}_{win, 1} \right] \nonumber\\
                                             & \geq \mathbb{E}_{\substack{s \leftarrow \{0,1\} \\(\pk, \sk) \leftarrow \mathcal{F}.\gen(1^{\lambda}, s)        \\ z \leftarrow \mathcal{A}(\pk)}} \left[ p^{\pk, z}_{win, 0} + p^{\pk, z}_{win, 1} - 1 \right] \label{eq:4}\\
                                             & =  2 \cdot \mathbb{E}_{\substack{s \leftarrow \{0,1\} \\(\pk, \sk) \leftarrow \mathcal{F}.\gen(1^{\lambda}, s) \\ z \leftarrow \mathcal{A}(\pk)}}  \left[ \frac12 \left( p^{\pk, z}_{win, 0} + p^{\pk, z}_{win, 1} \right) \right]- 1 \label{eq:3} \,,
    \end{align}
    where \eqref{eq:4} follows from the inequality $x \cdot y \geq x+y-1$ for all $x,y \in [0,1]$ (subtracting $x$ from both sides makes this inequality apparent).

    Now, notice that
    \begin{equation}
        \mathbb{E}_{\substack{s \leftarrow \{0,1\} \\(\pk,\sk) \leftarrow \mathcal{F}.\gen(1^{\lambda}, s)  \\ z \leftarrow \mathcal{A}(\pk)}}  \left[ \frac12 \left( p^{\pk, z}_{win, 0} + p^{\pk, z}_{win, 1} \right) \right] = \Pr[\mathcal{A} \textnormal{ wins}] \,.
    \end{equation}
    Thus, plugging this into \eqref{eq:3}, gives
    \begin{equation}
        \Pr[\mathcal{A}' \textnormal{ wins}] \geq 2 \cdot  \Pr[\mathcal{A} \textnormal{ wins}] -1 = 2 \cdot \left(\frac34 + \delta \right) - 1  = \frac12 + 2\delta\,,
    \end{equation}
    as desired.

\end{proof}

\pagebreak{}

\part{A Computational Test of Contextuality---for size 2 contexts}

\section{Contextuality Games\label{sec:Contextuality-Games}}

We start by defining contextuality games. This section uses the list/vector notation from \Subsecref{notation}.
\begin{defn}[Contextuality game, strategy, value]
  \label{def:ContextualityGame}A contextuality game $\G$ is specified by a tuple $(Q,A,\Call,{\rm pred},{\cal D})$
  where
  \begin{itemize}
    \item $Q$ is a set of questions.
    \item $A$ is a set of answers. %
    \item $\Call$ is a set of subsets of $Q$. We refer to an element of $\Call$ as a \emph{context}.
    \item ${\cal D}$ is a distribution over contexts in $\Call$, and
    \item ${\rm pred}$ is a binary-valued function. It takes as input pairs of the form $(\mathbf{a}, C)$, where $C\in\Call,$ and $\mathbf{a} \in A^{|C|}$. We think of $\mathbf{a}$ as being indexed by the questions in $C$, and we write $\mathbf{a}[q]\in A$ to represent the answer to question $q\in C$.
  \end{itemize}
  $\G$ can be thought of as a 2-message game between a referee and a player that proceeds as follows:
  \begin{enumerate}[(i)]
    \item The referee samples a context $C\leftarrow\calD$ and sends $C$ to the prover.
    \item The player responds with $\mathbf{a} \in A^{|C|}$.
    \item The referee accepts if ${\rm pred}(\mathbf{a},C)=1$.
  \end{enumerate}
  A \emph{strategy} for the game $\G$ is specified by a family $P$ of probability distributions (one for each $C \in \Call$), where $P[\mathbf{a}|C]$ can be thought of as the probability of answering $\mathbf a$ on questions in $C$.

  We define the \emph{value of $\G$} with respect to a strategy $P$ to be
  \[
    {\rm val}(P):=\sum_{C\in\Call, \mathbf{a} \in A^{|C|}}{\rm pred}(\mathbf{a},C)\cdot P(\mathbf{a}|C) \cdot \Pr_{\cal D}[C].
  \]
where $\Pr_{\cal D}[C]$ denotes the probability assigned to $C$ by $\cal D$. \end{defn}

\begin{rem}\label{rem:samesizecontexts} One can assume that all contexts $C\in \Call$ have the same size without loss of generality by adding additional observables, and having the predicate $\pred$ remain unchanged (i.e.\ the predicate does not depend on the values taken by the additional observables).
\end{rem}

 Let us consider the magic square as a contextuality game to clarify the notation. 
  \begin{example}[Magic Square Contextuality Game]
    \label{exa:masq} The set of questions is %
    $$Q:=\{(1,1), (1,2), (1,3), (2,1), (2,2), (2,3), (3,1), (3,2), (3,3)\} \,,$$ 
    $A:=\{+1,-1\}$
    and $\Call$ consists of all rows and columns of this matrix, i.e.
    subsets of the form {$\{(r,1),(r,2),(r,3)\}$ for $r\in\{1,2,3\}$ and $\{(1,c),(2,c),(3,c)\}$
    for $c\in\{1,2,3\}$}. The distribution ${\cal D}$ uniformly selects
    an element $C\leftarrow\Call$, i.e. it samples either the row or
    column uniformly. The predicate ${\rm pred}(\mathbf{a},C)$ is $0$
    unless the following holds:
    \[
      \prod_{q\in C}\mathbf{a}[q]=\begin{cases}
        -1 & {\text{if } C=\{(1,3),(2,3),(3,3)\}} \\
        1  & \text{else}.
      \end{cases}
    \]
    Call this game, $\G_{{\rm masq}}:=(Q,A,\Call,{\rm pred},{\cal D})$,
    the magic square game.
  \end{example}

  The game only becomes meaningful when the strategies that the players
  follow are somehow restricted. Thus, to consider non-contextual and
  quantum values associated with a contextuality game $\G$, we must
  first define the corresponding strategies.

For the following two definitions, let $\G = (Q,A,\Call)$ be a contextuality
game.
\begin{defn}[Classical/Non-Contextual Strategy for $\G$ | \textbf{${\calS}_{NC}(\G)$}%
]
  We write a joint probability distribution over answers to \emph{all} questions in $Q$ as
  \[
    P_{{\rm joint}}[\mathbf{a_{joint}}]\in [0,1]
  \]
  where $\mathbf{a_{joint}}$ is a vector of length $|Q|$, indexed
  by $Q$. A strategy $P$ is a \emph{non-contextual strategy} if $P[\mathbf{a}|C]$ is derived
  from some joint probability $P_{{\rm joint}}[\mathbf{a_{joint}}]$,
  i.e.
  \[
    P(\mathbf{a}|C)=\sum_{\mathbf{a_{joint}}:\mathbf{a_{joint}}[C]=\mathbf{a}}P_{{\rm joint}}(\mathbf{a_{joint}})
  \]
  where by  $\mathbf{a_{joint}}[C]=\mathbf{a}$ {we mean that $\mathbf{a_{joint}}[\tilde{q}]=\mathbf{a}[\tilde{q}]$
    for all $\tilde{q}\in C$.}\footnote{Note that $\mathbf{a_{joint}}$ is a vector
    of length $|Q|$ while $\mathbf{a}$ is a vector of length $|C|$.} The set of all non-contextual strategies is denoted by ${\calS}_{NC}(\G)$.
\end{defn}

{We now define a quantum strategy for $\G$. It would be helpful to
first define ``qstrat''---a state and collection of observables---consistent with a contextuality game $\G$. }
  
The quantum strategy for a contextuality game $\G$ is defined as follows.

\begin{defn}[Quantum Strategy for $\G$ | \textbf{${\calS}_{Qu}(\G)$}] \label{def:qstrat}
  Given a contextuality game $\G$, consider a triple 
  $$\qstrat:=({\cal H},\left|\psi\right\rangle ,\mathbf{O}),$$
  where ${\cal H}$ is a Hilbert space, $\left|\psi\right\rangle \in{\cal H}$
  is a quantum state and $\mathbf{O}$ is a list of Hermitian operators
  (observables), indexed by the questions $Q$ (i.e. $\mathbf{O}[q]$
  is an observable for each $q\in Q$) acting on ${\cal H}$. 
  Furthermore, suppose that \textbf{$\mathbf{O}$} satisfies the following:

  (i) $\spectr[\mathbf{O}[q]]\subseteq A$ for all $q\in Q$, i.e. the
  values returned by the observables correspond to potential answers
  in the game.

  (ii) $[\mathbf{O}[q],\mathbf{O}[q']]=0$ for all $q,q'\in C$, for
  each $C\in\Call$, i.e. operators corresponding to the same context
  are compatible. 

  A strategy is a \emph{quantum strategy}, i.e. $P\in{\calS}_{Qu}(G)$,
  iff there is such a triple $\qstrat:=({\cal H},\left|\psi\right\rangle ,\mathbf{O})$
  satisfying
  \begin{equation}
    P[\mathbf{a}|C]=\Pr\left[\text{Start with \ensuremath{\left|\psi\right\rangle } measure \ensuremath{\mathbf{O}[C]} \text{and obtain \ensuremath{\mathbf{a}}}}\right]\quad\forall\ \mathbf{a},C.\label{eq:quantumProb}
  \end{equation}

\end{defn}

{
  To be concrete, consider the Magic square as an example. %
}

\begin{example}[$\qstrat$ for the magic square] %
  \label{exa:con-instance-The-Magic-Square}
  For the magic square as in \Exaref{MagicSquare}, the optimal quantum strategy corresponds to using $\qstrat:=(\cal{H},\ket{\psi}, \mathbf{O})$ where ${\cal{H}}:=\mathbb{C}^2\times \mathbb{C}^2$, $\ket{\psi}\in \cal{H}$ is any fixed state, say $\ket{\psi}=\ket{00}$, and $\mathbf{O}$ is specified by the following (indexed by the questions $Q$ in some canonical way):
  \begin{equation}
    \begin{array}{ccc}
      \mathbb{I}\otimes X & Z\otimes\mathbb{I}  & Z\otimes X   \\
      X\otimes\mathbb{I}  & \mathbb{I}\otimes Z & X\otimes Z   \\
      X\otimes X          & Z\otimes Z          & XZ\otimes XZ.
    \end{array}\label{eq:operators}
  \end{equation}
\end{example}

{
  Note that this is not the most general strategy, i.e. one could have different operators for each context and the operators in a given context may not even commute. Why then, do we restrict to the strategies defined above? This is because, as explained briefly in the introduction, the appeal of a contextuality test is that even by only measuring commuting observables at any given time, one can find scenarios where the idea of pre-determined values of observables becomes untenable. 
}

{

While conceptually appealing, the main obstacle to testing contextuality
using the contextuality game is that it is unclear how the provers'
strategies can be restricted to the ones above. %
As noted earlier, when one considers a Bell game as
a contextuality game, spatial separation automatically restricts the
provers' strategies in this way. %

Assuming %
that the provers follow only these
restricted strategies, one can nevertheless define the classical/non-contextual
and quantum value of the contextuality game.}

\begin{defn}[Non-Contextual and Quantum Value of a Contextuality Game $\G$]
  Let $\G=(Q,A,\Call,{\rm pred},{\cal D})$ be a contextuality game
  (see \Defref{ContextualityGame} and recall the definition of ${\rm val}$).
  The non-contextual value of $\G$ is given by
  \[
    {\valNC}:=\max_{P\in{\calS}_{NC}(\G)}{\rm val}(P)\quad\text{and similarly}\quad{\valQu}:=\max_{P\in{\calS_{Qu}(\G)}}{\rm val}(P).
  \]
  is the quantum value of $\G$.
\end{defn}

{Before moving further, we quickly note what these values are for the magic
  square.
  \begin{example}[Magic Square Contextuality Game (cont.)]
    Let $\G_{{\rm masq}}$ be as in \Exaref{masq} and note that
    \[
      \valNC=5/6,\quad\text{while}\quad {\valQu}=1
    \]
    using the Magic-Square con-instance as in \Exaref{con-instance-The-Magic-Square}. %
  \end{example}

}

{The KCBS example (\Exaref{KCBSintro}) mentioned in the introduction, in this notation, can be expressed as follows.}

\begin{thm}[KCBS]
  There exists a contextuality game $\G$ with ${\valNC}=0.8 < {\valQu}= \frac{2}{\sqrt{5}} \approx 0.8944$
  where the quantum strategy can be realised using qutrits and five
  binary valued observables (i.e. $\qstrat:=({\cal H},\left|\psi\right\rangle ,\mathbf{O})$
  is such that $\dim{\cal H}=3$, $|Q|=|\mathbf{O}|=5$ and $A=\{\pm1\}$).
\end{thm}

\section{Criteria for being an operational test of contextuality \label{subsec:criteria}}

\atul{We briefly mentioned in the technical overview that an ``operational test of contextuality'' has a close correspondance with a contextuality game. In this section, we formalise what we mean by this correspondence.} %

{As alluded to in the discussion above and in \Subsecref{tech-overview-contextuality}, it is reasonable to assert that the most general non-contextual physically relevant model of computation
is simply a probabilistic poly time Turing machine. Thus, any proof
of quantumness, i.e. any test that distinguishes a PPT machine from
a quantum machine, may be taken to be a proof of contextuality. In
other words, quantumness and contextuality become equivalent notions,
in accordance with this definition.

This, however, is unsatisfactory because quantumness could be arising
from some other non-classical feature of quantum mechanics which may
have no a priori connection to contextuality. For instance, (assuming
factoring is hard for a PPT machine) equivalence of quantumness and
contextuality would mean that factoring serves as a proof of contextuality.
Yet, it is unclear how Shor's algorithm demonstrates contextuality
in any direct way.

One proposal could be that a satisfactory test must be ``universal''
in the sense that corresponding to each contextuality game, there
should be a systematic way to construct a corresponding test such
that the completeness and soundness values of the test correspond
to the quantum and noncontextual value of the underlying contextuality
game. A little thought shows that this is not enough, as illustrated
by the example below.
\begin{example}
  \label{exa:simpleProofOfQuantumness}The example uses a general non-interactive
  proof of quantumness as an ingredient: Let $\PoQ=(\gen,\verify,\cert)$
  denote a proof of quantumness protocol where $\gen,\verify$ are PPT
  algorithms and $\cert$ is a QPT algorithm. Here, $\gen$ generates
  a challenge, $\cert$ generates a response to the challenge and $\verify$
  tests whether the response is valid. For simplicity, we take them to satisfy the property that
  no PPT machine can make $\verify$ accept while $\cert$ always produces
  a valid certificate, i.e.\ $\PoQ$ has perfect completeness and soundness. %
 Given a contextuality game $\G$, the protocol is the following: the verifier runs a proof of quantumness protocol $\PoQ$
  and if the prover passes, the verifier accepts with probability $\G.\valQu$
  and if the prover fails, accepts with probability $\G.\valNC$.
  \begin{center}
  \begin{tabular}{|V{\linewidth}|c|c|}
  \multicolumn{1}{l}{Verifier} & \multicolumn{1}{c}{} & \multicolumn{1}{c}{Prover}\tabularnewline
  \cline{1-1}\cline{3-3}
  $\chall\leftarrow\PoQ.\gen$ &  & \tabularnewline
   & $\xrightarrow{\chall}$ & \tabularnewline
   & $\xleftarrow{\poq}$ & \tabularnewline
  \begin{cellvarwidth}[t]
  if $\poq$ is a valid response for $\chall$ 

  $\quad$(i.e. $\PoQ.\verify(\chall,\poq)=1$)\\

  $\qquad$accept with probability $\G.\valNC$.\\

  Otherwise\\

  $\qquad$accept with probability $\G.\valQu$.
  \end{cellvarwidth} &  & \tabularnewline
  \cline{1-1}\cline{3-3}
  \end{tabular}
  \par\end{center}
  Evidently, this protocol has completeness $c=\G.\valQu$ and soundness
  $s=\G.\valNC$.

\end{example}

{\Exaref{simpleProofOfQuantumness} is unsatisfactory as a test of contextuality because, even though its soundness and completeness values correspond to the classical and quantum values of the game $\G$, this correspondence is artifical: the prover is not asked a single question that is related to $\G$.} 
We now formalise a stronger and more natural notion of correspondence. To this end, we introduce some notation.

\subsection{Notation}

\atul{As a starting point towards a candidate ``operational test'' $\cal{P}$ for $\G$, consider a proof of quantumness protocol $\cal P$ involving a PPT verifier $V$ and a prover $P$, with soundness $s$ and completeness $c$. } 
To be concrete, suppose protocol ${\cal P}$
involves four messages. %
Further
suppose that the messages are parsed as follows:

\begin{center}
  \begin{tabular}{|c|c|c|}
    \multicolumn{1}{c}{Verifier} & \multicolumn{1}{c}{}        & \multicolumn{1}{c}{Prover}\tabularnewline
    \cline{1-1}\cline{3-3}
                                 & $\xrightarrow{t_{1},t_{2}}$ & \tabularnewline
                                 & $\xleftarrow{t_{3},t_{4}}$  & \tabularnewline
                                 & $\xrightarrow{t_{5},t_{6}}$ & \tabularnewline
                                 & $\xleftarrow{t_{7},t_{8}}$  & \tabularnewline
    \cline{1-1}\cline{3-3}
  \end{tabular}
  \par
\end{center}

For instance $t_{1}$ could be a public key and $t_{2}$ could be an encrypted
message.

\paragraph*{Mapping messages in ${\cal P}$ to questions/answers in $\G$.}
Let $\G=(Q,A,\Call,{\rm pred},{\cal D})$. We require that ${\cal P}$ classifies each message into one of three
categories: (1) a question in $\G$, (2) an answer in $\G$ or (3)
other. More precisely, we require that ${\cal P}$ specifies a map $\mathsf{map}_{i}$
for each message $t_{i}$. \\
For indices $i$ corresponding to messages
\emph{received by the prover} (in our example, $t_{1},t_{2},t_{5},t_{6}$),
$\map_{i}$ is either
\begin{itemize}
  \item the constant $\perp$ output map, $\map^{\perp}$ (i.e.\ $\map^{\perp}$
        outputs $\perp$ on all inputs) or
  \item a map $r_{i}:C_{i}\to Q$ from the set of possible messages $C_{i}$
        to a question in the game $\G$. 
\end{itemize}
{One can think of $r_i$ as a ``decoding map'' for the underlying question.}\\
For indices $i$ corresponding to messages\emph{ sent by the prover
}(in our example $t_{3},t_{4},t_{7},t_{8}$), $\map_{i}$ is either
\begin{itemize}
  \item the constant $\perp$ output map $\map^{\perp}$ or
  \item a map $s_{i}:A\to C_{i}$ from the set of answers $A$ in the game
        to possible messages $C_{i}$. %
\end{itemize}
{Similarly, one can think of $s_i$ as an ``encoding map'' for the underlying answer.}
Continuing \atul{in the same vein as the} example above, suppose $t_{1}=\pk$ is a public key
for some encryption scheme and $t_{2}\leftarrow\Enc_{\pk}(q)$ is
an encryption of a question $q\in Q$. Then, a natural choice for
the corresponding maps is $\map_{1}=\map^{\perp}$ and $\map_{2}=\Dec_{\sk}$
where $\sk$ is the secret key corresponding to $\pk$.
Note that we gave one choice for these maps but ${\cal P}$ can specify
these maps arbitrarily. The non-triviality arises from the requirements
we place on ${\cal P}$, using these maps.

\paragraph*{Faithfulness to $\G$.}
In order to capture the fact that protocol ${\cal P}$ is faithfully executing an instance
of the game $\G$, and not rewarding some other capability of the prover, we consider two families of ``simulators''. These are computationally unbounded machines and are meant to simulate
either a classical or a quantum strategy. The idea is that using $\{\map_{i}\}_{i}$
one can isolate the messages that correspond to questions and answers
in the game $\G$. One can then define machines that answer these
questions non-contextually, or using a quantum strategy, \atul{and can behave arbitrarily otherwise.} 
More precisely, we have the following:
\begin{itemize}
  \item A classical simulator $S_{\tau,\aux}$ is an unbounded machine that
        is designed to interact with $V$ and responds to ``encoded'' questions (as specified by $\{\map_{i}\}_{i}$) using some non-contextual assignment
        $\tau:Q\to A$ \atul{i.e.\ the simulator decodes the message $t_i$ using $r_i$ to recover a question $q\in Q$, computes $a =\tau(q)$ and responds with the encoding $s_j(a)$, for each pair of message indices $(i,j)$ corresponding to a question and its answer.}
        \atul{It responds to the remaining messages using an arbitrary}
        strategy specified by\footnote{Here, $\aux$ may be thought of as
        describing the ``program of a Turing Machine'' and may involve potentially
        unbounded computation.} $\aux$.
        We denote by $\mathbf{S}_{{\rm NC}} =\{S_{\tau,\aux}\}_{\tau,\aux}$
        the set of all classical simulators.
  \item A quantum simulator $S_{\qstrat,\aux}$ is an unbounded
        machine that is also designed to interact with $V$ and responds to
        the ``encoded'' questions
        using %
        $\qstrat=(\left|\psi\right\rangle,\mathbf{O} )$ of $\G$ (see \Defref{qstrat}) 
        but responds to
        the remaining messages using an arbitrary strategy specified by $\aux$.
        We denote by $\mathbf{S}_{\qu}=\{S_{\qstrat,\aux}\}_{\qstrat,\aux}$
        the set of all quantum simulators.
\end{itemize}

In addition to the simulators, we will also require ${\cal P}$ to
specify an efficient classical procedure $\mathsf{qProver}$ that
maps any (classical description of a) quantum strategy 
\[
\qstrat=(\left|\psi\right\rangle,\underbrace{\{O_{1},O_{2}\dots\}}_{=\mathbf{O}} )
\]
 for $\G$ to a QPT prover 
\[
\mathsf{qProver}(\qstrat)=P_{\qstrat}
\]
 that is designed to interact with the verifier $V$ in protocol ${\cal P}$.
$\mathsf{qProver}$ must satisfy the following {\emph{``marginal''
requirement}}: the behaviour of $P_{\qstrat}$ when asked questions
within a context $\{q_{1}\dots q_{k}\}=C\in\Call$ under the encoding
specified by $\{\map_{i}\}$, depends only on the state $\left|\psi\right\rangle $
and observables $\mathbf{O}[C]=\{O_{q_{1}},\dots O_{q_{k}}\}$ but
not on the remaining observables $\mathbf{O}[Q\backslash C]=\mathbf{O}\backslash\mathbf{O}[C]=\mathbf{O}\backslash\{O_{q_{1}},\dots O_{q_{k}}\}$,
specified by $\qstrat$. Just as the case when $\qstrat$ is
played in $\G$ and questions are asked in $C$, the remaining observables
don't make any difference.

\subsection{Criteria}
\atul{We can now state our definition. We ignore negligible additive factors for clarity.}

\begin{defn}[Operational Test of Contextuality]\label{def:OperationalTestOfContextuality}
  We say ${\cal P}$ is an \emph{operational test of contextuality} if there exist $s$ and $c$ (where $s<c$) such that, in
  addition to being a proof of quantumness with soundness $s$ and completeness
  $c$, ${\cal P}$ is \emph{faithful to} some \emph{contextuality game} $\G$. We
  say ${\cal P}$ is \emph{faithful to} $\G=(Q,A,\Call,\pred,{\cal D})$ (with parameters $s$ and $c$)
  if the following hold:
  \begin{enumerate}
    \item \label{enu:wellformed}Well-formedness: ${\cal P}$ must specify maps
          $\{\map_{i}\}_{i}$ and the procedure $\qProver$ as described above. Moreover, for any possible set $Q'$ of questions that the verifier $V$ asks in a single execution, i.e.\
          \begin{equation}
            Q':=\{\map_{i}(t_{i})\}_{i:t_{i}=r_i}\,,\label{eq:Qprime} %
          \end{equation}
          we require that the questions belong to some context, i.e.\ $Q'\subseteq C$ for some context $C\in\Call$. 
    \item \label{enu:gsoundness}$\G$-soundness: For all classical simulators,
          i.e. $S\in\mathbf{S}_{{\rm NC}}$, the probability that the verifier $V$ accepts when interacting with $S$ should be at most the classical
          value $s$, i.e.
          \[
            \Pr\left(1\leftarrow\left\langle V,S\right\rangle \right)\le s\quad\forall\quad S\in\mathbf{S}_{\NC}.
          \]
    \item \label{enu:decision-faithfulness} Decision Faithfulness: For all $S_{\tau,\aux}\in\mathbf{S}_{\NC}$,
          consider the questions $Q'$ asked by $V$ (see \Eqref{Qprime}) in
          an execution of $\left\langle V,S_{\tau,\aux}\right\rangle $.
          \begin{enumerate}
            \item If $Q'$ asks all questions in some context, i.e. $Q'=C$ for some
                  $C\in\Call$, then the verifier outputs $\pred(\tau[C],C)$.
            \item Otherwise, the verifier outputs $1$.
          \end{enumerate}
    \item \label{enu:gcompleteness}$\G$-completeness: Given $\mathsf{qProver}$, the following must hold: 
          \begin{enumerate}
          \item \emph{Quantum completeness.} \\
          First, there should exist $\qstrat=(\ket{\psi}, \{O_1,O_2 \dots \})$ (as in \Defref{qstrat}) such that $\Pr\left[1\leftarrow\left\langle V,P_{\qstrat}\right\rangle \right]=c$
          where $\mathsf{qProver}(\qstrat)=:P_{\qstrat}$. \\
          Second, for every $P_{\qstrat}$, there is a quantum simulator, i.e.
          $S_{\qstrat,\aux}\in\mathbf{S}_{\qu}$, satisfying the following:
          the transcript $\transcript(V,S_{\qstrat,\aux})$ produced by the
          interaction of $S_{\qstrat,\aux}$ with $V$ is distributed identically
          to the transcript $\transcript(V,P_{\qstrat})$ produced by the interaction
          of $P_{\qstrat}$ with $V$. In particular, this implies 
          \[
          \exists\:S_{\qstrat,\aux}\in\mathbf{S}_{\qu},\ \text{s.t.}\ \ \Pr[1\leftarrow\left\langle V,S_{\qstrat,\aux}\right\rangle ]=\Pr\left[1\leftarrow\left\langle V,P_{\qstrat}\right\rangle \right].
          \]
          \item \emph{Classical completeness.} \\
          For every $\qstrat$ consisting of
          commuting observables, i.e. $[O_{i},O_{j}]=0$ for all $O_{i},O_{j}\in\mathbf{O}$,
          there is a PPT prover $P$ such that $\transcript(V,P)$ is distributed
          identically to $\transcript(V,P_{\qstrat})$.
          \end{enumerate}
  \end{enumerate}
\end{defn}

\subsection{Justification}
\atul{According to \Defref{OperationalTestOfContextuality}, if a proof of quantumness protocol $\cal{P}$ is faithful to a contextuality game $\G$, it must specify a mapping $\{\map_i\}_i$ relating it to $\G$. Now if a prover wins with probability
greater than $s$, from Condition \ref{enu:gsoundness} we know that
there is no way of interpreting the behaviour of this prover as being
consistent with a non-contextual assignment in $\G$ (via $\{\map_i\}_i$). }
Condition \ref{enu:decision-faithfulness},
ensures that the criterion for rewarding/penalising a prover is determined
solely by the predicate of $\G$ (via $\{\map_i\}_i$).  %
Finally, Condition \ref{enu:gcompleteness}
requires that there is a systematic procedure $\mathsf{qProver}$
for converting any quantum strategy in $\G$ to a prover in ${\cal P}$.
The marginal condition on $\mathsf{qProver}$ basically ensures that
the procedure does not artificially produce a different prover when
all observables commute vs when they do not. Condition (a) says that
$\mathsf{qProver}$ produces provers whose behaviour is consistent
with a quantum simulator implementing the same quantum strategy and
that the best strategy achieves the completeness value $c$. Condition
\ref{enu:gcompleteness} (b) ensures that when observables do commute
(i.e. we are in the non-contextual setting), the behaviour of the
prover produced by $\mathsf{qProver}$ can be understood in terms
of a non-contextual model, i.e. a PPT machine.

We now introduce these requirements sequentially, illustrating
why each one of them plays a crucial role in ruling out unsatisfactory
notions of operational tests of contextuality.
\begin{itemize}
  \item Condition \ref{enu:wellformed} is a very basic requirement. For instance,
        in \Exaref{simpleProofOfQuantumness}, one could assign map $\map^{\perp}$
        to every message in the protocol and this would satisfy \ref{enu:wellformed}.
    \item Condition \ref{enu:gsoundness} ensures the following: Suppose that the messages
        that map to $\perp$ are answered by using unbounded computational
        power. Even in this case, as long as the messages that correspond to questions/answers
        in $\G$ (as specified by the protocol ${\cal P}$) are answered non-contextually, the condition requires that one cannot do better than a PPT machine. \\For instance,
        \Exaref{simpleProofOfQuantumness} fails to satisfy this requirement:
        an unbounded classical simulator $S\in\mathbf{S}_{\NC}$ can easily
        produce a valid proof of quantumness and make the verifier accept
        with probability $\G.\valQu$ which is greater than the soundness
        value $s=\G.\valNC$. Clearly, \Exaref{simpleProofOfQuantumness}
        was a very simple construction. Consider the less trivial construction
        in \Exaref{q1q2poq} below where the verifier asks questions in the game
        $\G$, in addition to requesting a proof of quantumness,
        but only considers the prover's answers when the PoQ $\pi$ is
        valid. Using natural maps $\{\map_{i}\}$ (i.e.\ $\map_{1}=\map_{2}=\map_{4}=\map_{5}=\mathbb{I},\map_{3}=\map_{6}=\map^{\perp}$)
        it is immediate that Condition \ref{enu:gsoundness} is satisfied
        by this construction: even if a simulator $S\in\mathbf{S}_{\NC}$
        provides a valid proof of quantumness, the verifier does
        not accept with probability more than $\G.\valNC$. Yet, this construction
        is intuitively unsatisfactory as a test of contextuality because it
        is completely neglecting the answers $a_{1},a_{2}$ given by a classical
        prover.
\end{itemize}

\begin{example}
  \label{exa:q1q2poq}Given a contextuality game (with size-two contexts)
  $\G$, the verifier proceeds as described (where $\chall$ and $\poq$
  are as in \Exaref{simpleProofOfQuantumness}). This protocol satisfies Conditions \ref{enu:wellformed} and \ref{enu:gsoundness}, but fails to satisfy \ref{enu:decision-faithfulness} (a) (see \Defref{OperationalTestOfContextuality}).

  \begin{center}
    \begin{tabular}{|V{\linewidth}|c|c|}
      \multicolumn{1}{l}{Verifier}                                 & \multicolumn{1}{c}{}                                     & \multicolumn{1}{c}{Prover}\tabularnewline
      \cline{1-1}\cline{3-3}
      $(q_{1},q_{2})=C\leftarrow\Call$,$\chall\leftarrow\PoQ.\gen$ &                                                          & \tabularnewline
                                                                   & $\xrightarrow{(t_{1},t_{2},t_{3})=(q_{1},q_{2},\chall)}$ & \tabularnewline
      If $\poq$ is invalid, i.e. $\PoQ.\verify(\chall,\poq)=0$,

      $\qquad$accept with probability $\G.\valNC$

      else, i.e. $\PoQ.\verify(\chall,\poq)=1$

      $\qquad$output $\G.\pred(a_{1},a_{2},q_{1},q_{2})$           & $\xleftarrow{(t_{4},t_{5},t_{6})=(a_{1},a_{2},\poq)}$    & \tabularnewline
      \cline{1-1}\cline{3-3}
    \end{tabular}
    \par\end{center}
\end{example}

\begin{itemize}
  \item Condition \ref{enu:gcompleteness} is also satisfied by \Exaref{q1q2poq}
        and this illustrates why the last condition (below) is so crucial.
    \item Condition \ref{enu:decision-faithfulness} is where \Exaref{q1q2poq}
          finally fails: Consider two simulators $S_{\tau,\aux},S_{\tau,\aux'}\in\mathbf{S}_{\NC}$
          where the assignment $\tau$ corresponds to a game value $v<\G.\valNC$
          and $\aux$ corresponds to producing the correct proof of quantumness
          while $\aux'$ corresponds to producing an invalid proof
          of quantumness. Condition \ref{enu:decision-faithfulness}
          (a) requires that $\Pr[1\leftarrow\left\langle V,S_{\tau,\aux}\right\rangle ]=\Pr[1\leftarrow\left\langle V,S_{\tau,\aux'}\right\rangle ]$
          but, $\Pr[1\leftarrow\left\langle V,S_{\tau,\aux}\right\rangle ]=v$
          and $\Pr[1\leftarrow\left\langle V,S_{\tau,\aux'}\right\rangle ]=\G.\valNC$.
  \item The relevance of Condition \ref{enu:decision-faithfulness} (b) is
        evident from \Exaref{firstChallengeThenAsk} below, where the verifier
        starts by asking for a PoQ certificate and only if this is valid does
        it ask the questions for the contextuality game $\G$; otherwise it
        simply rejects. Intuitively, this is unsatisfactory because a classical
        prover is not even able to see the questions in $\G$, let alone answer them (unlike \Exaref{q1q2poq}).
        Yet, none of the previous conditions are violated.
        Such cases are excluded by Condition \ref{enu:decision-faithfulness}
        (b) because it requires that no prover can be penalised unless all
        questions in some context are asked.
\end{itemize}

\begin{example}
  \label{exa:firstChallengeThenAsk}Given a contextuality game (with
  size-two contexts) $\G$, the verifier proceeds as described below
  (where $\chall,\poq$ are as in \Exaref{simpleProofOfQuantumness}). It is easy to see that that this protocol satisfies Conditions~\ref{enu:wellformed} and \ref{enu:gsoundness} but not \ref{enu:decision-faithfulness} (b), with $s=\G.\valNC$ and $q=\G.\valQu$.
  \begin{center}
    \begin{tabular}{|V{\linewidth}|c|c|}
      \multicolumn{1}{l}{Verifier} & \multicolumn{1}{c}{}                                                                          & \multicolumn{1}{c}{Prover}\tabularnewline
      \cline{1-1}\cline{3-3}
      $\chall\leftarrow\PoQ.\gen$  &                                                                                               & \tabularnewline
                                   & \begin{cellvarwidth}[t]
                                       \centering
                                       $\xrightarrow{t_{1}=\chall}$

                                       $\xleftarrow{t_{2}=\poq}$
                                     \end{cellvarwidth}                                                                      & \tabularnewline
      $\ $                         &                                                                                               & $\ $\tabularnewline
      If $\poq$ is valid (i.e. $\PoQ.\verify(\chall,\poq)=1$)

      $\qquad$ask $(q_{1},q_{2})=C\leftarrow\Call$,

      $\qquad$receive $(a_{1},a_{2})$, and

      $\qquad$output $\pred(a_{1},a_{2},q_{1},q_{2})$.

      Else

      $\qquad$reject.              & \begin{cellvarwidth}[t]
                                       \centering
                                       If $\poq$ valid:

                                       $\xrightarrow{(t_{3},t_{4})=(q_{1},q_{2})}$

                                       $\xleftarrow{(t_{5},t_{6})=(a_{1},a_{2})}$

                                       $\xrightarrow{\mathsf{out}=\mathsf{pred}(a_{1},a_{2},q_{1},q_{2})}$

                                       otherwise:

                                       $\xrightarrow{\mathsf{out}=\mathsf{reject}}$
                                     \end{cellvarwidth} & \tabularnewline
      \cline{1-1}\cline{3-3}
    \end{tabular}
    \par\end{center}
  Here, we take the soundness value\footnote{Even though it is clear that every PPT prover succeeds with at most negligible probability; If one, in fact, takes $s=0$, then Condition \ref{enu:decision-faithfulness} (a) is already violated.} $s=\G.\valNC$ 
  while the quantum value is $q=\G.\valQu$.
\end{example}

\begin{itemize}
  \item All examples we have considered so far, have failed at least one of
  the criteria discussed above. However, it is easy to verify that \Exaref{q1q2poq}
  and \Exaref{firstChallengeThenAsk} above both satisfy Condition \ref{enu:gcompleteness}
  (a), using the natural choice for $\mathsf{qProver}$. We now consider
  an example, \Exaref{finalBoss} (below), that satisfies all conditions
  above Conditions \ref{enu:wellformed}, \ref{enu:gsoundness}, \ref{enu:decision-faithfulness}
  and Condition \ref{enu:gcompleteness} (a) and yet, it is unsatisfactory. 
  \begin{itemize}
  \item The idea in \Exaref{finalBoss} is to ``help'' the prover recover
  the post-measurement state, only if it produces a proof of quantumness
  certificate. More precisely, the verifier asks a QFHE encrypted question
  $\hat{q}_{1}$ together with a proof of quantumness challenge $\chall$.
  The prover is expected to respond with an encrypted answer $\hat{a}_{1}$
  together with a certificate for the proof of quantumness, $\poq$.
  The verifier then sends the second question $q_{2}$ together with
  a string $s$. If the prover provided a valid certificate $\poq$,
  the string $s$ is a key that allows the prover to decrypt any homomorphic
  computation performed using $\hat{q}_{1}$ (in particular, the prover
  can learn the post-measurement state; and also $q_{1}$ itself). Otherwise
  $s$ is some irrelevant random string. The verifier expects an answer
  $a_{2}$ and accepts if $\pred(a_{1},a_{2},q_{1},q_{2})=1$. 
  \item This seems unsatisfactory because the construction hinges entirely
  on the proof of quantumness: a prover who is given the ability to
  pass the proof of quantumness (as a black-box), but who is otherwise
  classical, is able to win the game with probability 1 (assuming $\QFHE.\Eval$
  is classical for classical circuits).
  \item Formally, it is easy to check that all conditions are satisfied (using
  natural maps) except Condition \ref{enu:gcompleteness} (b). However,
  there does not seem to be any procedure $\mathsf{qProver}$ that satisfies
  both quantum completeness and classical completeness simultaneously.
  For instance, the natural choice for $\mathsf{qProver}(\qstrat)=P_{\qstrat}$
  (recall $\qstrat=(\{O_{1},O_{2}\dots\},\left|\psi\right\rangle )$
  is simply to answer $\hat{q}_{1}$ by measuring $O_{q_{1}}$ on $\left|\psi\right\rangle $
  homomorphically, produce $\poq$ corresponding to $\chall$, then
  use the key sent by the verifier to recover the post-measurement state
  and measure $O_{q_{2}}$ on this state to obtain the answer $a_{2}$.
  Now, it is immediate that even when all observables commute, there
  is no PPT prover that can produce the same transcript as $P_{\qstrat}$
  (when interacting with the verifier)---a PPT prover cannot produce
  a valid proof of quantumness certificate. One might try to modify
  the $\mathsf{qProver}$ procedure so that it does not produce a proof
  of quantumness certificate $\poq$ when the observables commute, but
  the marginal requirement rules out any such attempt. Finally, if $\mathsf{qProver}$
  never produces a valid proof of quantumness, it cannot produce a $P_{\qstrat}$
  that makes the verifier accept with probability $c>s$. 
  \end{itemize}
  \end{itemize}
  \begin{example}
  \label{exa:finalBoss}Given a contextuality game (with size-two contexts)
  $\G$, the verifier proceeds as described below (where $\chall,\poq$
  are as in \Exaref{simpleProofOfQuantumness} while the QFHE scheme
  is as in \Defref{QFHEscheme}). 
  \begin{center}
  \begin{tabular}{|V{\linewidth}|c|c|}
  \multicolumn{1}{l}{Verifier} & \multicolumn{1}{c}{} & \multicolumn{1}{c}{Prover}\tabularnewline
  \cline{1-1}\cline{3-3}
  $(q_{1},q_{2})=C\leftarrow\Call$
  
  $\sk_{0},\sk_{1}\leftarrow\QFHE.\gen$
  
  $\hat{q}_{1}\leftarrow\QFHE.\Enc_{\sk_{1}}(q_{1})$
  
  $\chall\leftarrow\PoQ.\gen$ &  & \tabularnewline
   & \begin{cellvarwidth}[t]
  \centering
  $\xrightarrow{(t_{1},t_{2})=(\hat{q}_{1},\chall)}$
  
  $\xleftarrow{(t_{3},t_{4})=\left(\hat{a}_{1},\poq\right)}$
  \end{cellvarwidth} & \tabularnewline
  $b=\PoQ.\verify(\chall,\poq)$
  
  $a_{1}=\QFHE.\Dec_{\sk_{1}}(\hat{a}_{1})$ &  & \tabularnewline
   & \begin{cellvarwidth}[t]
  \centering
  $\xrightarrow{(t_{5},t_{6})=(q_{2},\sk_{b})}$
  
  $\xleftarrow{t_{7}=a_{2}}$
  \end{cellvarwidth} & \tabularnewline
  output $\pred(a_{1},a_{2},q_{1},q_{2})$ &  & \tabularnewline
  \cline{1-1}\cline{3-3}
  \end{tabular}
  \par\end{center}
  
  \end{example}

\paragraph{}

  We conclude this discussion by noting that in Condition~\ref{enu:gcompleteness}b, having computational indistinguishability of the transcripts does not suffice. To see this, observe that in \Exaref{finalBoss}, if messages $t_2,t_4,t_5,t_6, t_7$ are exchanged under homomorphic encryption (with an independent key), then it is straightforward to construct a PPT machine that produces a transcript that is computationally indistinguishable from that produced by $P_{\qstrat}$. In fact, one can apply such a transformation quite generically, rendering Condition~\ref{enu:gcompleteness}b irrelevant, if only computational indistinguishability is demanded. 

}

\section{OPad | Oblivious $\protect{\mathbf{U}}$-Pad}
\label{sec:opad}
{
  Before describing our compiler, we introduce the oblivious pad primitive, and show how to construct it.
}

\subsection{Definition\label{subsec:Definition_OPad}}

Let $\mathbf{U}:=\{U_{k}\}_{k\in K}$ be a set of unitaries acting on a %
Hilbert space $\cal{H}$, where $K$ is a finite
set.
\begin{defn}
  \label{def:Oblivious-U-pad}An Oblivious $\mathbf{U}$-Pad (or an
  OPad) is a tuple of algorithms $(\gen,\enc,\dec)$ as follows:
  \begin{itemize}
    \item $\gen$ is a PPT algorithm with the following syntax:
          \begin{itemize}
            \item Input: $1^{\lambda}$ (a security parameter in unary).
            \item Output: $(\pk,\sk)$.
          \end{itemize}
    \item $\Enc$ is A QPT algorithm with the following syntax:
          \begin{itemize}
            \item Input: $\pk,$ a state $\rho$ on $\cal H$.
            \item Output: a state $\sigma$ (also on $\cal H$) and a string $s$.
          \end{itemize}
    \item $\dec$ is a classical polynomial-time deterministic algorithm with the following syntax:
          \begin{itemize}
            \item Input: $\sk,s$.
            \item Output: $k\in K$.
          \end{itemize}
    \item $\samp$ is a PPT algorithm with the following syntax:
          \begin{itemize}
            \item Input: $\pk$
            \item Output: a string $s$
          \end{itemize}
  \end{itemize}
  We require the following.
  \begin{itemize}
    \item Correctness: There exists a negligible function $\negl$ such that, for all $\rho$ on $\mathcal{H}$, $\lambda \in \mathbb{N}$, the following holds with probability at least $1-\negl(\lambda)$ over sampling $(\pk,\sk)\leftarrow\gen(1^{\lambda})$, and $(\sigma, s) \leftarrow \enc(\pk, \rho)$:
         $$\sigma \approx_{\negl(\lambda)}U_{k}\rho U_{k}^{\dagger} \,,$$
        where $k = \dec(\sk, s)$.
    \item Soundness: For any PPT prover $P$, there exists a negligible function $\negl$ (not necessarily equal to the previous one) such that, for all $\lambda \in \mathbb{N}$, $P$ wins in the following game with probability at most  $1/2+\negl(\lambda)$.
  \begin{center}
    \begin{tabular}{|>{\centering}p{4cm}|>{\centering}p{3cm}|>{\raggedright}m{5cm}|}
      \multicolumn{1}{>{\centering}p{4cm}}{Prover} & \multicolumn{1}{>{\centering}p{3cm}}{} & \multicolumn{1}{>{\raggedright}m{5cm}}{Challenger}\tabularnewline
      \cline{1-1} \cline{3-3}
                                                   &                                        & \tabularnewline
                                                   &                                        & $(\pk,\sk)\leftarrow\gen(1^{\lambda})$\tabularnewline
                                                   & {\Large $\overset{\pk}{\longleftarrow}$}            & \tabularnewline

      (The PPT prover can proceed arbitrarily)      &                                        & \tabularnewline
                                                   & {\Large $\overset{s}{\longrightarrow}$}             & \tabularnewline
                                                   &                                        & $b\leftarrow\{0,1\}$

      $k_{0}\leftarrow K$

      $k_{1}=\dec(\sk,s)$                           \tabularnewline
                                                   & {\Large $\overset{k_{b}}{\longleftarrow}$}          & \tabularnewline
                       &                                        & \tabularnewline
                                                   & {\Large $\overset{b'}{\longrightarrow}$}            & \tabularnewline
                                                   &                                        & \text{Accept if } $b'=b$\tabularnewline
                                                   &                                        &               \tabularnewline
      \cline{1-1} \cline{3-3}
    \end{tabular}
    \par\end{center}
    \item Classical range sampling: $\samp$ can sample a string $s$ from the same distribution as $\enc$, i.e. for all $\rho \in \cal{H}$, and for all unbounded distinguishers D, it holds that 
    $$\left| \Pr[1 \leftarrow D(s): s \leftarrow \samp(\pk) ] - \Pr[1 \leftarrow D(s'): {(s',\sigma) \leftarrow \enc(\pk,\rho)}] \right| \le \negl(\lambda)$$ 
    where $\pk \leftarrow \gen(1^\lambda)$. 
  \end{itemize}
\end{defn}

\subsection{Relation between the OPad and proofs of quantumness}\label{subsec:OPadRelation}

\newblue
On a first read, one may skip to \Subsecref{instantiateOPad} where we instantiate the $\OPad$ in the random oracle model. Here, we briefly discuss the relation between the $\OPad$ and proofs of quantumness. We describe how the $\OPad$ implies a proof of quantumness and discuss how proofs of quantumness can be used to realise some properties of the $\OPad$.\footnote{The latter was pointed out by an anonymous QCrypt reviewer.}

\paragraph{OPad implies proof of quantumness.}

A 4-message proof of quantumness is immediate. Consider the case where
$\mathbf{U}=(\mathbb{I},\sigma_{x},\sigma_{y},\sigma_{z})$ indexed
by $K=(0,x,y,z)$ of size 4. Observe that the security game for the
OPad immediately yields a proof of quantumness. From the security
guarantee, there is no PPT algorithm that wins with probability non-negligibly
greater than $1/2$. Yet, there is a quantum prover that can distinguish
$k_{0}$ from $k_{1}$ with probability at least $3/4$. 
\begin{center}
\begin{tabular}{|V{\linewidth}|c|V{\linewidth}|}
\multicolumn{1}{l}{Prover} & \multicolumn{1}{c}{} & \multicolumn{1}{l}{Challenger}\tabularnewline
\cline{1-1}\cline{3-3}
 & $\xleftarrow{\pk}$ & $(\pk,\sk)\leftarrow\gen(1^{\lambda})$\tabularnewline
Start with $\left|1\right\rangle $, and

apply $\Enc$ on $\ket{1}$ to obtain $(U_{k}\left|1\right\rangle ,s)\leftarrow\Enc(\pk,\left|1\right\rangle )$ &  & \tabularnewline
 & $\overset{s}{\rightarrow}$ & \tabularnewline
 &  & $b\leftarrow\{0,1\}$

$k_{0}\leftarrow K$

$k_{1}=\dec(\pk,\sk,s)$\tabularnewline
 & $\overset{k_{b}}{\leftarrow}$ & \tabularnewline
Measure $U_{k_{b}}^{\dagger}U_{k}\left|1\right\rangle $ in the standard
basis.

$\qquad$

If the outcome is $\left|1\right\rangle $,

$\qquad$set $b'=1$.

Else

$\qquad$set $b'=0$.

 &  & \tabularnewline
 & $\overset{b'}{\rightarrow}$ & \tabularnewline
 &  & accept if $b'=b$\tabularnewline
\cline{1-1}\cline{3-3}
\end{tabular}
\par\end{center}

It is elementary to check that the prover succeeds with probability
$3/4$: when $b=1$, the prover always reports $b'=1$ but when $b=0$
the prover reports $b'=0$ with probability at least half (because
for any $B\in\text{U}$ there are two distinct $A,A'\in\mathbf{U}$
such that $\left\langle 0\right|AB\left|1\right\rangle =\left\langle 0\right|A'B\left|1\right\rangle =1$).

One can also construct a 2-message proof of quantumness protocol where
the challenger samples $(\pk,\sk)\leftarrow\gen(1^{\lambda})$, sends
$\pk$ and accepts if the prover returns $k'=\Dec(\sk,s)$.

\global\long\def\salt{\mathsf{pad}}%
\global\long\def\calS{\mathcal{S}}%
\global\long\def\pad{\mathsf{pad}}%
\global\long\def\QDec{\mathsf{QBreak}}%

\paragraph{Non-interactive proofs of quantumness imply a weakened variant of
OPad.}

Any non-interactive proof of quantumness protocol allows one
to construct a weakened variant of $\OPad$.\footnote{This observation is due
to an anonymous QCrypt referee.} The idea is best illustrated by considering
factoring. Let $\OPad=:(\gen,\Enc,\Dec)$ and proceed as follows:
\begin{itemize}
\item $\gen$: Generates a random large composite number $N$ and returns
$(\pk,\sk):=(N,N)$.
\item $\Enc(\pk,\left|\psi\right\rangle )$:
\begin{itemize}
\item Samples $k\leftarrow K$,
\item Obtains a factor (using Shor's algorithm) $F$ of $N$ and defines
$s:=(k,F)$
\item Returns $(U_{k}\left|\psi\right\rangle ,s)$
\end{itemize}
\item $\Dec(\sk,s )$: 
\begin{itemize}
\item Parses $s=(k,F)$
\item If $F$ is a factor of $N$ then

$\qquad$outputs $k$.

Otherwise

$\qquad$outputs $k'\leftarrow K$.
\end{itemize}
\end{itemize}
While correctness is immediate, this construction may seem unsound
since $\Enc$ is revealing $k$ in the clear. However, the point
is that no PPT machine can produce a ``valid encryption'' (i.e.
one that encodes a factor), and the decrypt procedure outputs a random
value from $K$ upon being given an ``invalid encryption''. Effectively,
this means that a PPT machine can only have the decrypt algorithm
output a random value from $K$ which is exactly what the soundness
condition demands. 

There are two issues with this construction. First, the $\Dec$ procedure
is randomised while the definition requires it to be deterministic.
But more importantly, second, there is no $\samp$ procedure that
can sample from the range of $\Enc$---no PPT procedure can produce
correct factors of $N$.

The first issue is not too hard to resolve (described below) but to
resolve both simultaneously, we use the random oracle model. 

\paragraph{Hardness of integer factoring implies OPad without efficient range
sampling.}

We outline the idea here. Let $\calS=(\calS.\gen,\calS.\enc,{\cal S}.\dec)$
denote a classically CCA-secure public key encryption such that (a)
${\cal S}.\Dec$ is a deterministic algorithm and (b) there is an
efficient quantum algorithm $\QDec$ that correctly decrypts any ciphertext
without needing the secret key. Such a scheme is known, based on integer
factoring \cite{HKS13}.
Suppose $(\pk',\sk')\leftarrow\calS.\gen(1^{\lambda})$. Then the
OPad is constructed as follows:
\begin{itemize}
\item $\OPad.\gen(1^{\lambda})$ 
\begin{itemize}
\item Runs $(\pk',\sk')\leftarrow{\cal S}.\gen(1^{\lambda})$
\item Samples $\pad\leftarrow K$
\item Computes $\pk:=(\pk',{\cal S}.\Enc_{\pk'}(\salt))$ and defines $\sk:=(\sk',\pad)$
\item Returns $(\pk,\sk)$
\end{itemize}
\item $\opad.\enc(\pk,\left|\psi\right\rangle )$ 
\begin{itemize}
\item Samples $\tilde{k}\leftarrow K$
\item Recovers $\pad$ from $\pk$ using $\QDec$
\item Defines $k=\salt\oplus\tilde{k}$
\item Applies $U_{k}$ to the input state $\left|\psi\right\rangle $
\item Computes $s\leftarrow\calS.\Enc_{\pk'}(k)$
\item Returns $(U_{k}\left|\psi\right\rangle ,s)$
\end{itemize}
\item $\opad.\Dec(\sk,s)$ returns $k:=\calS.\Dec_{\sk'}(s)\oplus\pad$.
\end{itemize}
Correctness is again immediate. Soundness also holds intuitively:
no PPT prover should be able to win the security game for the $\OPad$
as long as ${\cal S}$ is CCA secure. However, again, we do not have
efficient range sampling in this case.

\subsection{Instantiating the OPad | Oblivious Pauli Pad}\label{subsec:instantiateOPad}

{

\Algref{ObliviousPauliPad} below shows how to instantiate an oblivious $\mathbf{U}$-pad in the random oracle model, where $\mathbf{U}=\{X^{x}Z^{z}\}_{x,z}$ where $x,z$ are strings. We refer to this as an \emph{oblivious Pauli pad}. 
The instantiation leverages ideas from the \emph{proof of quantumness} protocol in \cite{BKVV20}. For convenience, we restate the informal description from \Subsecref{Contributions}. %
We describe the encryption procedure for a single qubit, and using a TCF pair $f_0,f_1$ (rather than an NTCF as in \Algref{ObliviousPauliPad}). To encrypt a multi-qubit state, one simply applies the following encryption procedure to each qubit.

We take $\pk = (f_0, f_1)$, and $\sk$ to be the corresponding trapdoor. Then, $\enc_{\pk}$ is as follows:
\begin{itemize}
\item[(i)] On input a qubit state $\ket{\psi} = \alpha \ket{0} + \beta \ket{1}$, evaluate $f_0$ and $f_1$ in superposition, controlled on the first qubit, and measure the output register. This results in some outcome $y$, and the leftover state $\alpha \ket{0} \ket{x_0} + \beta \ket{1}\ket{x_1}$, where $f(x_0) = f(x_1) = y$.
\item[(ii)] Compute the random oracle ``in the phase'', to obtain $(-1)^{H(x_0)}\alpha \ket{0} \ket{x_0} + (-1)^{H(x_1)} \beta \ket{1}\ket{x_1}$. Measure the second register in the Hadamard basis. This results in a string $d$, and the leftover qubit state
$$ \ket{\psi_Z} = Z^{d\cdot(x_0 \oplus x_1) + H(x_0) + H(x_1)} \ket{\psi} \,.$$
\item[(iii)] Repeat steps (i) and (ii) on $\ket{\psi_{Z}}$, but \emph{in the Hadamard basis}! This results in strings $y'$ and $d'$, as well as a leftover qubit state 
$$\ket{\psi_{XZ}} = X^{d'\cdot(x_0' \oplus x_1') + H(x_0') + H(x_1')}Z^{d\cdot(x_0 \oplus x_1) + H(x_0) + H(x_1)} \ket{\psi}\,,$$
where $x_0'$ and $x_1'$ are the pre-images of $y'$.
\end{itemize}
Notice that the leftover qubit state $\ket{\psi_{XZ}}$ is of the form $X^x Z^z \ket{\psi}$ where $x,z$ have the following two properties: (a) a verifier in possession of the TCF trapdoor can learn $z$ and $x$ given respectively $y, d$ and $y',d'$, and (b) no PPT prover can produce strings $y,d$ as well as predict the corresponding bit $z$ with non-negligible advantage (and similarly for $x$). Intuitively, this holds because a PPT prover that can predict $z$ with non-negligible advantage must be querying the random oracle at \emph{both} $x_0$ and $x_1$ with non-negligible probability. By simulating the random oracle (by lazy sampling, for instance), one can thus extract a claw $x_0, x_1$ with non-negligible probability, breaking the claw-free property.

%
%
%
%
%
%
%
%
%
%
%
%

%
%
%
%
%
%
%
%
%
%
%
%
%
%
%
%
%
%
%
}

\begin{algorithm}[H]
  Let
  \begin{itemize}
    \item ${\cal F}$ be an NTCF family (see \Defref{NTCF}). %
    \item $\left|\psi\right\rangle \in {\cal H}$ be a state acting on $J$ qubits.
    \item $H:\{0,1\}^{*}\to\{0,1\}$ denote the random oracle.
  \end{itemize}
  Define:
  \begin{itemize}
    \item $\gen(1^{\lambda})$:
          \begin{itemize}
            \item Execute $(\pk,\sk)\leftarrow\calF.\gen(1^{\lambda})$ and output $(\pk,\sk)$.
          \end{itemize}
    \item $\enc(\pk,\text{\ensuremath{\rho}})$:
          \begin{itemize}
            \item For each $j\in\{1\dots J\}$,
                \begin{itemize}
                    \item[ ] execute $\qubitEnc$ (as described in \Algref{Oblivious-Pauli-Pad-procedure}) on qubit $j$ of $\rho$ using the public key $\pk$, and
                    \item[ ] use $(k_j,s_j)$ to denote $(k',s')$ as in \Eqref{kPrime}.%
                \end{itemize}
            \item Let $k=(k_{1}\dots k_{J})$, and $s=(s_{1}\dots s_{J})$. Denote the resulting state by $\sigma_{k}$.
                \item Return $\sigma_{k}$ and $s$.
          \end{itemize}
    \item $\dec(\sk,s)$:
          \begin{itemize}
            \item For each $j\in \{1\dots J\}$ 
            \begin{itemize}
              \item[ ] denote by $k_j$ the output of $\qubitDec(\sk, s_j)$ (as described in \Algref{qubitDec}).\footnotemark
            \end{itemize}
            \item Return $k=(k_{1}\dots k_{j})$.
          \end{itemize}
    \item $\samp(\pk)$:
          \begin{itemize}
            \item For each $j\in \{1\dots J\}$, 
            \begin{itemize}
                \item[ ] sample $s_j \leftarrow \sSamp(\pk) $ (as described in \Algref{sSamp}).
            \end{itemize}
            \item Return $s=(s_1\dots s_J)$. 
          \end{itemize}
  \end{itemize}
  \caption{Oblivious Pauli Pad \label{alg:ObliviousPauliPad}}
\end{algorithm}

\footnotetext{Note that $\qubitDec$ does not output a qubit, but just a classical string! Nonetheless, we choose to call it $\qubitDec$, since the output are the quantum one-time pad keys associated to the $\qubitEnc$ procedure.}

\begin{algorithm}[H]
~\\  $\qubitEnc(\pk,\left|\psi\right\rangle )$ where $\left|\psi\right\rangle $
  is a single qubit state (extended to mixed states by linearity)%
  \begin{itemize}
    \item Without loss of generality, suppose the state of the qubit is given
          by $\left|\psi\right\rangle =\alpha\left|0\right\rangle +\beta\left|1\right\rangle $.
    \item Proceed as follows
          \begin{align}
             & \ket{\psi} = \alpha\left|0\right\rangle +\beta\left|1\right\rangle \nonumber                                                                                                                                                                                                                                                                               \\
             & \mapsto\sum_{x\in{\cal X},y\in\calY}\alpha\sqrt{(\calF.f'_{\pk,0}(x))(y)}\left|0\right\rangle \left|x\right\rangle \left|y\right\rangle +\beta\sqrt{({\cal F}.f'_{\pk,1}(x))(y)}\left|1\right\rangle \left|x\right\rangle \left|y\right\rangle               & \text{{Using \ensuremath{\calF.\samp}}}\nonumber                                 \\
             & \approx_{\epsilon}\sum_{x\in{\cal X},y\in\calY}\alpha\sqrt{({\cal F}.f{}_{\pk,0}(x))(y)}\left|0\right\rangle \left|x\right\rangle \left|y\right\rangle +\beta\sqrt{({\cal F}.f_{\pk,1}(x))(y)}\left|1\right\rangle \left|x\right\rangle \left|y\right\rangle & \nonumber \text{where }\epsilon\text{ is negligible}\\ %
             & \mapsto\left(\alpha\left|0\right\rangle \left|x_{0}\right\rangle +\beta\left|1\right\rangle \left|x_{1}\right\rangle \right)\ \left|y\right\rangle                                                                                                           & \text{{ Measuring the last register to get some $y$}}, \nonumber \\ 
             & & \text{where \ensuremath{y} is s.t.}\nonumber                                   \\
             &                                                                                                                                                                                                                                                              & \text{ \ensuremath{\calF.f_{\pk,0}(x_{0})=\calF.f_{\pk,1}(x_{1})=y}.}\nonumber \\
             & \mapsto\left((-1)^{H(x_{0})}\alpha\left|0\right\rangle \left|x_{0}\right\rangle +(-1)^{H(x_{1})}\beta\left|1\right\rangle \left|x_{1}\right\rangle \right)\ \left|y\right\rangle                                                                             & \text{Applying the random oracle}\nonumber                                     \\
             &                                                                                                                                                                                                                                                              & \text{ \ensuremath{H} in the phase.}\nonumber                                  \\
             & \mapsto\left((-1)^{d\cdot x_{0}+H(x_{0})}\alpha\left|0\right\rangle +(-1)^{d\cdot x_{1}+H(x_{1})}\beta\left|1\right\rangle \right)\ \left|d\right\rangle \left|y\right\rangle                                                                                & \text{Measuring the second register}\nonumber                                  \\
             &                                                                                                                                                                                                                                                              & \text{ in the Hadamard basis to get some $d$.}\nonumber                                        \\
             & =\underbrace{Z^{d\cdot(x_{0}\oplus x_{1})+H(x_{0})+H(x_{1})}\left|\psi\right\rangle }_{=:\left|\phi\right\rangle }\left|d\right\rangle \left|y\right\rangle                                                                                                  & \text{Up to a global phase}.\label{eq:postMeasuredState}
          \end{align}
    \item Relabel $(d,y,x_{0},x_{1})$ to $(d_{Z},y_{Z},x_{Z,0},x_{Z,1})$.
    \item Denote by $\left|\phi\right\rangle =\alpha'\left|+\right\rangle +\beta'\left|-\right\rangle $
          the state of the first qubit in \Eqref{postMeasuredState}. Continue as follows
          \begin{align}
             & \alpha'\left|+\right\rangle +\beta'\left|-\right\rangle \label{eq:startHere}                                                                                                                                                                                                                                                \\
             & \mapsto\sum_{x\in{\cal X},y\in\calY}\alpha\sqrt{(\calF.f'_{\pk,0}(x))(y)}\left|+\right\rangle \left|x\right\rangle \left|y\right\rangle +\beta\sqrt{({\cal F}.f'_{\pk,1}(x))(y)}\left|-\right\rangle \left|x\right\rangle \left|y\right\rangle & \text{By applying \ensuremath{\samp}. }\nonumber                 \\
             & \vdots\nonumber                                                                                                                                                                                                                                                                                                  \\
             & \text{proceed as above to obtain}\nonumber                                                                                                                                                                                                                                                                       \\
             & \vdots\nonumber                                                                                                                                                                                                                                                                                                  \\
             & =X^{d\cdot(x_{0}\oplus x_{1})+H(x_{0})+H(x_{1})}\left|\phi\right\rangle \left|d\right\rangle \left|y\right\rangle                                                                                                                              & \text{Up to a global phase.}\label{eq:obliviousPauliPad_second}
          \end{align}
    \item Relabel $(d,y,x_{0},x_{1})$ to $(d_{X},y_{X},x_{X,0},x_{X,1})$.
    \item Now note that the final state in \Eqref{obliviousPauliPad_second} can be written as  
          \begin{equation}
          X^{\phase(d_{X},x_{X,0},x_{X,1})}Z^{\phase(d_{Z},x_{Z,0},x_{Z,1})} \ket{\psi}  \label{eq:correctnessOpad}
          \end{equation}
    where ${\rm \phase}(d,x_{0},x_{1}):=d\cdot(x_{0}\oplus x_{1})+H(x_{0})+H(x_{1})$.
    \item Define 
        \begin{align}
            k'&:=({\phase(d_{X},x_{X,0},x_{X,1})},{\phase(d_{Y},x_{Y,0},x_{Y,1})}),\, \text{and}\, s':=(d_{X},y_{X},\ d_{Z},y_{Z}).  \label{eq:kPrime} %
        \end{align} 
    \item Return the state in \Eqref{correctnessOpad}, and $s'$.
  \end{itemize}
  \caption{$\mathsf{qubitEnc}$ \label{alg:Oblivious-Pauli-Pad-procedure}} %
\end{algorithm}

\begin{algorithm}[H]
  $\qubitDec(\sk, \underbrace{(d_X,y_X,\, d_Z,y_Z)}_{s'})$
  \begin{itemize}
    \item From $y_{X}$, compute $x_{X,0}=\calF.\inv(\sk,0,y_{X})$ and $x_{X,1}=\calF.\inv(\sk,1,y_{X})$.
    \item Similarly, from $y_{Z}$, compute $x_{Z,0}=\calF.\inv(\sk,0,y_{Z}),x_{Z,1}=\calF.\inv(\sk,1,y_{Z})$.
    \item Compute $k'$ as in \Eqref{kPrime}, i.e.
      \begin{equation}
        k'=({\phase(d_{X},x_{X,0},x_{X,1})},{\phase(d_{Y},x_{Y,0},x_{Y,1})}) \label{eq:qubitDecOpad}
      \end{equation}
    where 
    ${\rm \phase}(d,x_{0},x_{1}):=d\cdot(x_{0}\oplus x_{1})+H(x_{0})+H(x_{1})$
    \item Return $k'$.
  \end{itemize}
  \caption{$\qubitDec$ \label{alg:qubitDec}}
\end{algorithm}

\begin{algorithm}[H]
  $\sSamp(\pk)$
  \begin{itemize}
    \item Sample $x_X\leftarrow \calX$ and $x_Y \leftarrow \calX$ (where $\calX$ is specified by the NTCF $\calF$ and the security parameter). 
    \item Evaluate $y_X = \calF.f_{\pk,0}(x_X)$ and $y_Y = \calF.f_{\pk,0}(x_Y)$.
    \item Sample $d_X\leftarrow \calX \setminus \{0\}$ and $d_Y \leftarrow \calX \setminus \{0\}$.
    \item Return $(d_X,y_X,d_Z,y_Z)$.
  \end{itemize}
  \caption{$\sSamp$ \label{alg:sSamp}}
\end{algorithm}

\begin{lem}[Correctness and Efficient Range Sampling]
  Suppose ${\cal F}$ is an NTCF as in \Defref{NTCF}. Then, \Algref{ObliviousPauliPad} satisfies both the correctness and the efficienc range sampling requirements
  in \Defref{Oblivious-U-pad}.
 \end{lem}

\begin{proof}
This is straightforward to verify given the properties of the NTCF.
\end{proof}

\begin{lem}[Soundness]
  Suppose ${\cal F}$ is an NTCF as in \Defref{NTCF}. Then, \Algref{ObliviousPauliPad}
  satisfies the soundness requirement in \Defref{Oblivious-U-pad}.
\end{lem}

{\begin{proof}
    We first analyse the simpler game ${\cal G}$ (described below) and
    then observe that the reasoning carries over to the soundness of \Algref{ObliviousPauliPad}.
    \begin{center}
      \begin{tabular}{|>{\centering}p{4cm}|>{\centering}p{3cm}|>{\raggedright}m{5cm}|}
        \multicolumn{1}{>{\centering}p{4cm}}{Prover $(\calP)$} & \multicolumn{1}{>{\centering}p{3cm}}{} & \multicolumn{1}{>{\raggedright}m{5cm}}{Challenger $(\calC)$}\tabularnewline
        \cline{1-1} \cline{3-3}
                                                               &                                        & \tabularnewline
                                                               &                                        & $(\pk,\sk)\leftarrow\gen(1^{\lambda})$\tabularnewline
                                                               & {\Large $\overset{\pk}{\longleftarrow}$}            & \tabularnewline
                                                               &                                        & \tabularnewline
                                                               & {\Large $\overset{y}{\longrightarrow}$}             & \tabularnewline
                                                               &                                        & $b\leftarrow\{0,1\}$                                                        \\
        $h_{0}\leftarrow\{0,1\}$                                                                                                                                                      \\
        $h_{1}=H(x_{0})\oplus H(x_{1})$ where $x_{0}=\calF.\inv(\sk,0,y)$
        and $x_{1}=\calF.\inv(\sk,1,y)$.\tabularnewline
                                                               & {\Large $\overset{h_{b}}{\longleftarrow}$}          & \tabularnewline
                                                               &                                        & \tabularnewline
                                                               & {\Large $\overset{b'}{\longrightarrow}$}            & \tabularnewline
                                                               &                                        & Accept if $b'=b$\tabularnewline
                                                               &                                                    & \tabularnewline
        \cline{1-1} \cline{3-3}
      \end{tabular}
      \par\end{center}

    Intuitively, it is clear that no PPT prover ${\cal P}$ wins with
    probability more than $1/2+{\rm negl}$ because if a PPT prover can
    distinguish $h_{0}$ from $h_{1}$, it must know $H$ at both $x_{0}$
    and $x_{1}$ with non-negligible probability. By simulating the random
    oracle, one can then construct a PPT algorithm to extract preimages
    $x_{0},x_{1}$ of $y$. This violates the claw-free property of $\calF$.

    Formally, suppose that ${\cal P}$ succeeds with probability at most
    $1/2+\eta$ for some non-negligible function $\eta$. We show below
    that the following straightforward reduction extracts preimages $x_{0},x_{1}$
    of $y$ with non-negligible probability:
    \begin{center}
      \begin{tabular}{|>{\centering}p{1in}|>{\centering}p{1cm}|>{\raggedright}m{5cm}|>{\raggedright}m{1.5cm}|>{\raggedright}m{4cm}|}
        \multicolumn{1}{>{\centering}p{1in}}{Prover $(\calP)$, queries to $H$ are simulated by ${\cal R}$} & \multicolumn{1}{>{\centering}p{1cm}}{\centering{}} & \multicolumn{1}{>{\raggedright}m{5cm}}{Reduction ${\cal R}$ that uses ${\cal P}$} & \multicolumn{1}{>{\raggedright}m{1.5cm}}{\centering{}} & \multicolumn{1}{>{\raggedright}m{4cm}}{Challenger $\calF.\calC$, for the claw-free property of $\calF$}\tabularnewline
        \cline{1-1} \cline{3-3} \cline{5-5}
                                                                                                            &   &   &  & \tabularnewline
                                                                                                           & \centering{}                                       & Answer queries to $H$ made by ${\cal P}$ by maintaining a database
        $D$. If the query to $H$ is at $x$, check if $D$ contains $x$,
        otherwise sample ${\rm val}\leftarrow\{0,1\}$, respond with ${\rm val}$
        and append $(x,{\rm val})$ to $D$.                                                                 & \centering{}                                       & \centering{}{$(\pk,\sk)\leftarrow\gen(1^{\lambda})$}\tabularnewline
                                                                                                           & \centering{}{\Large $\overset{\pk}{\longleftarrow}$}            &                                                                                   & \centering{}{\Large $\overset{\pk}{\longleftarrow}$}                & \tabularnewline
                                                                                                           & \centering{}                                       &                                                                                   & \centering{}                                           & \tabularnewline
                                                                                                           & \centering{}{\Large $\overset{y}{\longrightarrow}$}             &                                                                                   & \centering{}                                           & \tabularnewline
                                                                                                           & \centering{}                                       & $h_{0}\leftarrow\{0,1\}$                                                          & \centering{}                                           & \tabularnewline
                                                                                                           & \centering{}{\Large $\overset{h_{0}}{\longleftarrow}$}          &                                                                                   & \centering{}                                           & \tabularnewline
                                                                                                           & \centering{}                                       &                                                                                   & \centering{}                                           & \tabularnewline
                                                                                                           & \centering{}{\Large $\overset{b'}{\longrightarrow}$}  & & & \tabularnewline
                                                                                                           & & Ignore $b'$.                                                                      & \centering{}                                           & \tabularnewline
                                                                                                           & \centering{}                                       & Check if $D$ has two inputs $(x_{0},x_{1})$ with the same output
        $y$ under ${\cal F}$ (i.e. ${\cal F}.\chk(\pk,0,x_{0},y)=$

        $\calF.\chk(\pk,1,x_{1},y)=1$).

        Otherwise, set $x_{0}=x_{1}=0$.                                                                    & \centering{}                                       & \tabularnewline
                                                                                                           & \centering{}                                       &                                                                                   & \centering{}{\Large $\xrightarrow{(x_{0},x_{1},y)}$}   & \tabularnewline
        \cline{1-1} \cline{3-3} \cline{5-5}
      \end{tabular}
      \par\end{center}

    We lower bound the success probability of ${\cal R}$. Assume, without
    loss of generality, that $\calP$ does not repeat queries. Let $E$
    be the following event, during the execution of $\calP$, interacting
    with $\calC$: $\calP$ makes a query at $x\in\{x_{0},x_{1}\}$ such
    that after this query, both $x_{0}$ and $x_{1}$ have been queried.

    Note that $\neg E$ means that, by an information theoretic argument,
    $\calP$ cannot distinguish the random variable $H(x_{0})\oplus H(x_{1})$
    from a uniformly random bit.

    Suppose $\calP$ interacts with ${\cal C}$. Then,
    \[
      \Pr[\calP\text{ wins}]=\Pr[\neg E]\Pr[\calP\text{ wins}|\neg E]+\Pr[E]\Pr[\calP\text{ wins}|E].
    \]
    NB: $\Pr[b=0|\neg E]=\Pr[b=1|\neg E]=\frac{1}{2}$ because $\Pr[b=0|\neg E]=\Pr[\neg E|b=0]\Pr[b=0]/\Pr[\neg E]$
    and $\Pr[\neg E|b=0]=\Pr[\neg E|b=1]$ because until $E$ happens,
    $b=0$ and $b=1$ cannot make any difference in the execution of the
    protocol.

    Continuing,
    \begin{align*}
      \Pr[\calP\text{ wins}] & =\Pr[\neg E]\left(\Pr[b'=0|b=0\land\neg E]\cancelto{1/2}{\Pr[b=0|\neg E]}+\underbrace{\Pr[b'=1|b=1\land\neg E]}_{=\Pr[b'=1|b=0\land\neg E]}\cancelto{1/2}{\Pr[b=1|\neg E]}\right)+ \\
                             & \ \ \Pr[E]\Pr[\calP\text{ wins}|E]                                                                                                                                                 \\
                             & =\Pr[\neg E]\cdot\frac{1}{2}+\Pr[E]\cdot\Pr[\calP\text{ wins}|E]                                                                                                                   \\
                             & =\frac{1}{2}+\Pr[E]\cdot\left(\Pr[\calP\text{ wins}|E]-\frac{1}{2}\right)
    \end{align*}
    and since $\Pr[\calP\text{ wins}]\ge\frac{1}{2}+\eta$, it implies
    that $\Pr[E]\ge\eta'$ for some other non-negligible function $\eta'$.
    Since $\Pr[{\cal R}\text{ wins}]=\Pr[E]$, we conclude ${\cal R}$,
    a PPT algorithm, wins against $\calF.\calC$ with probability $\eta$
    which contradicts the claw-free property of ${\cal F}$.

    To use this result, we first need to generalise to the case of multiple
    $y$s. This is also quite straightforward. Consider the following
    game ${\cal G}'$: \\

    \begin{tabular}{|>{\centering}p{4cm}|>{\centering}p{3cm}|>{\raggedright}m{5cm}|}
      \multicolumn{1}{>{\centering}p{4cm}}{Prover $(\calP)$} & \multicolumn{1}{>{\centering}p{3cm}}{}     & \multicolumn{1}{>{\raggedright}m{5cm}}{Challenger $(\calC)$}\tabularnewline
      \cline{1-1} \cline{3-3}
                                                            &   &   \tabularnewline
                                                             &                                            & $(\pk,\sk)\leftarrow\gen(1^{\lambda})$\tabularnewline
                                                             & {\Large $\overset{\pk}{\longleftarrow}$}                & \tabularnewline
                                                             &                                            & \tabularnewline
                                                             & {\Large $\xrightarrow{y_{1}\dots y_{2J}}$} & \tabularnewline
                                                             &                                            & $b\leftarrow\{0,1\}$                                                        \\
      $h_{0}\leftarrow\{0,1\}^{2J}$

      $h_{1}=h_{11},h_{21}\dots h_{2J,1}$ where $h_{j1}=H(x_{j0})\oplus H(x_{j1})$
      where in turn

      $x_{j,0}=\calF.\inv(\sk,0,y_{j})$ and $x_{j,1}=\calF.\inv(\sk,1,y_{j})$.\tabularnewline
                                                             & {\Large $\overset{h_{b}}{\longleftarrow}$}              & \tabularnewline
                                                             &                                            & \tabularnewline
                                                             & {\Large $\overset{b'}{\longrightarrow}$}                & \tabularnewline
                                                             &                                            & Accept if $b'=b$\tabularnewline
                                                             &   &   \tabularnewline
      \cline{1-1} \cline{3-3}
    \end{tabular}
    \newline
    
    Also consider the modified reduction ${\cal R}'$ which is the same
    as the reduction ${\cal R}$, except that it receives $y_{1}\dots y_{2J}$
    and at the last step, where it checks $D$ for any two inputs corresponding
    to any one of the $y_{1}\dots y_{2J}$. The success probability of
    ${\cal R}'$ can be lower bounded as before, except that the event
    $E$ now becomes the following: ${\cal P}$ makes a query $x\in\{x_{j,0},x_{j,1}\}_{j\in[2J]}$
    such that all for at least some $j$, both $x_{j,0}$ and $x_{j,1}$
    have been queried.

    The last step is to relate the game ${\cal G}'$ with the soundness
    game in \Defref{Oblivious-U-pad}. The main difference between the
    two is that instead of $h_{0}$ and $h_{1}$, the soundness game uses
    $k_{0}$ and $k_{1}$ where, for $\dec$ given by the oblivious pauli
    pad (i.e. \Algref{ObliviousPauliPad}), $k_{1}$ has the form (for
    $j$ odd)
    \[
      k_{j1}=(d_{j}\cdot(x_{j0}\oplus x_{j1})+H(x_{j,0})+H(x_{j,1}),\ d_{j+1}\cdot(x_{j+1,0}\oplus x_{j+1,1})+H(x_{j+1,0})+H(x_{j+1,1})).
    \]
    The reasoning in computing the probability of $E$ goes through as
    above. The reduction ${\cal R}'$ remains unchanged (it anyway
    was not computing $k_{1}$). Assuming that the claw-free property
    of ${\cal F}$ holds, we conclude that \Algref{ObliviousPauliPad}
    satisfies the soundness condition. %
  \end{proof}
}

In \Secref{construction11}, we show how to use the oblivious pad, along with a $\QFHE$ scheme, to obtain a contextuality compiler. In order for this to be possible, the oblivious pad needs to be ``compatible'' with the $\QFHE$ scheme in the following sense.

\begin{defn}[$\opad$ compatible with $\QFHE$] \label{def:QFHEcomptableopad}
  Suppose a $\QFHE$ scheme satisfies \Defref{QFHEscheme} with the form of encryption of $n$-qubit states specified by $\{ U_k \}_{k\in K}$ as in \Eqref{MBform}. Then, an oblivious $\mathbf{U}$-Pad (or an $\opad$) as in \Defref{Oblivious-U-pad} is compatible with the $\QFHE$ scheme if $\mathbf{U} = \{ U_k \}_{k\in K}$.
\end{defn}

While for our compiler, taking $\mathbf U$ to be the Pauli group suffices, we give a more general construction below.

\subsection{Instantiating the OPad | General Oblivious $\mathbf{U}$-Pad}

{We informally describe how one can instantiate the oblivious ${\mathbf U}$-Pad
more generically for essentially any group of unitaries ${\mathbf U}$. 

To describe the idea, for simplicity, let ${\mathbf V}=\{\mathbb{I},X\}$
and suppose the input state is a single qubit state $\left|\psi\right\rangle $.
In the Pauli Pad construction, we used TCFs to apply the unitary
$V_{k}:=X^{d\cdot(x_{0}\oplus x_{1})+H(x_{0})+H(x_{1})}\in{\mathbf V}$
to $\left|\psi\right\rangle $ where $k=d\cdot(x_{0}\oplus x_{1})+H(x_{0})+H(x_{1})$.
We had defined $s:=(d,y)$ where $y$ is the image of $x_{0}$ and
$x_{1}$. The $\Enc$ procedure returned $(V_{k}\left|\psi\right\rangle ,s)$.
It was, however, not immediate how one would extend this construction
to unitaries beyond the Pauli group.

Now, in order to apply a general unitary $U_{k}\in{\mathbf U}$ (which
does not necessarily have to be a Pauli matrix), we use TCFs to generate
a pair $(k,s)$ as above but on auxiliary qubits, and then apply the
unitary $U_{k}$ on the input state $\left|\psi\right\rangle $. Intuitively,
since a PPT algorithm producing $(k,s)$ entails that it must know
claws for $y$ (as it must have evaluated $H(x_{0})$ and $H(x_{1})$),
we conclude that this construction satisfies the soundness condition
for an OPad. 

}

More precisely, suppose the unitary group ${\mathbf U}=\{U_{k}\}_{k}$
is indexed by $J$-bit strings $k$ and each unitary acts on a Hilbert
space $\mathcal{H}$. The corresponding OPad may be instantiated as in \Algref{ObliviousPadGeneralU}.

\global\long\def\Qsample{\mathsf{qSamp}}%
\global\long\def\Qsamp{\mathsf{qSamp}}%

\global\long\def\Csample{\mathsf{cSamp}}%

\begin{algorithm}[H]
Let
\begin{itemize}
\item ${\cal F}$ be an NTCF family (see \Defref{NTCF}),
\item $\epsilon>0$ be a fixed small constant, and
\item $\left|\psi\right\rangle \in \mathcal{H}$
\item $H:\{0,1\}^{*}\to\{0,1\}$ denote the random oracle.
\end{itemize}
Define:
\begin{itemize}
\item $\gen(1^{\lambda})$:
\begin{itemize}
\item Execute $(\pk,\sk)\leftarrow\calF.\gen(1^{\lambda})$ and output $(\pk,\sk)$.
\end{itemize}
\item $\enc(\pk,\text{\ensuremath{\rho}})$:
\begin{itemize}
\item For each $j\in\{1\dots J\}$, execute $\Qsamp_{j}(\pk)$ (as described in
\Algref{enc_j_general}) to obtain $k=(k_{1}\dots k_{J})$ and $s=(s_{1}\dots s_{J})$. 
\item Return $\sigma_{k}:=U_{k}\rho U_{k}^{\dagger}$ along with $s$.
\end{itemize}
\item $\dec(\sk,s)$:
\begin{itemize}
\item For each $j\in\{1\dots J\}$, 
\begin{itemize}
\item Parse $s_{j}=:(d,y)$.
\item From $y$, compute $x_{0}=\inv(\sk,0,y)$ and $x_{1}=\inv(\sk,1,y)$.
\item Compute 
\[
k_{j}(s_{j})=\phase(d,x_{0},x_{1})
\]
where ${\rm \phase}(d,x_{0},x_{1}):=d\cdot(x_{0}\oplus x_{1})+H(x_{0})+H(x_{1})$
as before.
\end{itemize}
\item Return $k=(k_{1}\dots k_{j})$.
\end{itemize}
\item $\samp(\pk)$:
\begin{itemize}
  \item For each $j\in \{1\dots J\}$, 
  \begin{itemize}
      \item[ ] sample $s_j \leftarrow \sSamp(\pk) $ (as described in \Algref{sSamp_}).
  \end{itemize}
  \item Return $s=(s_1\dots s_J)$. 
\end{itemize}\end{itemize}
\caption{Oblivious ${\mathbf U}$-Pad for a general ${\mathbf U}$.\protect\label{alg:ObliviousPadGeneralU}}

\end{algorithm}

\begin{algorithm}[H]
\begin{raggedright}
$\Qsamp_{j}(\pk)$ 
\par\end{raggedright}
\begin{raggedright}
Apply the steps listed in \Algref{Oblivious-Pauli-Pad-procedure}
starting with \Eqref{startHere} and ending at \Eqref{obliviousPauliPad_second}
to the state $\ket{\phi} = \left|0\right\rangle $ to obtain the state 
\[
X^{d\cdot(x_{0}\oplus x_{1})+H(x_{0})+H(x_{1})}\left|0\right\rangle \left|d\right\rangle \left|y\right\rangle =\underbrace{\left|d\cdot(x_{0}\oplus x_{1})+H(x_{0})+H(x_{1})\right\rangle }_{k_{j}}\left|d\right\rangle \left|y\right\rangle .
\]
Return $(k_{j},s_{j})$ where $s_{j}:=(d,y)$.
\par\end{raggedright}
\caption{\protect\label{alg:enc_j_general}}
\end{algorithm}

\begin{algorithm}[H]
  $\sSamp(\pk)$
  \begin{itemize}
    \item Sample $x\leftarrow \calX$ (where $\calX$ is specified by the NTCF $\calF$ and the security parameter). 
    \item Evaluate $y = \calF.f_{\pk,0}(x)$.
    \item Sample $d\leftarrow \calX \setminus \{0\}$.
    \item Return $(d,y)$.
  \end{itemize}
  \caption{$\sSamp$ \label{alg:sSamp_}}
\end{algorithm}

In light of this, one can also consider the following primitive which
implies an oblivious ${\mathbf U}$-pad using the ideas above. 
\begin{defn}[PoQ with efficient range sampling]
 A Proof of Quantumness with efficient range sampling is given by
four algorithms $(\gen,\Qsample,\samp,\Dec)$ where $\Qsample$
is QPT while $\gen,\samp,\Dec$ are PPT. Let $(\pk,\sk)\leftarrow \gen(1^\lambda)$ and require that the following conditions hold.
\begin{itemize}
\item Correctness. $k=\Dec(\sk,s)$ for all $(k,s)\leftarrow\Enc(\pk)$.
\end{itemize}
\begin{center}
\begin{tabular}{|>{\centering}p{4cm}|>{\centering}p{3cm}|>{\raggedright}m{5cm}|}
\multicolumn{1}{>{\centering}p{4cm}}{Prover} & \multicolumn{1}{>{\centering}p{3cm}}{} & \multicolumn{1}{>{\raggedright}m{5cm}}{Challenger}\tabularnewline
\cline{1-1}\cline{3-3}
& & \tabularnewline
 &  & $(\pk,\sk)\leftarrow\gen(1^{\lambda})$\tabularnewline
 & $\overset{\pk}{\leftarrow}$ & \tabularnewline
$(k,s)\leftarrow\Qsample(\pk)$ &  & \tabularnewline
 & $\overset{s}{\rightarrow}$ & \tabularnewline
 &  & $b\leftarrow\{0,1\}$

$k_{0}\leftarrow K$

$k_{1}=\dec(\sk,s)$\tabularnewline
 & $\overset{k_{b}}{\leftarrow}$ & \tabularnewline
$k_{b}=k$, set $b'=1$,

$k_{b}\neq k$, set $b'=0$ &  & \tabularnewline
 & $\overset{b'}{\rightarrow}$ & \tabularnewline
 &  & accept if $b'=b$\tabularnewline
 & & \tabularnewline
\cline{1-1}\cline{3-3}
\end{tabular}
\par\end{center}
\begin{itemize}
\item Soundness. No PPT algorithm wins the security game above with probability
more than $1/2+\negl(1/2)$.
\item Efficient range sampling condition. The distributions of $s$ produced
by $(k,s)\leftarrow\Qsample(\pk)$ and $s\leftarrow\samp(\pk)$
are identical.
\end{itemize}
\end{defn}

\pagebreak{}

\section{Construction of the $(1,1)$ compiler}\label{sec:construction11}
{
  With the Oblivious $\mathbf{U}$-Pad in place, we are now ready to describe our compiler. We start with the compiler for contextuality games where each context has size exactly 2, i.e. $|C|=2$ for all $C\in \Call$. This subsumes $2$-player non-local games as a special case. We describe the general compilers in \Partref{GeneralCompilers}, but most of the main ideas already appear in the $|C|=2$ case. In this section, we formally describe our compiler, and state its guarantees. In Section~\ref{sec:soundness-analysis}, we provide proofs.
}

Our compiler takes as input the following:
\begin{itemize}
  \item A contextuality game $\G=:(Q,A,\Call,{\rm pred},{\cal D})$ with contexts
        of size $2$, i.e. satisfying $|C|=2$ for all $C\in\Call$.
  \item A strategy specified in terms of $\qstrat=:({\cal H},\left|\psi\right\rangle ,\mathbf{O})$ 
        for $\G$ (one that achieves the quantum value of $\G$; used to describe
        the honest prover)
  \item A $\QFHE$ scheme as in \Defref{QFHEscheme} and a security parameter
        $\lambda$.
        \begin{itemize}
          \item Let $\mathbf{U}=\{U_k\}_{k\in K}$ be the group (up to global phases) of %
          unitaries %
          acting on ${\cal H}$, as in \Eqref{MBform}.
          \item Recall we use $\hat{k}$ to denote a classical encryption of $k$
                under the secret key of $\QFHE$ as in \Subsecref{FHE-Scheme}.
        \end{itemize}
  \item An Oblivious $\mathbf{U}$-Pad scheme, $\opad$, as in \Defref{Oblivious-U-pad}.
\end{itemize}
The compiler produces the following compiled game $\G'$ between a verifier and prover (summarised
in \Algref{The-1-1-compiler}). We describe $\G'$, along with the actions of an honest prover that achieves the completeness guarantee.
\begin{enumerate}
  \item The verifier proceeds as follows:
        \begin{enumerate}
          \item Sample a secret key $\sk\leftarrow\QFHE.\gen(1^{\lambda}),$ and
                a context $C\leftarrow\calD$ according to the distribution specified
                by the game ${\cal D}$, picks a question $q'$ from this context
                at random $q'\leftarrow C$. Evaluate $c_{q'}\leftarrow\Enc_{\sk}(q')$.
          \item Sample a public key/secret key pair $(\opad.\pk,\opad.\sk)\leftarrow\opad(1^{\lambda})$.
        \end{enumerate}
        Send $(c_{q'},\opad.\pk)$ to the prover.
  \item The honest prover prepares, under the $\QFHE$ encryption,
        the state $\left|\psi\right\rangle $ and subsequently measure $O_{q'}$
        to obtain an answer $a'$. It then applies an oblivious $\mathbf{U}$-pad
        and returns the classical responses.
        \begin{enumerate}
          \item Since the operations happen under the $\QFHE$ encryption, the prover
                ends up with a $\QFHE$ encryption of $a'$ which we denote by
                \[
                  c_{a'}.
                \]
                Further, it holds a $\QFHE$ encryption of the post-measurement state
                which (by the assumption on the form of the $\QFHE$ encryption) looks
                like
                \begin{equation}(U_{k''}\left|\psi_{q',a'}\right\rangle ,\hat{k}'').\label{eq:stateAfterQprime}
                \end{equation}
                Here $\left|\psi_{q',a'}\right\rangle $ is the post-measurement state
                when $\left|\psi\right\rangle $ is measured using $O_{q'}$ and the
                outcome $a'$ is obtained.
          \item The honest prover applies an oblivious $\mathbf{U}$-pad
                (see \Algref{ObliviousPauliPad}) to obtain
                \[
                  (U_{k'}U_{k''}\left|\psi_{q',a'}\right\rangle ,s')\leftarrow\opad.\enc(\opad.\pk,U_{k''}\left|\psi_{q',a'}\right\rangle ).
                \]

          \item The prover returns $(c_{a'},\hat{k}'',s')$.
        \end{enumerate}
  \item The verifier proceeds as follows:
        \begin{enumerate}
          \item Compute $a':=\QFHE.\Dec_{\sk}(c_{a'})$, $k''$ from $\hat{k}''$
                (the latter is by the form of QFHE) and $k'=\opad.\dec(\opad.\sk,s')$.
          \item Find $k$ such that\footnote{up to a global phase} $U_{k}=U_{k''}U_{k'}$ (such a $k$ always exists
                because $\mathbf{U}$ form a group).
          \item Samples a new question $q\leftarrow C$.
        \end{enumerate}
        Send
        $(q,k)$ in the clear to the prover.
  \item The honest prover measures observable $O_{q}$ conjugated by $U_{k}$,
        i.e. $U_{k}O_{q}U_{k}^{\dagger}$ and returns the corresponding answer
        $a$.
  \item There are two cases, either the questions are the same or they are
        different (since $|C|=2$). The verifier proceeds as follows:
        \begin{enumerate}
          \item If $q=q'$, accept if $a=a'$;
          \item If $q\neq q'$, accept if
                ${\rm pred}((a,a'),C)=1$.
        \end{enumerate}
\end{enumerate}
We say that a prover wins $\G'$ if the verifier accepts.

\begin{algorithm}
  \centering{}%
  \begin{tabular}{|>{\raggedright}p{5cm}|>{\centering}p{3cm}|>{\raggedright}m{5.5cm}|}
    \multicolumn{1}{>{\raggedright}p{5cm}}{Honest Prover (${\cal A}$)} & \multicolumn{1}{>{\centering}p{3cm}}{}         & \multicolumn{1}{>{\raggedright}m{5.5cm}}{Challenger ($\calC$)}\tabularnewline
    \cline{1-1} \cline{3-3}
                                                                    &   &   \tabularnewline
                                                                       &                                                & $\sk\leftarrow\QFHE.\gen(1^{\lambda})$

    $C\leftarrow\calD$

    $q'\leftarrow C$

    $c_{q'}\leftarrow\QFHE.\enc_{\sk}(q')$

    \,

    $(\opad.\pk,\opad.\sk)\leftarrow\opad(1^{\lambda})$\tabularnewline
                                                                       & {\Large $\xleftarrow{(c_{q'},\opad.\pk)}$}     & \tabularnewline
    Under the QFHE encryption, measures $O_{q'}$ and obtains an encrypted
    answer $c_{a'}$ and post-measurement state $(U_{k''}\left|\psi_{q'a'}\right\rangle ,\hat{k}'')$.

    Applies an oblivious $\mathbf{U}$-pad to this state to obtain \,$\left(U_{k'}U_{k''}\left|\psi_{q'a'}\right\rangle ,s'\right)$\,$\leftarrow$\,

    $\opad.\Enc(\opad.\pk,U_{k''}\left|\psi_{q'a'}\right\rangle ).$    &                                                & \tabularnewline
                                                                       & {\Large $\xrightarrow{(c_{a'},\hat{k}'',s')}$} & \tabularnewline
                                                                       &                                                & Using the secret keys $\sk,\opad.\sk$, finds the $k$ such that $U_{k}=U_{k''}U_{k'}$,
    samples $q\leftarrow C$\tabularnewline
                                                                       & {\Large $\xleftarrow{(q,k)}$}                  & \tabularnewline
    Measures $U_{k}O_{q}U_{k}^{\dagger}$ and obtains $a$               &                                                & \tabularnewline
                                                                       & {\Large $\overset{a}{\longrightarrow}$}                     & \tabularnewline
                                                                       &                                                & Computes $a'=\Dec_{\sk}(c_{a'})$.

    If $q=q'$, accept if $a=a'$

    If $q\neq q'$, accept if ${\rm pred}((a,a'),C)=1$.\tabularnewline
                                                                    &   &   \tabularnewline
    \cline{1-1} \cline{3-3}
  \end{tabular}\caption{Game $\protect\G'$ produced by the $(1,1)$-compiler for any contextuality
    game $\protect\G$ with contexts of size two.\label{alg:The-1-1-compiler}}
\end{algorithm}

\subsection{Compiler Guarantees}\label{subsec:compilerguarantees}

The compiler satisfies the following.
\begin{thm}[Guarantees of the $(1,1)$ compiled contextuality game $\G'$]
  \label{thm:CompiledGameIsSecure} Suppose $\QFHE$ and $\opad$ are secure (as in \Defref{QFHEscheme,Oblivious-U-pad}), and compatible (as in \Defref{QFHEcomptableopad}). Let $\G$ be a contextuality game
  with $\valNC<1$ and $|C|=2$ for all contexts $C\in \Call$. Let $\G'_{\lambda}$ be the compiled game produced by \Algref{The-1-1-compiler} on input
  $\G$ and a security parameter $\lambda$. Then, the following holds. 
  \begin{itemize}
  \item (Completeness) There is a negligible function $\negl$, such that, for all $\lambda \in \mathbb{N}$, the honest QPT prover from \Algref{The-1-1-compiler} wins $\G'_{\lambda}$ with probability at least
  \[
    c(\lambda):=\frac{1}{2}\left(1+ {\valQu}\right)-\negl(\lambda).
  \]
  \item (Soundness) For every PPT adversary $\cal A$,
  there is a negligible function $\negl'$ such that, for all $\lambda \in \mathbb{N}$, the probability
  that $\cal A$ wins $\G'_{\lambda}$
  is at most
  \[
    s(\lambda):=\frac{1}{2}\left(1+ \valNC\right)+\negl'(\lambda).
  \]
\end{itemize}
  Furthermore, $G'_{\lambda}$ is faithful to $\G$ (as in \Defref{OperationalTestOfContextuality}) with parameters $s(\lambda)$ and $c(\lambda)$.%
\end{thm}

Completeness is straightforward to verify. We assume soundness for now and defer the proof to Section \ref{sec:soundness-analysis} (below). We outline how faithfulness in \Defref{OperationalTestOfContextuality} is satisfied. 

\begin{proof}[Proof sketch]

Parse the messages as follows,  %
\begin{align*}
  (t_1,t_2) &= (c_{q^\prime},\OPad.\pk), \\
  (t_3,t_4,t_5) &= (c_{a^\prime},\hat{k}'',s'), \\
  (t_6,t_7)&=(q,k), \text{ and} \\
  t_8 &= a.
\end{align*}
The four criteria for faithfulness in \Defref{OperationalTestOfContextuality} are satisfied as follows.
\begin{enumerate}
    \item \emph{Well-formedness.} To satisfy the well-formedness property, we define the following.
        \begin{itemize}
            \item $\{\map_i\}_i$ are naturally specified: $(\map_2,\map_4,\map_5,\map_7)$ output $\perp$ at all inputs, $(\map_6,\map_8)$ are identity maps (output whatever they are given as input) while $\map_1 = \QFHE.\Dec_{\sk}$ and $\map_3:= \QFHE.\Enc_{\pk}$. Under these maps, evidently, the verifier only asks questions in a single context, during any given execution.
            \item $\qProver$ on input a description of $\qstrat':=(\cal H,\ket{\psi'},\mathbf{O'})$, behaves as the honest prover in \Algref{The-1-1-compiler} except that it uses $\qstrat'$ instead of the optimal state and measurement in $\qstrat$. It is straightforward to note that $\qProver$ satisfies the ``marginal'' requirement.
        \end{itemize}
    \item \emph{$\G$-soundness.} This is straightforward to verify from the soundness value $s$ and the construction of the compiler.
    \item \emph{Decision Faithfulness.} This is also straightforward to verify from the construction of the compiler.
    \item \emph{$\G$-completeness.} 
        \begin{enumerate}
            \item \emph{Quantum completeness.} \\
            The first requirement (about existence of $\qstrat$ achieving $c$) follows directly from the fact that the completeness value is $c$ and it is achieved by the honest prover in \Algref{The-1-1-compiler}. \\
            The second requirement (about the distribution of transcripts) is straightforward to verify by constructing a simulator $S_{\qstrat,\aux}$ that behaves exactly like $\qProver(\qstrat)$: it answers the questions corresponding $\G$ using $\qstrat$ and ${\map_i}_i$ and $\aux$ specifies that the responses $(t_4,t_5)$ are distributed exactly as $(\hat k '', s')$ produced by $\qProver(\qstrat)$.
            \item \emph{Classical completeness.} \\
            Any $\qstrat'$ where all observables commute, may be seen as a distribution over truth tables. We show that $P$ is a PPT machine that produces an identical transcript as $\qProver(\qstrat)$. $P$ proceeds as follows: 
                \begin{itemize}
                    \item Sample a $\tau$ truth table according to the distribution above.
                    \item Receive $(c_{q'},\OPad.\pk)$.
                    \item Use $\QFHE.\cEval$ to obtain $c_{\tau(q')}$ (see algorithm in the proof of \Claimref{classicalEval}).\\ 
                    Sample a random $r_{\tilde k} \leftarrow \QFHE.\Enc(\tilde k)$ for a uniform $\tilde k\leftarrow K$ and $r_s \leftarrow \OPad.\samp(\pk)$ (see \Algref{sSamp}).\\ 
                    Send $(c_{\tau(q')},r_{\tilde k},r_s)$.
                    \item Receive $(q,k)$
                    \item Sends $\tau(q)$.
                \end{itemize} 
            It is straightforward to verify that the transcript $(t_1,\dots , t_8)$ produced by $P$ is indeed identically distributed to that produced by $\qProver(\qstrat)$ by using the classical range sampling property of $\OPad$ (see \Defref{Oblivious-U-pad}) and the $\cEval$ property of $\QFHE$ (see \Claimref{classicalEval}).
        \end{enumerate}
\end{enumerate}
\end{proof}

\section{Soundness Analysis}
\label{sec:soundness-analysis}
{
  The security notion we considered in the definition of our $\QFHE$ scheme (see \Defref{QFHEscheme}) says that encryptions of any two distinct messages is indistinguishable from one another. The corresponding security game is often referred to as $2$-IND. What if one considers $\ell$-IND, i.e. $\ell$-many different encryptions? 
  We prove that $2$-IND implies a version of $\ell$-IND which we call $\cal D$-IND'. $\cal D$-IND' is more general since it allows the $\ell$ messages to be sampled from an arbitrary distribution $\cal D$ (instead of being limited to uniform). $\cal D$-IND' is also more restricted in that (1) the $\ell$ messages are fixed and known to the challenger in advance and (2) we only consider classical algorithms.\footnote{It appears that proving the general version in the quantum case may not be completely straightforward because of the inability to apply rewinding directly.} 

  But why are we suddenly talking about $\ell$-IND and $\cal D$-IND'? It turns out that the proof technique used to show $2$-IND implies $\cal D$-IND' can be applied, albeit it takes some care, to prove the soundness of the compiler. 

  Thus, as a warmup, \Subsecref{warmup} shows that $2$-IND implies $\cal D$-IND'. %
  Then \Subsecref{The-Reduction} shows how to apply these ideas to show that the $2$-IND of $\QFHE$ encryptions implies soundness of the compiled game $\G'$ (using \Algref{The-1-1-compiler}) as stated in \Thmref{CompiledGameIsSecure} (assuming that the $\opad$ is secure). We emphasise that we do not use the result in \Subsecref{warmup} directly---we only use the proof technique.

}

\subsection{Warm up | $2$-IND implies ${\cal D}$-IND'}\label{subsec:warmup}

{
  This subsection first recalls the $2$-IND game (which is essentially a restatement of \Eqref{distinguisher-1}). Then it defines the ${\cal D}$-IND' game (the prime emphasises that the number of messages is constant and known apriori\footnote{Proving the implication for the general case may be non-trivial in the quantum setting.}). It ends by showing that $2$-IND implies ${\cal D}$-IND'.
}

\begin{claim}[$2$-IND Game for QFHE]
  \label{claim:2-IND_QFHE}Let $(\gen,\enc,\dec)$ correspond to a
  $\QFHE$ encryption scheme. The scheme satisfies \Eqref{distinguisher-1}
  if and only if the following holds: \\
  In the game below, for every PPT distinguisher, there is a negligible
  function ${\rm negl}$ such that the challenger accepts with probability
  at most $\frac{1}{2}+\ngl{\lambda}.$
  \begin{center}
    \begin{tabular}{|>{\centering}p{4cm}|>{\centering}p{3cm}|>{\raggedright}m{5cm}|}
      \multicolumn{1}{>{\centering}p{4cm}}{Distinguisher} & \multicolumn{1}{>{\centering}p{3cm}}{} & \multicolumn{1}{>{\raggedright}m{5cm}}{Challenger}\tabularnewline
      \cline{1-1} \cline{3-3}
                                                          & {\Large $\xrightarrow{(x_{0},x_{1})}$} & \tabularnewline
                                                          &                                        & $\sk\leftarrow\gen(1^{\lambda})$

      \textbf{$b\leftarrow\{0,1\}$}

      $c_{b}\leftarrow\enc_{\sk}(x_{b})$\tabularnewline
                                                          & {\Large $\overset{c_{b}}{\longleftarrow}$}          & \tabularnewline
      Objective: guess $b$                                &                                        & \tabularnewline
                                                          & {\Large $\overset{b'}{\longrightarrow}$}            & \tabularnewline
                                                          &                                        & Accept if $b'=b$.\tabularnewline
                                                          &     &   \tabularnewline
      \cline{1-1} \cline{3-3}
    \end{tabular}
    \par\end{center}

\end{claim}
 
\begin{defn}[${\cal D}$-IND' Game (where ${\cal D}$ is a distribution over a
    fixed message set) for QFHE]
  \label{def:D-INDprime_QFHE}Let
  \begin{itemize}
    \item $\lambda$ be a security parameter, $\ell=O(1)$ be a fixed constant
          (relative to $\lambda$),
    \item $M=\{m_{1}\dots m_{\ell}\}$ be a set of distinct messages where $m_{1},\dots m_{\ell}\in\{0,1\}^{O(1)}$
          are of constant size,
    \item ${\cal D}$ be a probability distribution over $M$ and let $q_{{\rm guess}}=\max_{i}q_{i}$
          where $q_{i}$ is the probability assigned to $m_{i}$ by by ${\cal D}$.
  \end{itemize}
  The ${\cal D}$-IND' Security Game is a two-party game, for $(\gen,\enc,\dec)$
  as specified by $\QFHE$, is as follows:
\end{defn}
  \begin{center}
    \begin{tabular}{|>{\centering}p{4cm}|>{\centering}p{3cm}|>{\raggedright}m{5cm}|}
      \multicolumn{1}{>{\centering}p{4cm}}{Distinguisher} & \multicolumn{1}{>{\centering}p{3cm}}{} & \multicolumn{1}{>{\raggedright}m{5cm}}{Challenger}\tabularnewline
      \cline{1-1} \cline{3-3}
                                                            &       &       \tabularnewline
                                                          &                                        & $\sk\leftarrow\gen(1^{\lambda})$

      $m\leftarrow\calD$

      $c_{m}\leftarrow\enc_{\sk}(m)$\tabularnewline
                                                          & {\Large $\overset{c_{m}}{\longleftarrow}$}          & \tabularnewline
      Objective: guess $m$                                &                                        & \tabularnewline
                                                          & {\Large $\overset{m'}{\longrightarrow}$}            & \tabularnewline
                                                          &                                        & Accept if $m'=m$.\tabularnewline
                                                          &     &       \tabularnewline
      \cline{1-1} \cline{3-3}
    \end{tabular}
    \par\end{center}

\begin{prop}
  \label{prop:DINDprime}A $\QFHE$ scheme that satisfies \Eqref{distinguisher-1},
  implies that every PPT distinguisher for the ${\cal D}$-IND' game
  above, wins with probability at most $q_{{\rm guess}}+\ngl{\lambda}$.
\end{prop}

{\begin{proof}
    We can assume, without loss of generality, that the messages are given
    by $m_{i}=i$. This is because they can be uniquely indexed, and our
    arguments go through unchanged (as one can check).

    Goal: we want to show that if there is a PPT distinguisher\footnote{Note that ${\cal A}_{\ell}$ and $C_{\ell}$ can depend on ${\cal D}$;
      the use of the subscript $\ell$ is just notational convenience and
      not meant to convey all the dependencies. } ${\cal A}_{\ell}$ that wins against $C_{\ell}$ in the ${\cal D}$-IND'
    game above with probability $q_{{\rm guess}}+\eta'$ for some non-negligible
    function $\eta'$, then one can construct a PPT distinguisher ${\cal A}_{2}$
    that wins against $C_{2}$ in the $2$-IND game (see \Claimref{2-IND_QFHE})
    with probability at least $\frac{1}{2}+\eta''$ for some other non-negligible
    function $\eta''$.

    To this end, denote by $p_{ij}$ the probability that ${\cal A}_{\ell}$
    outputs $j$ when given the encryption of $i$ as input, i.e.
    \[
      p_{ij}:=\Pr\left[j\leftarrow{\cal A}_{\ell}(c_{i}):\begin{array}{c}
          \sk\leftarrow\gen(1^{\lambda}) \\
          c_{i}\leftarrow\enc_{\sk}(i)
        \end{array}\right].
    \]
    Since the message space is of constant size, ${\cal A}_{2}$ can compute
    $p_{ij}$ to inverse polynomial errors. For now, we assume ${\cal A}_{2}$
    knows $p_{ij}$ exactly and handle the precision issue at the end.
    To proceed, consider the following observations:
    \begin{enumerate}
      \item $\Pr[\text{accept}\leftarrow\left\langle {\cal A}_{\ell},C_{\ell}\right\rangle ]=\sum_{i\in\{1\dots\ell\}}q_{i}p_{ii}$
            where recall that $q_{i}$ is the probability assigned to $i$ by
            the distribution ${\cal D}$.
      \item There exist $k^{*}\neq i^{*}$ such that $\sum_{j}\left|p_{i^{*}j}-p_{k^{*}j}\right|\ge\eta$
            for some non-negligible function $\eta$.
    \end{enumerate}
    The first observation follows directly from the definition of ${\cal D}$-IND'
    and $p_{ij}$. The second follows from the assumption that ${\cal A}_{\ell}$
    wins with probability at least $q_{{\rm guess}}+\eta'$ for some non-negligible
    function $\eta'$. To see this, proceed by contradiction: Suppose for
    all $i^{*}\neq k^{*}$, it is the case that $\sum_{j}\left|p_{i^{*}j}-p_{k^{*}j}\right|\le{\rm negl}$
    for some negligible function, then one could write
    \begin{align*}
      \Pr[\text{accept}\leftarrow\left\langle {\cal A}_{\ell},C_{\ell}\right\rangle ] & =\sum_{i\in\{1\dots\ell\}}q_{i}p_{ii}                                                                                           \\
                                                                                      & =\sum_{i\in\{1\dots\ell\}}q_{i}(p_{i^{*}i}+{\rm negl})                          & \text{by the assumption above}      \\
                                                                                      & \le q_{{\rm guess}}\cancelto{1}{\sum_{i\in\{1\dots\ell\}}p_{i^{*}i}}+{\rm negl} & \text{for another }{\rm negl}\text{ function} \\
                                                                                      & \le q_{{\rm guess}}+{\rm negl}.
    \end{align*}
    But this cannot happen because we assumed that ${\cal A}_{\ell}$
    wins with probability non-negligibly greater than $q_{{\rm guess}}$.
    We therefore conclude that observation 2 must hold. Since the message
    space is constant, a PPT ${\cal A}_{2}$ can determine the following:
    \begin{itemize}
      \item The indices $i^{*}\neq k^{*}$
      \item The disjoint sets $J_{0},J_{1}\subseteq\{1\dots\ell\}$ such that
            \begin{itemize}
              \item for $j\in J_{0}\subseteq\{1\dots\ell\}$, $p_{i^{*}j}\ge p_{k^{*}j}$,
                    and
              \item for $j\in J_{1}\subseteq\{1\dots\ell\}$, $p_{i^{*}j}<p_{k^{*}j}$.
            \end{itemize}
    \end{itemize}
    Once these two indices, and the sets $J_{0},J_{1}$ are known, ${\cal A}_{2}$'s
    remaining actions are as follows:
    \begin{center}
      \begin{tabular}{|>{\centering}p{2cm}|>{\centering}p{1cm}|>{\centering}p{5cm}|>{\centering}p{2.5cm}|>{\raggedright}m{3cm}|}
        \multicolumn{1}{>{\centering}p{2cm}}{${\cal A}_{3}$ (adversary for ${\cal D}$-IND')

        \ } & \multicolumn{1}{>{\centering}p{1cm}}{}   & \multicolumn{1}{>{\centering}p{5cm}}{${\cal A}_{2}$ (the reduction)} & \multicolumn{1}{>{\centering}p{2.5cm}}{}              & \multicolumn{1}{>{\raggedright}m{3cm}}{${\cal C}_{2}$ (challenger for $2$-IND)}\tabularnewline
        \cline{1-1} \cline{3-3} \cline{5-5}
            &       &       &       &       \tabularnewline
            &                                          & Let $k^{*}\neq i^{*}\in\{1\dots\ell\}$ be as above.                  &                                                       & \tabularnewline
            &                                          &                                                                      & {\Large $\xrightarrow{(k^{*},i^{*})=:(x_{0},x_{1})}$} & \tabularnewline
            &                                          &                                                                      &                                                       & $\sk\leftarrow\gen(1^{\lambda})$

        \textbf{$b\leftarrow\{0,1\}$}

        $c_{b}\leftarrow\enc_{\sk}(x_{b})$\tabularnewline
            & {\Large $\overset{c_{b}}{\longleftarrow}$}            &                                                                      & {\Large $\overset{c_{b}}{\longleftarrow}$}                         & \tabularnewline
            & {\Large $\xrightarrow{j_{{\rm guess}}}$} &                                                                      &                                                       & \tabularnewline
            &                                          & $b'=\begin{cases}
                                                                 0 & \text{if }j_{{\rm guess}}=J_{0} \\
                                                                 1 & \text{if }j_{{\rm guess}}=J_{1}
                                                               \end{cases}$                               & {\Large $\overset{b'}{\longrightarrow}$}                           & \tabularnewline
            &                                          &                                                                      &                                                       & Accept if $b'=b$.\tabularnewline
            &       &       &       &       \tabularnewline
        \cline{1-1} \cline{3-3} \cline{5-5}
      \end{tabular}
      \par\end{center}

    Now,
    \begin{align}
      \Pr[\text{accept}\leftarrow\left\langle {\cal A}_{2},{\cal C}_{2}\right\rangle ]= & \frac{1}{2}\cdot\sum_{j\in J_{0}}\Pr[{\cal A}_{\ell}\text{ outputs }j|i^{*}\text{ was encrypted}]+\nonumber                                                                                                                   \\
                                                                                        & \frac{1}{2}\cdot\sum_{j\in J_{1}}\Pr[{\cal A}_{\ell}\text{ outputs }j|k^{*}\text{ was encrypted}]\nonumber                                                                                                                    \\
      =                                                                                 & \frac{1}{2}+\frac{1}{2}\text{\ensuremath{\sum_{j\in J_{0}}\left(p_{i^{*}j}-p_{k^{*}j}\right)}}              & \text{By def of }p_{ij}\text{ and normalisation}\nonumber                                                \\
      =                                                                                 & \frac{1}{2}+\frac{1}{4}\sum_{j}|p_{i^{*}j}-p_{k^{*}j}|                                                      & \text{\ensuremath{\left\Vert a-b\right\Vert _{1}=2\sum_{i:a_{i}>b_{i}}a_{i}-b_{i}}} \text{ for $a,b$ distrib.}\label{eq:p_istarj-p_kstarj} \\
      =                                                                                 & \frac{1}{2}+\frac{\eta}{4}                                                                                  & \text{from Observation 2}.\nonumber
    \end{align}
    Since the $\QFHE$ scheme satisfies \Eqref{distinguisher-1}, this
    is a contradiction (via \Claimref{2-IND_QFHE}) and thus the claim
    follows---up to the precision issue which we now address. %

    \paragraph{Handling the precision issue.}

    Denote by ${\cal A}_{2}$ the algorithm above where $p_{ij}$ are
    known exactly. Suppose that
    \begin{equation}
      \Pr[\text{accept}\leftarrow\langle{\cal A}_{\ell},C_{\ell}\rangle]\ge q_{{\rm guess}}+\epsilon\label{eq:prAccept}
    \end{equation}
    where $\epsilon$ is a non-negligible function. For concreteness,
    suppose\footnote{For every non-negligible function $\epsilon$, there is an infinite
      subset $\Lambda$ of its domain where $\epsilon(\lambda)\ge1/\lambda^{c}$.
      One can restrict the entire argument to this domain and still reach the same conclusion.} $\epsilon(\lambda)=1/\lambda^{c}$ for some constant $c>0$.
    \begin{prop}
      \label{prop:finitePrecisionD-IND_prime}Given that \Eqref{prAccept}
      holds, consider an algorithm $\hat{{\cal A}}_{2}$ that uses an estimate
      for $\hat{p}_{ij}$ up to precision $O(\epsilon^{3})$ (i.e. $|\hat{p}_{ij}-p_{ij}|\le O(\epsilon^{3})$).
      Then,
      \begin{align}
        \Pr[{\rm accept}\leftarrow\langle\hat{{\cal A}}_{2},C_{2}\rangle] & \ge\Pr[{\rm accept}\leftarrow\langle{\cal A}_{2},C_{2}\rangle]-O(\epsilon^{3}).\label{eq:epsiloncubeerrorinA_2hat}
      \end{align}
    \end{prop}

    To prove \Propref{finitePrecisionD-IND_prime}, we first show (see
    \Claimref{epsilonsquare} below) that $\sum_{j}|p_{i^{*}j}-p_{k^{*}j}|$
    is at least $\Omega(\epsilon^{2})$ given that \Eqref{prAccept} holds.
    Using this and \Eqref{p_istarj-p_kstarj}, it follows that $\Pr[{\rm accept}\leftarrow\langle{\cal A}_{2},C_{2}\rangle]$
    is at least one half plus $\Omega(\epsilon^{2})$ which means $\hat{{\cal A}}_{2}$---the
    finite-precision variant of ${\cal A}_{2}$---succeeds with one half
    plus non-negligible probability in the $2$-IND game of the $\QFHE$
    scheme (see \Claimref{2-IND_QFHE}).

    \begin{claim*}
      \label{claim:epsilonsquare}If \Eqref{prAccept} holds then there
      exist $i^{*}\neq k^{*}$ such that $\left\Vert p_{i^{*}}-p_{j^{*}}\right\Vert _{1}:=\sum_{j}\left|p_{i^{*}j}-p_{k^{*}j}\right|\ge\Omega(\epsilon^{2})$.

      {\begin{proof}[Proof of \Claimref{epsilonsquare}]
            We start with an elementary fact:
              Let $f,g: \Lambda\to\mathbb{R}$, where $\Lambda \subset \mathbb{N}$ is an infinite subset of $\mathbb{N}$. Let $P$ be the proposition that $f\le O(g)$ on the set $\Lambda$.
              Then $\neg P$ implies that $f\ge\Omega(g)$ on an infinite set $\Lambda'\subseteq\Lambda$.
          Using this fact, deduce that the negation of
          $\sum_{j}|p_{i^{*}j}-p_{k^{*}j}|\ge\Omega(\epsilon^{2})$ implies
          there is some infinite set $\Lambda\subseteq\mathbb{N}$ over which
          $\sum_{j}|p_{i^{*}j}-p_{k^{*}j}|\le O(\epsilon^{2})$. We prove the
          claim by contradiction. Consider to the contrary that for all $i^{*}\neq k^{*}$,
          $\sum_{j}|p_{i^{*}j}-p_{k^{*}j}|<O(\epsilon^{2})$. %
          Then, it holds
          that
          \begin{align*}
            \Pr[\text{accept}\leftarrow\langle{\cal A}_{\ell},C_{\ell}\rangle] & =\sum_{i\in\{1\dots\ell\}}q_{i}p_{ii}                                                                    \\
                                                                               & \le\sum_{i\in\{1\dots\ell\}}q_{i}(p_{i^{*}i}+O(\epsilon^{2}))                                            \\
                                                                               & \le q_{{\rm guess}}\cancelto{1}{\sum_{i\in\{1\dots\ell\}}p_{i^{*}i}}+O(\epsilon^{2}) & \because\ell=O(1) \\
                                                                               & =q_{{\rm guess}}+O(\epsilon^{2})
          \end{align*}
          but this violates \Eqref{prAccept}.
        \end{proof}
      }
    \end{claim*}
    We can now prove \Propref{finitePrecisionD-IND_prime}.
    \begin{proof}[Proof of \Propref{finitePrecisionD-IND_prime}]
      Since $\hat{{\cal A}}_{2}$ estimates $p_{ij}$ to precision $O(\epsilon^{3})$,
      it follows that it will find $i^{*}\neq k^{*}$ and a $j$ such that
      $|p_{i^{*}j}-p_{k^{*}j}|\ge\Omega(\epsilon^{2})$; recall that by
      the definition, $\ell=O(1)$ so $\left\Vert p_{i^{*}}-p_{j^{*}}\right\Vert _{1}\ge\Omega(\epsilon^{2})$
      implies there is some $j$ for which $|p_{i^{*}j}-p_{k^{*}j}|\ge\Omega(\epsilon^{2})$.

      Further, the only $j$s for which $\hat{{\cal A}}_{2}$ makes an error
      in deciding whether to have $j\in J_{0},$ or $j\in J_{1}$, are those
      for which $|p_{i^{*}j}-p_{k^{*}j}|\le O(\epsilon^{3})$. (Note that
      it cannot be that for all $j$s, $|p_{i^{*}j}-p_{k^{*}j}|\le O(\epsilon^{3})$
      because then $\left\Vert p_{i^{*}}-p_{k^{*}}\right\Vert _{1}$ cannot
      be $\ge\Omega(\epsilon^{2})$.) Thus, the error in computing $\Pr[{\rm accept}\leftarrow\langle\hat{{\cal A}}_{2},C_{\ell}\rangle]$
      using $\Pr[{\rm accept}\leftarrow\langle{\cal A}_{2},C_{\ell}\rangle]=\frac{1}{2}+\frac{1}{2}\sum_{j\in J_{0}}\left(p_{i^{*}j}-p_{k^{*}j}\right)$
      is at most $O(\epsilon^{3})$ (again, using $\ell=O(1)$). This yields
      \Eqref{epsiloncubeerrorinA_2hat} as asserted.
    \end{proof}
    This completes the proof.
  \end{proof}
}

\subsection{The Reduction\label{subsec:The-Reduction}}

We are going to essentially adapt the proof of \Propref{DINDprime}
to our setting. Start by considering the honest prover ${\cal A}$
in \Algref{The-1-1-compiler}. We reduce ${\cal A}$ to a distinguisher
${\cal A}_{2}$ for the $2$-IND security game of $\QFHE$.
\begin{itemize}
  \item Notation: Let $(Q,A)$ denote the set of questions
        and answers in the contextuality game that was compiled.
  \item Phase 1: Learning ``$p_{ij}$''.
        \begin{itemize}
          \item Consider the following procedure ${\cal A}_{2}.\truthtable$ that
                \\
                takes $q'\in Q$ as input and \\
                produces a truth table $\tau:Q\to A.$
                \begin{enumerate}
                  \item ${\cal A}_{2}$ simulates the challenger ${\cal C}$ in \Algref{The-1-1-compiler}
                        as needed (see \Algref{ATruthTable}), except that it takes $q'$
                        as input (instead of uniformly sampling it in the beginning) and uses
                        $k\leftarrow K$ (in the third step). It uses this simulation to generate
                        the first message $(c_{q'},\opad.\pk)$ and feeds it to ${\cal A}$.
                  \item ${\cal A}_{2}$ receives $(c_{a'},\hat{k}'',s')$ from ${\cal A}$.
                        For each $q\in Q$, ${\cal A}_{2}$ sends $(q,k)$ to ${\cal A}$
                        (where recall $k$ is sampled uniformly at random from $K$) and receives
                        an answer $a$. Denote by $\tau$ the truth table, i.e. the list of
                        answers indexed by the questions.
                \end{enumerate}
          \item $\calA_{2}$ repeats the procedure $\calA_{2}\cdot\truthtable(q')$
                above for each $q'\in Q$ to estimate the probability $p_{q'\tau}$
                of the procedure outputting $\tau$ on input $q'$ (the randomness
                is also over the encryption procedure etc). \\
                Here $p_{q'\tau}$ is analogous to $p_{ij}$.
          \item It finds questions $q_{*0}\neq q_{*1}$ such that
                \begin{equation}
                  \left\Vert p_{q_{*0}}-p_{q_{*1}}\right\Vert :=\sum_{\tau}|p_{q_{*0}\tau}-p_{q_{*1}\tau}|\ge\eta\label{eq:non-negl-ptau}
                \end{equation}
                for some non-negligible function $\eta$.\\
                (We defer the proof that questions (or indices) $q_{*0}\neq q_{*1}$
                exist if the $\OPad$ is secure and ${\cal A}$ wins the compiled
                contextuality game $\G'$, with probability non-negligibly greater
                than $\frac{1}{2}(1+\valNC)$. )
          \item It defines the disjoint sets $T_{0}$ and $T_{1}$ as follows: $\tau\in T_{0}$
                if $p_{q_{*0}\tau}\ge p_{q_{*1}\tau}$ and $\tau\in T_{1}$ if $p_{q_{*1}\tau}>p_{q_{*0}\tau}$.
        \end{itemize}
        \begin{algorithm}
          \begin{centering}
            \begin{tabular}{|>{\centering}p{2cm}|c|>{\centering}p{5cm}|}
              \multicolumn{1}{>{\centering}p{2cm}}{${\cal A}$} & \multicolumn{1}{c}{}                           & \multicolumn{1}{>{\centering}p{5cm}}{${\cal A}_{2}.\truthtable(q')$}\tabularnewline
              \cline{1-1} \cline{3-3}
                                                               &                                                & \tabularnewline
                                                               &                                                & $\sk\leftarrow\QFHE.\gen(1^{\lambda})$

              $c_{q'}\leftarrow\QFHE.\Enc(q')$

              $(\opad.\pk,\opad.\sk)\leftarrow\opad(1^{\lambda})$\tabularnewline
                                                               & {\Large $\xleftarrow{(c_{q'},\opad.\pk)}$}     & \tabularnewline
                                                               & {\Large $\xrightarrow{(c_{a'},\hat{k}'',s')}$} & \tabularnewline
                                                               &                                                & Samples $k\leftarrow K$\tabularnewline
                                                               & {\Large $\circlearrowleft$}                             & For each $q\in Q$, ask $(q,k)$, receive $a$ and rewind until $\tau$
              is fully specified.\tabularnewline
                                                               &                                                & \tabularnewline
                                                               &                                                & Outputs $\tau$\tabularnewline
                                                               &        &       \tabularnewline
              \cline{1-1} \cline{3-3}
            \end{tabular}
            \par\end{centering}
          \caption{\label{alg:ATruthTable}The procedure ${\cal A}_{2}.\protect\truthtable$
            takes as input a question $q'$ and produces a truth table $\tau$
            corresponding to it. Note that this is a randomised procedure (depends
            on the $\protect\QFHE$ encryption procedure) so for the same $q'$
            the procedure may output different $\tau$s. The goal is to learn
            the probabilities of different $\tau$s appearing for each question
            $q'$.}
        \end{algorithm}

  \item Phase 2: Interaction with $C_{2}$ of the $2$-IND game.
        \begin{itemize}
          \item ${\cal A}_{2}$ sends $q_{*0},q_{*1}$ to $C_{2}$ and $C_{2}$ returns
                $c_{b}$, the QFHE encryption of $q_{*b}$ (where $C_{2}$ picks $b\leftarrow\{0,1\}$).
          \item ${\cal A}_{2}$ simulates the $\opad$ itself and forwards $c_{b}$
                to ${\cal A}$, learns the truth table $\tau$ corresponding to $c_{b}$
                and outputs $b'\in\{0,1\}$ such that $\tau\in T_{b'}$.
        \end{itemize}
        \begin{algorithm}
          \begin{centering}
            \begin{tabular}{|>{\centering}p{1cm}|c|>{\centering}p{5cm}|c|>{\centering}p{4cm}|}
              \multicolumn{1}{>{\centering}p{1cm}}{}                         & \multicolumn{1}{c}{}                           & \multicolumn{1}{>{\centering}p{5cm}}{}                 & \multicolumn{1}{c}{}                     & \multicolumn{1}{>{\centering}p{4cm}}{}\tabularnewline
              \multicolumn{1}{>{\centering}p{1cm}}{${\cal A}$}               & \multicolumn{1}{c}{}                           & \multicolumn{1}{>{\centering}p{5cm}}{${\cal A}_{2}$}   & \multicolumn{1}{c}{}                     & \multicolumn{1}{>{\centering}p{4cm}}{}\tabularnewline
              \cline{1-1} \cline{3-3}
                                                                             &                                                &                                                        & \multicolumn{1}{c}{}                     & \multicolumn{1}{>{\centering}p{4cm}}{}\tabularnewline
                                                                             &                                                & \textbf{Phase 1}                                       & \multicolumn{1}{c}{}                     & \multicolumn{1}{>{\centering}p{4cm}}{}\tabularnewline
                                                                             & {\Large $\longleftrightarrow$}                              & Compute $p_{q'\tau}$ (as in the description)

              for each $q'\in Q$ by

              running ${\cal A}_{2}.\truthtable(q')$

              (see point 3, \Remref{RandomK_ConsistentA_PrecisionOfP_itau}). & \multicolumn{1}{c}{}                           & \multicolumn{1}{>{\centering}p{4cm}}{}\tabularnewline
              \cline{1-1}
              \multicolumn{1}{>{\centering}p{1cm}}{}                         &                                                &                                                        & \multicolumn{1}{c}{}                     & \multicolumn{1}{>{\centering}p{4cm}}{}\tabularnewline
              \multicolumn{1}{>{\centering}p{1cm}}{}                         &                                                & Use $p_{q'\tau}$ to learn

              the questions $q_{*0},q_{*1}$ and

              the disjoint sets $T_{0},T_{1}$ of truth tables.               & \multicolumn{1}{c}{}                           & \multicolumn{1}{>{\centering}p{4cm}}{}\tabularnewline
              \multicolumn{1}{>{\centering}p{1cm}}{}                         &                                                &                                                        & \multicolumn{1}{c}{}                     & \multicolumn{1}{>{\centering}p{4cm}}{}\tabularnewline
              \multicolumn{1}{>{\centering}p{1cm}}{${\cal A}$}               &                                                & \textbf{Phase 2}                                       & \multicolumn{1}{c}{}                     & \multicolumn{1}{>{\centering}p{4cm}}{$C_{2}$}\tabularnewline
              \cline{1-1} \cline{5-5}
                                                                             &                                                &                                                        & {\Large $\xrightarrow{(q_{*0},q_{*1})}$} & \tabularnewline
                                                                             &                                                & $(\opad.\pk,\opad.\sk)\leftarrow\opad(1^{\lambda})$    &                                          & $\sk\leftarrow\QFHE.\gen(1^{\lambda})$

              $b\leftarrow\{0,1\}$

              $c_{b}\leftarrow\QFHE.\Enc_{\sk}(q_{*b})$\tabularnewline
                                                                             & {\Large $\xleftarrow{(c_{b},\opad.\pk)}$}      &                                                        & {\Large $\overset{c_{b}}{\longleftarrow}$}            & \tabularnewline
                                                                             & {\Large $\xrightarrow{(c_{a'},\hat{k}'',s')}$} &                                                        &                                          & \tabularnewline
                                                                             &                                                & Samples $k\leftarrow K$                                &                                          & \tabularnewline
                                                                             & {\Large $\circlearrowleft$}                             & For each $q\in Q$, repeats $(q,k)$ to determine $\tau$ &                                          & \tabularnewline
                                                                             &                                                &                                                        &                                          & \tabularnewline
                                                                             &                                                & Set $b'=0$ if $\tau\in T_{0}$,

              and $b'=1$ if $\tau\in T_{1}$                                  & {\Large $\overset{b'}{\longrightarrow}$}                    & \tabularnewline
                                                                             &                                                &                                                        &                                          & Accept if $b'=b$\tabularnewline
                                                                             &      &       &       &       \tabularnewline
              \cline{1-1} \cline{3-3} \cline{5-5}
            \end{tabular}
            \par\end{centering}
          \caption{\label{alg:A_2_reduction}The algorithm ${\cal A}_{2}$ uses the adversary
            ${\cal A}$ for the compiled contextuality game $\protect\G'$ to
            break the $2$-IND security game for the $\protect\QFHE$ scheme.}

        \end{algorithm}

        The intuition is that given an encryption of $q_{*0}$, on an average,
        the above procedure would output $b'=0$ more often than $b'=1$,
        essentially by definition of $T_{0}$ and $T_{1}$, and \Eqref{non-negl-ptau}.

\end{itemize}
\begin{rem}
  \label{rem:RandomK_ConsistentA_PrecisionOfP_itau}There are three
  subtleties that are introduced, aside from the rewinding needed to
  learn $\tau$, in analysing $\G'$ as opposed to the ${\cal D}$-IND'
  game. We glossed over these above.
  \begin{enumerate}
    \item \emph{Random $k$}: When interacting with the prover ${\cal A}$,
          ${\cal A}_{2}$ above is feeding in a uniformly random $k$ instead of the correct
          $k$ which depends on $\hat{k}''$ and $s'$ (see \Algref{The-1-1-compiler}
          where ${\cal A}$ interacts with ${\cal C}$ and compare it to the
          interaction of ${\cal A}$ with ${\cal A}_{2}$).
          \begin{itemize}
            \item However, the guarantees about ${\cal A}$ are for the correct $k$.
            \item We need to formally show that the security of $\opad$ allows us to
                  work with random $k$s directly.
            \item This is straightforward enough but we need to use this a few times;
                  we'll see.
          \end{itemize}
    \item \emph{${\cal A}$ is consistent}: The other subtlety has to do with
          the proof that ${\cal A}$ winning with probability non-negligibly
          more than $\frac{1}{2}(1+\valNC)$ implies $\left\Vert p_{q_{*0}}-p_{q_{*1}}\right\Vert _{1}$
          is non-negligible. In the proof, we use the assumption that ${\cal A}$
          is \emph{consistent}, i.e. if it answers with an encryption of $a'$
          upon being asked an encryption of $q'$ in the first interaction,
          it will respond consistently if the same question is asked in the
          clear, i.e. it answers $a'$ upon being asked the same question $q'$
          in the next round. We will show that given an ${\cal A}$ that is
          non-consistent, one can still bound it's success probability by treating
          it as though it is consistent.
    \item \emph{Precision of $p_{q'\tau}$}: We also completely skipped the precision
          issue, i.e. we assumed ${\cal A}_{2}$ can learn $p_{q'\tau}$ exactly.
          In practice, this is not possible, of course. However, $p_{q'\tau}$
          can be learnt to enough precision to make the procedure work, just
          as we did in the proof of \Propref{DINDprime}.
  \end{enumerate}
\end{rem}

\subsection{Proof Strategy}

Let ${\cal A}$ be any PPT algorithm interacting with ${\cal C}$
in the compiled contextuality game $\G'$.

The following allows us to assume that we can give a uniformly random $k$ as input to ${\cal A}$ without changing the output distribution of the interaction in \Algref{random-k-or-not} (as alluded to in point 1 of \Remref{RandomK_ConsistentA_PrecisionOfP_itau}).
\begin{lem}[Uniformly random $k$ is equivalent to the correct $k$]
  \label{lem:1guessK=00003DgoodK}Let $B_{0}$ (resp. $B_{1}$) be
  a PPT algorithm that takes $q'\in Q$ as an input, interacts with
  ${\cal A}$ and outputs a bit, as described in \Algref{random-k-or-not}.
  Then there is a negligible function $\mathsf{negl}$ such that $\left|\Pr[0\leftarrow\left\langle B_{0},{\cal A}\right\rangle ]-\Pr[0\leftarrow\left\langle B_{1},{\cal A}\right\rangle ]\right|\le\mathsf{negl}$.
\end{lem}

\begin{algorithm}
  \begin{centering}
    \begin{tabular}{|>{\raggedright}p{5cm}|>{\centering}p{3cm}|>{\raggedright}m{5.5cm}|}
      \multicolumn{1}{>{\raggedright}p{5cm}}{${\cal A}$} & \multicolumn{1}{>{\centering}p{3cm}}{}         & \multicolumn{1}{>{\raggedright}m{5.5cm}}{$B_{0}(q')$ (resp. $B_{1}(q')$)}\tabularnewline
      \cline{1-1} \cline{3-3}
                                                        &           &       \tabularnewline
                                                         &                                                & $\sk\leftarrow\QFHE.\gen(1^{\lambda})$

      $c_{q'}\leftarrow\QFHE.\enc_{\sk}(q')$

      \,

      $(\opad.\pk,\opad.\sk)\leftarrow\opad(1^{\lambda})$\tabularnewline
                                                         & {\Large $\xleftarrow{(c_{q'},\opad.\pk)}$}     & \tabularnewline
                                                         &                                                & \tabularnewline
                                                         & {\Large $\xrightarrow{(c_{a'},\hat{k}'',s')}$} & \tabularnewline
                                                         &                                                & $B_{0}$ computes $k':=\opad.\Dec(\opad.\sk,s')$ and uses the secret
      key $\sk$ to compute $k''$ and then finds the $k$ satisfying $U_{k}=U_{k''}U_{k'}$.

      (resp. \textbf{$B_{1}$} samples a uniform $k\leftarrow K$).\tabularnewline
                                                         & {\Large $\overset{(q,k)}{\circlearrowleft}$}            & Potentially rewinds ${\cal A}$ to this step and queries with $(q,k)$
      for arbitrary $q\in Q$.\tabularnewline
                                                         &                                                & Runs an arbitrary procedure to compute a bit $b'$.\tabularnewline
                                                         &          &       \tabularnewline
      \cline{1-1} \cline{3-3}
    \end{tabular}
    \par\end{centering}
  \caption{\label{alg:random-k-or-not}Whether a PPT adversary ${\cal A}$ for
    the compiled contextuality game $\protect\G'$ is used with the correct
    $k$ or a uniformly random $k$ does not affect the outcome of this interaction more than negligibly.}

\end{algorithm}

The following allows us to treat ${\cal A}$ as though it is consistent
(as anticipated in point 2 of \Remref{RandomK_ConsistentA_PrecisionOfP_itau}).
\begin{lem}[Consistency only helps]
  \label{lem:2EstimateAssumingConsistent}Let
  \begin{itemize}
    \item ${\cal C}_{k\leftarrow K}$ be exactly the same as the challenger
          ${\cal C}$ for the compiled contextuality game $\G'$ except that
          it samples a uniformly random $k$ instead of computing it correctly,
          let
    \item ${\cal A}$ be any PPT algorithm that is designed to play the complied
          contextuality game $\G'$ and makes ${\cal C}_{k\leftarrow K}$ accept
          with probability $p$, i.e. $\Pr[\accept\leftarrow\left\langle {\cal A},{\cal C}_{k\leftarrow K}\right\rangle ]=p$,
          and denote by
    \item $p_{q'\tau}$ the probability that on being asked $q'$, the truth
          table ${\cal A}$ produces is $\tau$ (as defined in \Subsecref{The-Reduction}).
  \end{itemize}
  Then
  \[
    p\le\sum_{C}\Pr(C)\cdot\frac{1}{2}\left(1+\frac{1}{|C|}\sum_{q'\in C}\sum_{\tau}p_{q'\tau}\pred(\tau[C],C)\right)
  \]
  where $\Pr(C)$ denotes the probability with which ${\cal C}$ samples
  the context $C$.
\end{lem}

The two terms in the sum above, correspond to the consistency test
($q=q'$) and the predicate test ($q'\neq q$). Finally, the following
allows us to neglect the precision issue (as detailed in point 3 of
\Remref{RandomK_ConsistentA_PrecisionOfP_itau}).
\begin{lem}[Precision is not an issue]
  \label{lem:3PrecisionNonIssue}Suppose
  \begin{itemize}
    \item ${\cal A}$ wins with probability $\Pr[\accept\leftarrow\left\langle {\cal A},{\cal C}\right\rangle ]\ge\frac{1}{2}(1+\valNC)+\epsilon$
          for some non-negligible function $\epsilon$
    \item ${\cal A}_{2}$ is as in \Algref{A_2_reduction}, i.e. it is a PPT
          algorithm except for the time it spends in learning $p_{q'\tau}$
          exactly
    \item Denote by ${\cal A}_{2,\epsilon}$ a PPT algorithm that is the same
          as ${\cal A}_{2}$ except that it computes and uses an estimate $\hat{p}_{q'\tau}$
          satisfying $|\hat{p}_{q'\tau}-p_{q'\tau}|\le O(\epsilon^{3})$, in
          place of $p_{q'\tau}$.
  \end{itemize}
  Then, the PPT algorithm ${\cal A}_{2,\epsilon}$ wins with essentially
  the same probability as ${\cal A}_{2}$, i.e.
  \[
    \Pr[\accept\leftarrow\left\langle {\cal A}_{2,\epsilon},C_{2}\right\rangle ]\ge\Pr[\accept\leftarrow\left\langle {\cal A}_{2},C_{2}\right\rangle ]-O(\epsilon^{3}).
  \]

\end{lem}

We prove the following contrapositive version of the soundness guarantee
in \Thmref{CompiledGameIsSecure}.
\begin{thm}[Soundness condition restated from \Thmref{CompiledGameIsSecure}]
  \label{thm:MainSoundness} Suppose
  \begin{itemize}
    \item ${\cal A}$ is any PPT algorithm that wins with probability $\Pr[\accept\leftarrow\left\langle {\cal A},{\cal C}\right\rangle ]\ge\frac{1}{2}(1+\valNC)+\epsilon$
          for some non-negligible function $\epsilon$, and
    \item the $\opad$ used is secure, then
  \end{itemize}
  there is a PPT algorithm ${\cal A}_{2,\epsilon}$ that wins the $2$-IND
  security game of the $\QFHE$ scheme with probability
  \[
    \Pr[\accept\leftarrow\left\langle {\cal A}_{2,\epsilon},C_{2}\right\rangle ]\ge\frac{1}{2}+\nonnegl,
  \]
    where $\nonnegl$ is a non-negligible function that depends on $\epsilon$.
\end{thm}

In \Subsecref{Step1}, we prove \Thmref{MainSoundness} assuming \Lemref{1guessK=00003DgoodK, 2EstimateAssumingConsistent, 3PrecisionNonIssue}.
In \Subsecref{Step2}, we prove the lemmas.

\subsection{Proof assuming the lemmas (Step 1 of 2)\label{subsec:Step1}}

{\begin{proof}
    This part is analogous to the proof of \Propref{DINDprime}. First,
    recall the challenger ${\cal C}$ of the complied game $\G'$ and
    let ${\cal C}_{k\leftarrow K}$ denote the the same challenger, except
    that it samples $k$ uniformly at random. From \Lemref{1guessK=00003DgoodK},
    one can conclude that $|\Pr[\accept\leftarrow\left\langle {\cal A},{\cal C}\right\rangle ]-\Pr[\accept\leftarrow\left\langle {\cal A},{\cal C}_{k\leftarrow K}\right\rangle ]|\le\negl$.

    Recall the definition of $p_{q'\tau}$ from the discussion in \Subsecref{The-Reduction}.
    Using \Lemref{2EstimateAssumingConsistent}, one can write

    \begin{align}
      \Pr[\accept\leftarrow\left\langle {\cal A},{\cal C}\right\rangle ]-\negl & \le\sum_{C}\Pr(C)\cdot\frac{1}{2}\left(1+\frac{1}{|C|}\sum_{q'\in C}\sum_{\tau}p_{q'\tau}\pred(\tau[C],C)\right)\nonumber      \\
                                                                               & =\frac{1}{2}\left(1+\sum_{C}\Pr(C)\frac{1}{|C|}\sum_{q'\in C}\sum_{\tau}p_{q'\tau}\pred(\tau[C],C)\right)\label{eq:pr_acc_A_C}
    \end{align}

    Observe also that there exist $q_{*0}\neq q_{*1}$ such that
    \begin{equation}
      \left\Vert p_{q_{*0}}-p_{q_{*1}}\right\Vert _{1}:=\sum_{\tau}|p_{q_{*0}\tau}-p_{q_{*1}\tau}|\ge\eta\label{eq:p_q_star0tau_minus_p_q_star1tau}
    \end{equation}
    for some non-negligible function $\eta$. This is a consequence of
    the assumption that
    \begin{equation}
      \Pr[\accept\leftarrow\left\langle {\cal A},{\cal C}\right\rangle ]\ge\frac{1}{2}(1+\valNC)+\epsilon\label{eq:A_wins_against_C_w_non-negl}
    \end{equation}
    for some non-negligible function $\epsilon$. To see this, proceed by contradiction: Suppose that for all $q_{*0}\neq q_{*1}$, it is the
    case that $\sum_{\tau}|p_{q_{*0}\tau}-p_{q_{*1}\tau}|\le\negl$ for
    some negligible function, then one could write, using \Eqref{pr_acc_A_C},
    \begin{align}
      \Pr[\accept\leftarrow\left\langle {\cal A},{\cal C}\right\rangle ] & \le\frac{1}{2}\left(1+\sum_{C}\Pr(C)\cancel{\frac{1}{|C|}\sum_{q'\in C}}\sum_{\tau}p_{q_{*0}\tau}\pred(\tau[C],C)\right)+\negl'\nonumber \\
                                                                         & =\frac{1}{2}\left(1+\sum_{\tau}p_{q_{*0}\tau}\sum_{C}\Pr(C)\pred(\tau[C],C)\right)+\negl'\nonumber                                       \\
                                                                         & \le\frac{1}{2}\left(1+\valNC\right)+\negl'\label{eq:valNC_at_most}
    \end{align}
    where in the first step, we used $p_{q_{*0}\tau}$ instead of $p_{q'\tau}$
    at the cost of a $\negl'$ additive error, in the second step, we
    rearranged the sum, and in the final step, we observe that the expression
    is just a convex combination of values achieved using non-contextual
    strategies, and this is at most $\valNC$. However, \Eqref{valNC_at_most}
    contradicts \Eqref{A_wins_against_C_w_non-negl} that says, by assumption,
    ${\cal A}$ wins with probability non-negligibly more than $\frac{1}{2}\left(1+\valNC\right)$.

    So far, we have established \Eqref{p_q_star0tau_minus_p_q_star1tau}
    holds for some distinct questions $q_{*0}\neq q_{*1}$. We assume
    that ${\cal A}_{2}$ can learn $p_{q'\tau}$ exactly and invoke \Lemref{3PrecisionNonIssue}
    to handle the fact that these can only be approximated to inverse
    polynomial errors but that does not change the conclusion. Therefore,
    ${\cal A}_{2}$ can learn $q_{*0}\neq q_{*1}$ and construct the sets
    $T_{0},T_{1}$ of truth tables (recall: $\tau$ is in $T_{0}$ if
    $p_{q_{*0}\tau}\ge p_{q_{*1}\tau}$ and in $T_{1}$ otherwise). Proceeding
    as in the proof of \Propref{DINDprime}, and focusing on Phase 2
    of the interaction, it holds that
    \begin{align*}
      \Pr[\accept\leftarrow\left\langle {\cal A}_{2},C_{2}\right\rangle ] & =\frac{1}{2}\cdot\sum_{\tau\in T_{0}}\Pr[{\cal A}\text{ outputs }\tau|q_{*0}\text{ was encrypted}]+   \\
                                                                          & \ \ \frac{1}{2}\cdot\sum_{\tau\in T_{1}}\Pr[{\cal A}\text{ outputs }\tau|q_{*1}\text{ was encrypted}] \\
                                                                          & =\frac{1}{2}+\frac{1}{2}\sum_{\tau\in T_{0}}\left(p_{q_{*0}\tau}-p_{q_{*1}\tau}\right)                \\
                                                                          & =\frac{1}{2}+\frac{1}{4}\left\Vert p_{q_{*0}}-p_{q_{*1}}\right\Vert _{1}                              \\
                                                                          & =\frac{1}{2}+\frac{\eta}{4}.
    \end{align*}
    Since the $\QFHE$ scheme satisfies \Eqref{distinguisher-1}, we have
    a contradiction (via \Claimref{2-IND_QFHE}) which means our assumption
    that ${\cal A}$ wins with probability non-negligibly more than $\frac{1}{2}\left(1+\valNC\right)$
    is false, completing the proof.
  \end{proof}
}

\subsection{Proof of the lemmas (Step 2 of 2) \label{subsec:Step2}}

\Lemref{1guessK=00003DgoodK}, about using a uniformly random $k$ instead of
the correct one, is almost immediate but we include a brief proof.

{\begin{proof}[Proof of \Lemref{1guessK=00003DgoodK}]
    Consider the adversary ${\cal A}_{\opad}$ for $\opad$ as defined
    in \Algref{red_for_lemma1}.
        \begin{algorithm}[h]
      \begin{centering}
       \begin{tabular}{|>{\raggedright}p{0.1cm}|>{\centering}c|>{\raggedright}m{5.5cm}|>{\raggedright}c|>{\raggedright}m{5.0cm}|}
          \multicolumn{1}{>{\raggedright}p{0.3cm}}{${\cal A}$}             & \multicolumn{1}{>{\centering}p{1cm}}{}            & \multicolumn{1}{>{\raggedright}m{5.5cm}}{$\calA_{\OPad}$}            & \multicolumn{1}{>{\centering}m{1cm}}{}      & \multicolumn{1}{>{\raggedright}m{5.0cm}}{$C_{\opad}$}\tabularnewline
          \cline{1-1} \cline{3-3} \cline{5-5}
                                                                        &       &       &       &       \tabularnewline
                                                                         &                                                   & $\sk\leftarrow\QFHE.\gen(1^{\lambda})$

          $c_{q'}\leftarrow\QFHE.\enc_{\sk}(q')$                         & \centering{}                                      & $(\opad.\pk,\opad.\sk)\leftarrow\opad(1^{\lambda})$\tabularnewline
                                                                         & {\Large $ \xleftarrow{(c_{q'},\opad.\pk)}$}        &                                                                      & \centering{}{\Large $\xleftarrow{\opad.\sk}$} & \tabularnewline
                                                                         &                                                   &                                                                      & \centering{}                                  & \tabularnewline
                                                                         & {\Large $\xrightarrow{(c_{a'},\hat{k}'',s')}$}    &                                                                      & \centering{}{\Large $\overset{s'}{\longrightarrow}$}       & \tabularnewline
                                                                        & & & & $b\leftarrow\{0,1\}$, $k_{0}\leftarrow K$, $k_{1}=\opad.\dec(\opad.\sk,s')$\tabularnewline
                                                                        & & &  \centering{}{\Large $\overset{k_{b}}{\longleftarrow}$}  & \tabularnewline
                                                                         &                                                   & $B$ defines $k':=k_{b}$ and uses the secret key $\sk$ to compute
          $k''$ and then finds the $k$ satisfying $U_{k}=U_{k''}U_{k'}$. &        & \tabularnewline
                                                                         & {\Large $\underset{a}{\overset{(q,k)}{\circlearrowleft}}$} & Potentially rewinds ${\cal A}$ to this step and queries with $(q,k)$
          for arbitrary $q\in Q$.                                        & \centering{}                                      & \tabularnewline
                                                                         &                                                   & Runs an arbitrary procedure to compute a bit $b'$.                   &        & \tabularnewline
                                                                        & & &  \centering{}{\Large $\overset{b'}{\longrightarrow}$} & \tabularnewline
                                                                         &                                                   &                                                                      & \centering{}                                  & Accept if $b'=b$\tabularnewline
                                                                         &      &       &       &       \tabularnewline
          \cline{1-1} \cline{3-3} \cline{5-5}
        \end{tabular}
        \par\end{centering}
      \caption{\label{alg:red_for_lemma1}Whether a PPT adversary ${\cal A}$ for
        the compiled contextuality game $\protect\G'$ is used with the correct
        $k$ or a uniformly random $k$, it makes no difference, if all algorithms involved
        are PPT.}
    \end{algorithm}
    Observe that
    \begin{align}
      \Pr\left[\accept\leftarrow\left\langle {\cal A}_{\opad},C_{\opad}\right\rangle \right]= & \frac{1}{2}\Pr[{\cal A}_{\opad}\text{ outputs }b'=0|b=0]+\nonumber                                                                                                   \\
                                                                                              & \frac{1}{2}\Pr[{\cal A}_{\opad}\text{ outputs }b'=1|b=1]\nonumber                                                                                                    \\
      =                                                                                       & \frac{1}{2}\Big(\Pr[{\cal A}_{\opad}\text{ outputs }b'=0|b=0]-\nonumber                                                                                              \\
                                                                                              & \ \ \Pr[{\cal A}_{\opad}\text{ outputs }b'=0|b=1]\Big) + \frac{1}{2} \nonumber                                                                                                      \\
      =                                                                                       & \frac{1}{2}\left(\Pr[0\leftarrow\left\langle B_{0},{\cal A}\right\rangle ]-\Pr[0\leftarrow\left\langle B_{1},{\cal A}\right\rangle ]\right) + \frac{1}{2} \label{eq:securityLemma1}
    \end{align}
    where $B_{0}$ and $B_{1}$ interact with ${\cal A}$ as in \Algref{random-k-or-not}. 
    The last equality holds because for a random $k'$, $k$ also becomes
    random.

    Let ${\cal A}_{\opad}'$ be ${\cal A}_{\opad}$ except that it outputs
    $b'\oplus1$ instead of $b'$. Using the analogue of \Eqref{securityLemma1}
    for ${\cal A}'_{\opad}$, \Eqref{securityLemma1} itself and the security
    of $\opad$, it follows that there is a negligible function $\negl$
    such that
    \[
      \left|\Pr[0\leftarrow\left\langle B_{0},{\cal A}\right\rangle ]-\Pr[0\leftarrow\left\langle B_{1},{\cal A}\right\rangle ]\right|\le\mathsf{negl}.
    \]
  \end{proof}
}

We now look at the proof of \Lemref{2EstimateAssumingConsistent}
which crucially relies on the fact that the consistency test and the
predicate test happen with equal probability.

{\begin{proof}[Proof of \Lemref{2EstimateAssumingConsistent}]
    Consider the adversary ${\cal A}'$ in \Algref{red_for_lemma2}
    that uses ${\cal A}$ to interact with ${\cal C}_{k\leftarrow K}=:{\cal C}'$.
    \begin{algorithm}[h]
      \begin{centering}
        \begin{tabular}{|>{\raggedright}p{0.5cm}|>{\centering}c|>{\raggedright}m{4cm}|>{\raggedright}c|>{\raggedright}m{5cm}|}
          \multicolumn{1}{>{\raggedright}p{0.5cm}}{${\cal A}$} & \multicolumn{1}{>{\centering}p{2cm}}{}              & \multicolumn{1}{>{\raggedright}m{4cm}}{${\cal A}'$}          & \multicolumn{1}{>{\raggedright}m{2cm}}{}                         & \multicolumn{1}{>{\raggedright}m{5cm}}{${\cal C}_{k\leftarrow K}=:\calC'$}\tabularnewline
          \cline{1-1} \cline{3-3} \cline{5-5}
                                                             &                                                     &                                                              & \centering{}                                                     & \tabularnewline
                                                             & {\Large $\xleftarrow{(c_{q'},\opad.\pk)}$}          &                                                              & \centering{}{\Large $\xleftarrow{(c_{q'},\opad.\pk)}$}           & \tabularnewline
                                                             &                                                     &                                                              & \centering{}                                                     & \tabularnewline
                                                             & {\Large $\xrightarrow{(c_{a'},\hat{k}'',s')}$}      &                                                              & \centering{}                                                     & \tabularnewline
                                                             &                                                     & $\tilde{k}\leftarrow K$                                      & \centering{}                                                     & \tabularnewline
                                                             & {\Large $\overset{(\tilde{q},\tilde{k})}{\circlearrowleft}$} & Rewinds ${\cal A}'$ to learn the truth table $\tau$.         & \centering{}                                                     & \tabularnewline
                                                             &                                                     & Evaluate $c_{\tau(q')}\leftarrow\QFHE.\Eval_{\tau}(c_{q'}).$ & \centering{}{\Large $\xrightarrow{(c_{\tau(q')},\hat{k}'',s')}$} & \tabularnewline
                                                             &                                                     &                                                              & \centering{}{\Large $\xleftarrow{(q,k)}$}                        & \tabularnewline
                                                             &                                                     &                                                              & \centering{}{\Large $\xrightarrow{a:=\tau(q)}$}                  & Recall:

          ($a'=\dec_{\sk}(c_{\tau(q')})$

          If $q=q'$, accept if $a=a'$.

          If $q\neq q'$, accept if ${\rm pred}((a,a'),C)=1$.\tabularnewline
          &     &       &       &       \tabularnewline
          \cline{1-1} \cline{3-3} \cline{5-5}
        \end{tabular}
        \par\end{centering}
      \caption{\label{alg:red_for_lemma2}${\cal A}'$, a potentially QPT algorithm,
        uses the PPT algorithm ${\cal A}$ to play the compiled contextuality
        game $\protect\G'$. Its winning probability upper bounds that of
        ${\cal A}$ and can be computed in terms of $p_{q'\tau}$ for ${\cal A}$.}
    \end{algorithm}

    We show that
    \begin{align}
      \Pr[\accept\leftarrow\left\langle {\cal A},{\cal C}'\right\rangle ] & \le\Pr[\accept\leftarrow\left\langle {\cal A}',\calC'\right\rangle ]\label{eq:Aprime_does_just_as_well_or_better}                                    \\
                                                                          & =\sum_{C}\Pr(C)\cdot\frac{1}{2}\left(1+\frac{1}{|C|}\sum_{q'\in C}\sum_{\tau}p_{q'\tau}\pred(\tau[C],C)\right).\label{eq:Aprime_wins_with_atmost}
    \end{align}
    We note that ${\cal A}'$ may be a QPT algorithm because it runs $\QFHE.\Eval$.
    Indeed, existing $\QFHE$ schemes don't produce a classical procedure
    for $\QFHE.\Eval$, even when the ciphertext and the logical circuit
    are classical. However, we only use $\calA'$ to compute an upper
    bound on the performance of ${\cal A}$ in terms of $p_{q'\tau}$
    and therefore ${\cal A}'$ being QPT (as opposed to PPT) is not a
    concern here.

    We first derive \Eqref{Aprime_does_just_as_well_or_better}. Consider
    the interactions $\left\langle {\cal A},{\cal C}'\right\rangle $
    and $\left\langle {\cal A}',{\cal C}'\right\rangle $. Note that for
    any given $c_{q'}$, the challenger ${\cal C}'$ either asks $q=q'$
    or $q\neq q'$ with equal probabilities. Conditioned on $c_{q'}$,
    there are two cases: (1) ${\cal A}$ is consistent, in which case
    the answers to $q,q'$ given by ${\cal A}'$ and ${\cal A}$ are identical.
    (2) ${\cal A}$ is inconsistent, in which case, ${\cal C}'$ accepts
    ${\cal A}$ with probability \emph{at most} $1/2$ while ${\cal C}'$
    accepts ${\cal A}'$ with probability \emph{at least $1/2$. }

    \Eqref{Aprime_wins_with_atmost} follows because ${\cal C}'$ selects
    a context $C$ with probability $\Pr(C)$, the $1/2$ denotes whether
    a consistency test ($q=q'$) is performed or a predicate test ($q\neq q'$)
    is performed. By construction of ${\cal A}'$, the consistency test
    passes with probability $1$. The $1/|C|$ factor indicates that either
    of the two questions in $C$ could have been asked as the first question
    $q'$ (under the $\QFHE$ encryption; and $|C|=2$ here). Again, by
    construction of ${\cal A}'$, the probability of clearing the predicate
    test is the weighted average of $\pred(\tau[C],C)$ where the weights
    are given by $p_{q'\tau}$. This completes the proof.
  \end{proof}
}

\Lemref{3PrecisionNonIssue} follows by proceeding as in the proof
of \Propref{finitePrecisionD-IND_prime}.

\clearpage{}

\part{General Computational Test of Contextuality---beyond size-$2$ contexts}\label{part:GeneralCompilers}

{
  So far, we looked at contextuality games with size $2$ contexts. This part %
  explains how to generalise our compiler to games with contexts of arbitrary size. In fact, we consider two generalisations of our compiler, both of which produce single prover, 4-message (2-round) compiled games, irrespective of the size of contexts in the original contextuality game. 
  \paragraph{The $(|C|,1)$ compiler.}
  The compiled game asks all $|C|$ questions from one context under the $\QFHE$ encryption in the first round, and asks one question in the clear in the second round. It guarantees that no PPT algorithm can succeed with probability more than \[
            1-\const_{1}+\negl
          \]
          where $\const_{1}=1/|Q|$ when all questions are asked uniformly
          at random.\footnote{$\const_{1}=\min_{C\in\Call}\Pr(C)/|C|$ in general.} The proof idea is similar to the $(1,1)$ compiler with one major difference---the analogue of the bound in \Lemref{2EstimateAssumingConsistent} changes. In essence, there one could obtain the bound by constructing a prover that is always consistent (i.e.\ the encrypted answer and the answer in the clear match when the corresponding questions match). Here, we have the ``dual'' property instead. The bound is obtained by constructing a prover that always satisfies the constraint but may not be consistent. \\
          As stated in the introduction, while this compiler works for games with perfect completeness, the bound on the classical value is sometimes more than the honest quantum value---so this compiler fails to give a separation between quantum and classical for some games (including the KCBS game of \Exaref{KCBSintro}). However, for other games (like the magic square game of \Exaref{MagicSquare}, it yields a larger completeness-soundness gap than the following universal compiler.
  \paragraph{The $(|C|-1,1)$ compiler.}
  This compiler addresses the limitation of the one above and is truly universal---in the sense that for any contextuality game $\G$ with $\valNC < \valQu$, the compiled game $\G'$ will also be such that every PPT algorithm wins with probability strictly smaller than a QPT algorithm. More precisely, PPT algorithms cannot succeed with probability more than
      \[
        \frac{1}{|C|}\left(|C|-1+\valNC\right)+\negl
      \]
      while there is a QPT algorithm that wins with probability at least
      \[
        \frac{1}{|C|}\left(|C|-1+\valQu\right)-\negl.
      \]
      Note that for $|C|=2$, we recover
      \[
        \frac{1}{2}\left(1+\valNC\right)+\negl,\quad\frac{1}{2}(1+\valQu)-\negl
      \]
      respectively, which are the bounds for our $(1,1)$ compiler.\\
  The idea behind the construction is the following:          
  \begin{itemize}
    \item Sample a context $C$ and pick a question $q_{\skp}\leftarrow C$
          uniformly at random.
    \item Ask all questions in $C$ except $q_{\skp}$, i.e. ask $C\backslash q_{\skp}$,
          under $\QFHE$ encryption
    \item Sample $q\leftarrow C$ and ask $q$ in the clear.
    \item Now,
          \begin{enumerate}
            \item if $q=q_{\skp}$, test the predicate using the decrypted answers to $C\backslash q_{\skp}$
                  and to $q$;
            \item if $q\neq q_{\skp}$, and so $q\in C\backslash q_{\skp}$; check that the corresponding answers are consistent.
          \end{enumerate}
  \end{itemize}
  The proof is based on the following idea. %
  \begin{itemize}
    \item The key observation is that, just as in the $(1,1)$ compiler, given
          an adversary ${\cal A}$, one can consider an adversary ${\cal A}'$
          such that it is ``consistent'' and wins with at least as much probability
          as ${\cal A}$ and using this, one can obtain a bound on the success
          probability of ${\cal A}$ in terms of $\valNC$.
    \item The compiled game was purposefully designed to ensure the following: being ``consistent''
          can only help. As we saw, this was not the case for the $(|C|,1)$
          game---so selecting fewer questions in the first round is what, perhaps surprisingly,
          allows us to construct a universal test (in the sense described above). %
  \end{itemize}
  Before moving to formal descriptions and proofs, we remark that it would be nice to have the success probability be independent of $|C|$, but this somehow seems hard to avoid. Why? Because the dependence on $|C|$ comes essentially from the fact the we are only asking one question in the clear in the second round. So, it seems that to avoid this dependence one
  would have to ask more than one question in the clear. However, this complicates the task of extracting a non-contextual assignment by rewinding the PPT adversary, because, for instance, the prover can now assign different values to $q_{1}$ depending on whether it was asked with $q_{2}$ or $q_{3}$. Thus, obtaining a better compiler seems to require new ideas.

}

\pagebreak{}

\section{Construction of the $(|C|,1)$ compiler}

{We define the compiler formally first. }

\begin{algorithm}[h]
  \begin{centering}
    \begin{tabular}{|>{\raggedright}p{5cm}|>{\centering}p{3cm}|>{\raggedright}m{5.5cm}|}
      \multicolumn{1}{>{\raggedright}p{5cm}}{Honest Prover (${\cal A}$)}     & \multicolumn{1}{>{\centering}p{3cm}}{}                 & \multicolumn{1}{>{\raggedright}m{5.5cm}}{Challenger ($\calC$)}\tabularnewline
      \cline{1-1} \cline{3-3}
                                                                            &       &       \tabularnewline
                                                                             &                                                        & $\sk\leftarrow\QFHE.\gen(1^{\lambda})$

      $C\leftarrow\calD$

      $c\leftarrow\QFHE.\enc_{\sk}(C)$

      \,

      $(\opad.\pk,\opad.\sk)\leftarrow\opad(1^{\lambda})$\tabularnewline
                                                                             & {\Large $\xleftarrow{(c,\opad.\pk)}$}                  & \tabularnewline
      Under the QFHE encryption, measures $\{O_{q}\}_{q\in C}$ and obtains
      an encrypted answers $\{c_{\mathbf{a}}\}$ (where $\mathbf{a}$ are
      the answers indexed by $C$)

      and post-measurement state $(U_{k''}\left|\psi_{\mathbf{a}C}\right\rangle ,\hat{k}'')$.

      Applies an oblivious $\mathbf{U}$-pad to this state to obtain \,$\left(U_{k'}U_{k''}\left|\psi_{\mathbf{a}C}\right\rangle ,s'\right)$\,$\leftarrow$\,

      $\opad.\Enc(\opad.\pk,U_{k''}\left|\psi_{\mathbf{a}C}\right\rangle ).$ &                                                        & \tabularnewline
                                                                             & {\Large $\xrightarrow{(c_{\mathbf{a}},\hat{k}'',s')}$} & \tabularnewline
                                                                             &                                                        & Using the secret keys $\sk,\text{ \& }\opad.\sk$, finds the $k$
      such that $U_{k}=U_{k''}U_{k'}$, samples $q\leftarrow C$\tabularnewline
                                                                             & {\Large $\xleftarrow{(q,k)}$}                          & \tabularnewline
      Measures $U_{k}O_{q}U_{k}^{\dagger}$ and obtains $a$                   &                                                        & \tabularnewline
                                                                             & {\Large $\overset{a}{\longrightarrow}$}                             & \tabularnewline
                                                                             &                                                        & Computes $\mathbf{a}=\Dec_{\sk}(c_{\mathbf{a}})$.

      Accepts if both (1) and (2) hold:

      (1) $\pred(\mathbf{a},C)=1$, and

      (2) $\mathbf{a}[q]=a$.\tabularnewline
                                                                            &       &       \tabularnewline
      \cline{1-1} \cline{3-3}
    \end{tabular}
    \par\end{centering}
  \centering{}\caption{Game $\protect\G'$ produced by the $(|C|,1)$-compiler on input a contextuality
    game $\protect\G$, and security parameter~$\lambda$.\label{alg:C-1-compiler}}
\end{algorithm}

\subsection{Compiler Guarantees}

The compiler satisfies the following.
\begin{thm}[Guarantees of the $(|C|,1)$ compiled contextuality game $\G'$]
  \label{thm:C_1_CompiledGameIsSecure}
  Suppose $\QFHE$ and $\opad$ are secure (as in \Defref{QFHEscheme,Oblivious-U-pad}), and compatible (as in \Defref{QFHEcomptableopad}).
  Let $\G$ be any contextuality
  game with $\valNC<1$. %
  Let $\G'_{\lambda}$ be the compiled game
  produced by \Algref{C-1-compiler} on input $\G$ and a security parameter
  $\lambda$. Then, the following holds.
\begin{itemize}
    \item (Completeness) There is a negligible function $\negl$, such that, for all $\lambda \in \mathbb{N}$, the honest QPT prover from \Algref{C-1-compiler} wins $\G'_{\lambda}$ with probability at least
  \[
    c(\lambda):=\valQu - \negl(\lambda). %
  \]
    \item (Soundness) For every PPT adversary $\cal A$, there
  is a negligible function $\negl'$ such that, for all $\lambda \in \mathbb{N}$, the probability
  that $\cal A$ wins $\G'_{\lambda}$ is at most
  \[
    s(\lambda):=1-\const_{1}+\negl'(\lambda) \,,
  \]
  where $\const_{1}=\min_{C\in\Call}\Pr(C)/|C|=O(1)$.
\end{itemize}

 Furthermore, $G'_{\lambda}$ is faithful to $\G$ (as in \Defref{OperationalTestOfContextuality}) with parameters $s(\lambda)$ and $c(\lambda)$.

\end{thm}

Completeness is straightforward to verify. The proof of faithfulness is analogous to that of \Thmref{CompiledGameIsSecure}. We prove soundness in Section \ref{sec:11}.

Note that in games where all questions are asked uniformly at random,
$\const_{1}=1/(|\Call|\cdot |C|)$, yielding the simple bound $1-1/(|\Call|\cdot|Q|)+\negl$
for PPT provers (that we quoted earlier).

{
Two brief remarks about the theorem are in order before we give the proof. First, we note that unlike the completeness value, which is $\valQu-\negl$, the soundness does not correspondingly depend on $\valNC$. %
Second, as mentioned earlier, this compiler is not universal: for KCBS the classical value in $\G'$ is $0.9$ which is greater than the honest quantum value which is approximately $0.8944$. However, the compiler is still outputs a non-trivial game, for instance, when given as input the GHZ game, because the latter has perfect completeness.}

\section{Soundness Analysis of the $(|C|,1)$ compiler}
\label{sec:11}

{The proof is analogous to that of the $(1,1)$ compiler with some
  crucial differences and, thus, we skip the intuition and details for
  parts that are essentially unchanged.}

\subsection{The Reduction\label{subsec:C_1_The-Reduction}}

{
   The procedure for TruthTable generation is the same as in \Secref{soundness-analysis}---except that one samples the context, and more importantly we consider the encrypted answer as well when conditioning, i.e. instead
  of writing $p_{q'\tau}$, we now use $p_{C,\mathbf{a},\tau}$;
   we have to add the $\mathbf{a}$ dependence more explicitly because
  of the following.

   The analogue of \Lemref{2EstimateAssumingConsistent} is slightly different. Instead of having the adversary be consistent (in the sense that the
        encrypted answers and the clear answers are always the same by construction),  
              we require the complementary property---the adversary always ensures
                        that the encrypted answer satisfies the predicate, but may not be
                        consistent.
        Since we now actually rely on the encrypted answers to check whether
                or not the predicate is satisfied
                   in the analysis, we cannot simply drop the encrypted answer.          

Here is the ``TruthTable'' algorithm and the actual reduction ${\cal A}_{2}$.
}

\begin{algorithm}[h]
  \begin{centering}
    \begin{tabular}{|>{\centering}p{2cm}|c|>{\centering}p{5.5cm}|}
      \multicolumn{1}{>{\centering}p{2cm}}{${\cal A}$} & \multicolumn{1}{c}{}                                   & \multicolumn{1}{>{\centering}p{5.5cm}}{${\cal A}_{2}.\truthtable(C)$}\tabularnewline
      \cline{1-1} \cline{3-3}
                                                       &                                                        & \tabularnewline
                                                       &                                                        & $\sk\leftarrow\QFHE.\gen(1^{\lambda})$

      $c\leftarrow\QFHE.\Enc(C)$

      $(\opad.\pk,\opad.\sk)\leftarrow\opad(1^{\lambda})$\tabularnewline
                                                       & {\Large $\xleftarrow{(c,\opad.\pk)}$}                  & \tabularnewline
                                                       & {\Large $\xrightarrow{(c_{\mathbf{a}},\hat{k}'',s')}$} & \tabularnewline
                                                       &                                                        & Samples $k\leftarrow K$\tabularnewline
                                                       & {\Large $\circlearrowleft$}                                     & For each $q\in Q$, ask $(q,k)$, receive $a$ and rewind until $\tau$
      is fully specified.\tabularnewline
                                                       &                                                        & \tabularnewline
                                                       &                                                        & Outputs $\mathbf{a}:=\QFHE.\dec_{\sk}(\mathbf{a})$ and $\tau$. \tabularnewline
                                                       &        &       \tabularnewline
      \cline{1-1} \cline{3-3}
    \end{tabular}
    \par\end{centering}
  \caption{Analogue of \Algref{ATruthTable}, except that it explicitly outputs $\mathbf{a}$, in addition to $\tau$. \label{alg:ATruthTable_C1}}

\end{algorithm}

\begin{itemize}
  \item The definition of $C_{*0},C_{*1}\in\Call$ is now ``indices'' such
        that
        \[
          \sum_{\tau}\left|\sum_{\mathbf{a}}p_{C_{*0}\mathbf{a}\tau}-\sum_{\mathbf{a}}p_{C_{*1}\mathbf{a}\tau}\right|\ge\eta
        \]
        where $\eta$ is non-negligible.
  \item $T_{0}$ and $T_{1}$ are defined as $\tau\in T_{0}$ if $\sum_{\mathbf{a}}p_{C_{*0}\mathbf{a}\tau}\ge\sum_{\mathbf{a}}p_{C_{*1}\mathbf{a}\tau}$
        and $\tau\in T_{1}$ otherwise, i.e. $\sum_{\mathbf{a}}p_{C_{*0}\mathbf{a}\tau}<\sum_{\mathbf{a}}p_{C_{*1}\mathbf{a}\tau}$.
\end{itemize}
\begin{algorithm}
  \begin{centering}
    \begin{tabular}{|>{\centering}p{0.5cm}|c|>{\centering}p{5cm}|c|>{\centering}p{4cm}|}
      \multicolumn{1}{>{\centering}p{0.5cm}}{}           & \multicolumn{1}{c}{}                                   & \multicolumn{1}{>{\centering}p{5cm}}{}                 & \multicolumn{1}{c}{}                     & \multicolumn{1}{>{\centering}p{4cm}}{}\tabularnewline
      \multicolumn{1}{>{\centering}p{0.5cm}}{${\cal A}$} & \multicolumn{1}{c}{}                                   & \multicolumn{1}{>{\centering}p{5cm}}{${\cal A}_{2}$}   & \multicolumn{1}{c}{}                     & \multicolumn{1}{>{\centering}p{4cm}}{}\tabularnewline
      \cline{1-1} \cline{3-3}
                                                       &                                                        &                                                        & \multicolumn{1}{c}{}                     & \multicolumn{1}{>{\centering}p{4cm}}{}\tabularnewline
                                                       &                                                        & \textbf{Phase 1}                                       & \multicolumn{1}{c}{}                     & \multicolumn{1}{>{\centering}p{4cm}}{}\tabularnewline
                                                       & {\Large $\longleftrightarrow$}                                      & Compute $p_{C\mathbf{a}\tau}$ (as in the description)

      for each $C\in\Call$ by

      running ${\cal A}_{2}.\truthtable(C)$.           & \multicolumn{1}{c}{}                                   & \multicolumn{1}{>{\centering}p{4cm}}{}\tabularnewline
      \cline{1-1}
      \multicolumn{1}{>{\centering}p{0.5cm}}{}           &                                                        &                                                        & \multicolumn{1}{c}{}                     & \multicolumn{1}{>{\centering}p{4cm}}{}\tabularnewline
      \multicolumn{1}{>{\centering}p{0.5cm}}{}           &                                                        & Use $p_{C\mathbf{a}\tau}$ to learn

      the contexts $C_{*0},C_{*1}$ and

      the disjoint sets $T_{0},T_{1}$ of truth tables. & \multicolumn{1}{c}{}                                   & \multicolumn{1}{>{\centering}p{4cm}}{}\tabularnewline
      \multicolumn{1}{>{\centering}p{0.5cm}}{}           &                                                        &                                                        & \multicolumn{1}{c}{}                     & \multicolumn{1}{>{\centering}p{4cm}}{}\tabularnewline
      \multicolumn{1}{>{\centering}p{0.5cm}}{${\cal A}$} &                                                        & \textbf{Phase 2}                                       & \multicolumn{1}{c}{}                     & \multicolumn{1}{>{\centering}p{4cm}}{$C_{2}$}\tabularnewline
      \cline{1-1} \cline{5-5}
                                                       &                                                        &                                                        & {\Large $\xrightarrow{(C_{*0},C_{*1})}$} & \tabularnewline
                                                       &                                                        & $(\opad.\pk,\opad.\sk)\leftarrow\opad(1^{\lambda})$    &                                          & $\sk\leftarrow\QFHE.\gen(1^{\lambda})$

      $b\leftarrow\{0,1\}$

      $c_{b}\leftarrow\QFHE.\Enc_{\sk}(C_{*b})$\tabularnewline
                                                       & {\Large $\xleftarrow{(c_{b},\opad.\pk)}$}              &                                                        & {\Large $\overset{c_{b}}{\longleftarrow}$}            & \tabularnewline
                                                       & {\Large $\xrightarrow{(c_{\mathbf{a}},\hat{k}'',s')}$} &                                                        &                                          & \tabularnewline
                                                       &                                                        & Samples $k\leftarrow K$                                &                                          & \tabularnewline
                                                       & {\Large $\circlearrowleft$}                                     & For each $q\in Q$, repeats $(q,k)$ to determine $\tau$ &                                          & \tabularnewline
                                                       &                                                        &                                                        &                                          & \tabularnewline
                                                       &                                                        & Set $b'=0$ if $\tau\in T_{0}$,

      and $b'=1$ if $\tau\in T_{1}$                    & {\Large $\overset{b'}{\longrightarrow}$}                            & \tabularnewline
                                                       &                                                        &                                                        &                                          & Accept if $b'=b$\tabularnewline
                                                       &        &       &       &       \tabularnewline
      \cline{1-1} \cline{3-3} \cline{5-5}
    \end{tabular}
    \par\end{centering}
  \caption{\label{alg:A_2_reduction-1}The algorithm ${\cal A}_{2}$ uses the
    adversary ${\cal A}$ for the compiled contextuality game $\protect\G'$,
    to break the $2$-IND security game for the $\protect\QFHE$ scheme.}
\end{algorithm}

\subsection{Proof Strategy}

\begin{lem}[{{[}Analogue of \Lemref{1guessK=00003DgoodK}{]} Uniformly random $k$ is equivalent to the correct $k$}]
  \label{lem:1_C1_guessk=00003Dgoodk} Let $B_{0}$ (resp. $B_{1}$)
  be a PPT algorithm that takes $q'\in Q$ as an input, interacts with
  ${\cal A}$ and outputs a bit, as described in \Algref{random-k-or-not}.
  Then there is a negligible function $\mathsf{negl}$ such that $\left|\Pr[0\leftarrow\left\langle B_{0},{\cal A}\right\rangle ]-\Pr[0\leftarrow\left\langle B_{1},{\cal A}\right\rangle ]\right|\le\mathsf{negl}$.
\end{lem}

\begin{algorithm}
  \begin{centering}
    \begin{tabular}{|>{\raggedright}p{5cm}|>{\centering}p{3cm}|>{\raggedright}m{5.5cm}|}
      \multicolumn{1}{>{\raggedright}p{5cm}}{${\cal A}$} & \multicolumn{1}{>{\centering}p{3cm}}{}                 & \multicolumn{1}{>{\raggedright}m{5.5cm}}{$B_{0}(C)$ (resp. $B_{1}(C)$)}\tabularnewline
      \cline{1-1} \cline{3-3}
                                                        &           &       \tabularnewline
                                                         &                                                        & $\sk\leftarrow\QFHE.\gen(1^{\lambda})$

      $c\leftarrow\QFHE.\enc_{\sk}(C)$

      \,

      $(\opad.\pk,\opad.\sk)\leftarrow\opad(1^{\lambda})$\tabularnewline
                                                         & {\Large $\xleftarrow{(c,\opad.\pk)}$}                  & \tabularnewline
                                                         &                                                        & \tabularnewline
                                                         & {\Large $\xrightarrow{(c_{\mathbf{a}},\hat{k}'',s')}$} & \tabularnewline
                                                         &                                                        & $B_{0}$ computes $k':=\opad.\Dec(\opad.\sk,s')$ and uses the secret
      key $\sk$ to compute $k''$ and then finds the $k$ satisfying $U_{k}=U_{k''}U_{k'}$.

      (resp. \textbf{$B_{1}$} samples a uniform $k\leftarrow K$).\tabularnewline
                                                         & {\Large $\overset{(q,k)}{\circlearrowleft}$}                    & Potentially rewinds ${\cal A}$ to this step and queries with $(q,k)$
      for arbitrary $q\in Q$.\tabularnewline
                                                         &                                                        & Runs an arbitrary procedure to compute a bit $b'$.\tabularnewline
                                                         &          &       \tabularnewline
      \cline{1-1} \cline{3-3}
    \end{tabular}
    \par\end{centering}
  \caption{\label{alg:C_1_random-k-or-not}Whether a PPT adversary ${\cal A}$
    for the compiled contextuality game $\protect\G'$ is used with the
    correct $k$ or a uniformly random $k$, it makes no difference, if all algorithms
    involved are PPT.}
\end{algorithm}

{We are finally ready to write the first step which is conceptually
  different from the previous analysis. It is the analogue of \Lemref{2EstimateAssumingConsistent}---except
  that:
  \begin{itemize}
    \item Earlier, we obtained the bound by essentially treating ${\cal A}$
          as being \emph{consistent} (with the answers under the encryption
          and those in the clear) where by virtue of being consistent, ${\cal A}$
          could not be \emph{feasible} i.e. ${\cal A}$ could not satisfy all
          the predicates.
    \item Now, we treat ${\cal A}$ as being \emph{feasible} (satisfies all
          predicates) , but by virtue of being feasible, it cannot be \emph{consistent
          }(can't have the encrypted answers be consistent with a global truth
          assignment)\emph{.}
  \end{itemize}
}
\begin{lem}[{{[}Analogue of \Lemref{2EstimateAssumingConsistent}{]} Feasibility
      only helps}]
  \label{lem:2_C1_feasibility_only_helps}%
  Let
    \begin{itemize}
      \item $\G'$ be the compiled game as in \Thmref{C_1_CompiledGameIsSecure},
            let ${\cal C}$ be the challenger in $\G'$, and let
      \item ${\cal C}_{k\leftarrow K}$ be exactly the same as the challenger
            ${\cal C}$ except that it samples a uniformly random $k$ instead
            of computing it correctly,
    \item ${\cal A}$ be any PPT algorithm that is designed to play the complied
          contextuality game $\G'$ and makes ${\cal C}_{k\leftarrow K}$ accept
          with probability $p$, i.e. $\Pr\left[\accept\leftarrow\left\langle {\cal A},{\cal C}_{k\leftarrow K}\right\rangle \right]=p$
          and denote by
    \item $p_{C\mathbf{a}\tau}$ be the probability that on being asked the
          context $C$, the answers given are $\mathbf{a}$ and the truth table
          produced is $\tau$ (as defined in \Subsecref{C_1_The-Reduction})
    \item $p'_{C\mathbf{a}\tau}$ be a probability distribution derived from
          $p_{C\mathbf{a}\tau}$ to be such that $p'$ is always feasible and
          whenever $p$ has support on $C,\mathbf{a}$ that satisfy the corresponding
          predicate, $p'$ and $p$ agree. More precisely, for any $(C,\mathbf{a},\tau)$
          satisfying $p_{C\mathbf{a}\tau}>0$ it holds that
          \begin{itemize}
            \item ($p'$ matches $p$ exactly when $p$ is feasible) if $p_{C\mathbf{a}\tau}>0\land\pred(\mathbf{a}[C],C)=1$,
                  \begin{itemize}
                    \item $p_{C\mathbf{a}\tau}'=p_{C\mathbf{a}\tau}$ and
                  \end{itemize}
            \item ($p'$ is always feasible) if $\pred(\mathbf{a}[C],C)=0$,
                  \begin{itemize}
                    \item $p'_{C\mathbf{a}\tau}=0$
                  \end{itemize}
            \item (when $p$ is infeasible, $p'$ preserves probabilities but makes
                  it feasible) if $p_{C\mathbf{a}\tau}>0\land\pred(\mathbf{a}[C],C)=0$,
                  there is an\footnote{The first property implies $\pred(\mathbf{a}'[C],C)=1$.}
                  $\mathbf{a}'$ such that
                  \begin{itemize}
                    \item $p'_{C\mathbf{a}'\tau}=p_{C\mathbf{a}\tau}$.
                  \end{itemize}
          \end{itemize}
  \end{itemize}
  Then
  \[
    p\le\sum_{C\mathbf{a}\tau}p_{C\mathbf{a}\tau}'\Pr(C)\frac{1}{|C|}\sum_{q\in C}\delta_{\mathbf{a}[q],\tau(q)}
  \]
  where $\Pr(C)$ denotes the probability with which ${\cal C}$ samples
  the context $C$.
\end{lem}

{The analogue of \Lemref{3PrecisionNonIssue} should go through with
  almost no changes.}
\begin{lem}[{{[}Analogue of \Lemref{3PrecisionNonIssue}{]} Precision is not an issue}]
  \label{lem:3_C1_precision_non_issue}Exactly as \Lemref{3PrecisionNonIssue}
  except that ${\cal C}$ is the challenger for the game compiled using
  the $(|C|,1)$ compiler.
\end{lem}

{As before, we prove the contrapositive of the soundness condition. }
\begin{thm}[Soundness condition restated from \Thmref{C_1_CompiledGameIsSecure}]
  \label{thm:C1_MainSoundness}Suppose
  \begin{itemize}
    \item ${\cal A}$ is any PPT algorithm that wins with probability
          \[
            \Pr[\accept\leftarrow\left\langle {\cal A},{\cal C}\right\rangle ]\ge1-\const_{1}+\epsilon
          \]
          for some non-negligible function $\epsilon$ where $\const_{1}=\min_{C\in\Call}\Pr(C)/|C|=O(1)$,
          and
    \item the $\opad$ used is secure, then
  \end{itemize}
  there is a PPT algorithm ${\cal A}_{2,\epsilon}$ that wins the $2$-IND
  security game of the $\QFHE$ scheme with probability
  \[
    \Pr[\accept\leftarrow\left\langle {\cal A}_{2,\epsilon},C_{2}\right\rangle ]\ge\frac{1}{2}+\nonnegl,
  \]
  where $\nonnegl$ is a non-negligible function that depends on $\epsilon$.

\end{thm}

In \Subsecref{C1_proof_step1}, we prove \Thmref{C1_MainSoundness} assuming \Lemref{1_C1_guessk=00003Dgoodk, 2EstimateAssumingConsistent, 3PrecisionNonIssue}.
In \Subsecref{C1_proof_step2} we prove the lemmas.

\subsection{Proof assuming the lemmas (Step 1 of 2)\label{subsec:C1_proof_step1}}

We introduce some notation. Recall the definition of $p_{C\mathbf{a}\tau}$ from \Subsecref{C_1_The-Reduction}
and of $p_{C\mathbf{a}\tau}'$ from \Lemref{2_C1_feasibility_only_helps}. %
\begin{itemize}
  \item We use the notation $p_{\mathbf{a}\tau|C}:=p_{C\mathbf{a}\tau}$ since
        $C$ was given as input to the procedure ${\cal C}_{2}\cdot\truthtable$
        and it produced outputs $\mathbf{a},\tau$. Similarly define $p'_{\mathbf{a}\tau|C}:=p'_{C\mathbf{a}\tau}$. %
  \item We also use
        \begin{equation}
          p_{\mathbf{a}\tau|C}=p_{\mathbf{a}|\tau C}\cdot p_{\tau|C}\text{ and }p_{\mathbf{a}\tau|C}'=p_{\mathbf{a}|\tau C}'\cdot p_{\tau|C}'\label{eq:conditional}
        \end{equation}
        to denote conditionals.\footnote{Using $p(a,b|c)=p(a|b,c)\cdot p(b|c)=\frac{p(a,b,c)}{p(b,c)}\cdot\frac{p(b,c)}{p(c)}$.}
  \item Finally, we use the convention of denoting marginals by dropping the
        corresponding index, i.e.
        \[
          p_{\tau|C}:=\sum_{\mathbf{a}}p_{\mathbf{a}\tau|C}\text{ and }p_{\tau|C}':=\sum_{\mathbf{a}}p_{\mathbf{a}\tau|C}'.
        \]
  \item Note that by definition of $p'_{\mathbf{a}\tau|C}$, it holds that
        \begin{equation}
          p'_{\tau|C}=p_{\tau|C}.\label{eq:C1_p_tau_C_equals_pprime_tau_C}
        \end{equation}
        {\begin{proof}
            From \Lemref{1_C1_guessk=00003Dgoodk}, one concludes that
            \[
              \left|[\Pr[\accept\leftarrow\left\langle {\cal A},{\cal C}\right\rangle ]-\Pr[\accept\leftarrow\left\langle {\cal A},{\cal C}_{k\leftarrow K}\right\rangle ]\right|\le\negl.
            \]
            Recall the definition of $p_{C\mathbf{a}\tau}$ from \Subsecref{C_1_The-Reduction}.
            Using \Lemref{2EstimateAssumingConsistent}, one can write
            \begin{align}
              \Pr[\accept\leftarrow\left\langle {\cal A},{\cal C}\right\rangle ]\le & \sum_{C\mathbf{a}\tau}p_{\mathbf{a}\tau|C}'\Pr(C)\frac{1}{|C|}\sum_{q\in C}\delta_{\mathbf{a}[q],\tau[q]}+\negl\label{eq:C1_pr_acc_A_C}                                                                                                                               \\
              =                                                                     & \sum_{C\tau}p{}_{\tau|C}\sum_{\mathbf{a}}p_{\mathbf{a}|\tau C}'\Pr(C)\frac{1}{|C|}\sum_{q\in C}\delta_{\mathbf{a}[q],\tau[q]}+\negl     & \text{using }\prettyref{eq:conditional}\text{ \& }\prettyref{eq:C1_p_tau_C_equals_pprime_tau_C}\label{eq:C1_pr_acc_further}
            \end{align}
            Observe also that there exist $C_{*0}\neq C_{*1}$ such that
            \begin{equation}
              \sum_{\tau}\left|p_{\tau|C_{*0}}-p_{\tau|C_{*1}}\right|\ge\eta\label{eq:p_tau_C_star_0_minus_p_tau_C_star_1}
            \end{equation}
            for some non-negligible function $\eta$. This is a consequence of
            the assumption that
            \begin{equation}
              \Pr[\accept\leftarrow\left\langle {\cal A},{\cal C}\right\rangle ]\ge1-\const_{1}+\epsilon\label{eq:C1_A_wins_against_C-w_non-negl}
            \end{equation}
            for some non-negligible function $\epsilon$. To see this, proceed by contradiction: Suppose that for all $C_{*0}\neq C_{*1}$, it is the
            case that
            \begin{equation}
              \sum_{\tau}|p_{\tau|C_{*0}}-p_{\tau|C_{*1}}|\le\negl\label{eq:C1_p_tau_C_negligible}
            \end{equation}
            then, setting $p_{\tau}:=p_{\tau|C_{*0}}$ (for instance), it holds
            that (using \Eqref{C1_pr_acc_further})
            \begin{align}
              \Pr[\accept\leftarrow\left\langle {\cal A},{\cal C}\right\rangle ] & \le\sum_{\tau}p{}_{\tau}\sum_{C}\Pr(C)\sum_{\mathbf{a}}p_{\mathbf{a}|\tau C}'\sum_{q\in C}\frac{\delta_{\mathbf{a}[q],\tau[q]}}{|C|}+\negl'\nonumber                                         \\
                                                                                 & \le\sum_{\tau}p_{\tau}\Bigg(\underbrace{\sum_{C\neq C_{\tau}}\Pr(C)\cdot1}_{\text{Term 1}}+\underbrace{\Pr(C_{\tau})\cdot\left(1-\frac{1}{|C|}\right)}_{\text{Term 2}}\Bigg)+\negl'\nonumber \\
                                                                                 & \le\sum_{\tau}p_{\tau}\left(1-\frac{\Pr(C_{\tau})}{|C|}\right)+\negl'\le1-\frac{\min_{C\in\Call}\Pr(C)}{|C|}+\negl'=1-\const_{1}+\negl'\label{eq:C1_atmost_1-const_1}
            \end{align}
            where the first line uses \Eqref{C1_p_tau_C_negligible} to substitute
            $p_{\tau|C}$ with $p_{\tau}$ at the cost a negligible factor but
            the second line needs some explanation. Observe that (a) for each
            truth table $\tau$, there is a context $C_{\tau}\in\Call$ such that
            $\pred(\tau[C],C)=0$ because $\valNC<1$. Observe also that
            (b) by assumption on $p'_{\mathbf{a}\tau|C}$ (recall \Lemref{2_C1_feasibility_only_helps}),
            all $\mathbf{a}$ with non-zero weight are feasible, i.e. $\pred(\mathbf{a}[C_{\tau}],C_{\tau})=1$
            for any $p'_{\mathbf{a}|\tau C}>0$. This implies that there is at
            least one question $q\in C_{\tau}$ such that $\mathbf{a}[q]\neq\tau[q]$
            (i.e. $\delta_{\mathbf{a}[q],\tau[q]}=0$)---else $\tau$ would also
            have satisfied the predicate on $C_{\tau}$ which it does not. Using
            (a) and (b), in Term 1, we simply upper bound the remaining sum by
            $1$ while in term $2$, we upper bound the sum by concluding that
            for at least one question in $C_{\tau}$, the delta function vanishes.
            The last inequality follows by setting term 1 to be $1-\Pr(C_{\tau})$
            and rearranging and minimising. Now, since \Eqref{C1_atmost_1-const_1}
            contradicts \Eqref{C1_A_wins_against_C-w_non-negl} we conclude that
            our assumption \Eqref{C1_p_tau_C_negligible} must be false, establishing
            \Eqref{p_tau_C_star_0_minus_p_tau_C_star_1}.

            The remaining analysis goes through almost unchanged from the $(1,1)$
            compiler case.
            \begin{align*}
              \Pr[\accept\leftarrow\left\langle {\cal A}_{2},C_{2}\right\rangle ]= & \frac{1}{2}\cdot\sum_{\tau\in T_{0}}\Pr[{\cal A}\text{ outputs }\tau|C_{*0}\text{ was encrypted}]+         \\
                                                                                   & \frac{1}{2}\cdot\sum_{\tau\in T_{1}}\Pr[{\cal A}\text{ outputs }\tau|C_{*1}\text{ was encrypted}]          \\
              =                                                                    & \frac{1}{2}+\frac{1}{2}\sum_{\tau\in T_{0}}(p_{\tau|C_{*0}}-p_{\tau|C_{*1}})                               \\
              =                                                                    & \frac{1}{2}+\frac{1}{4}\sum_{\tau}\left|p_{\tau|C_{*0}}-p_{\tau|C_{*1}}\right|=\frac{1}{2}+\frac{\eta}{4}.
            \end{align*}
            This, together with \Lemref{3_C1_precision_non_issue}, yields a contradiction
            with the security of the $\QFHE$ scheme. We therefore conclude that
            \Eqref{C1_A_wins_against_C-w_non-negl} is false, which means $\Pr[\accept\leftarrow\left\langle {\cal A},{\cal C}\right\rangle ]\le1-\const_{1}+\negl$.
          \end{proof}
        }
\end{itemize}

\subsection{Proof of the lemmas (Step 2 of 2)}\label{subsec:C1_proof_step2}

The proofs of \Lemref{1_C1_guessk=00003Dgoodk} and \Lemref{3_C1_precision_non_issue}
are analogous to those of \Lemref{1guessK=00003DgoodK} and
\Lemref{3PrecisionNonIssue}. Here, we prove \Lemref{2_C1_feasibility_only_helps}.

{\begin{proof}[Proof of \Lemref{2_C1_feasibility_only_helps}]

    \begin{algorithm}[h]
      \begin{centering}
        \begin{tabular}{|>{\raggedright}p{0.5cm}|>{\centering}p{2cm}|>{\raggedright}m{5cm}|>{\raggedright}m{2cm}|>{\raggedright}m{4.8cm}|}
          \multicolumn{1}{>{\raggedright}p{0.5cm}}{${\cal A}$} & \multicolumn{1}{>{\centering}p{2cm}}{}                              & \multicolumn{1}{>{\raggedright}m{5cm}}{${\cal A}'$}                 & \multicolumn{1}{>{\raggedright}m{2cm}}{}          & \multicolumn{1}{>{\raggedright}m{4.8cm}}{${\cal C}_{k\leftarrow K}=:\calC'$}\tabularnewline
          \cline{1-1} \cline{3-3} \cline{5-5}
                                                             &                                                                     &                                                                     & \centering{}                                      & \tabularnewline
                                                             & {\Large $\xleftarrow{(c,\opad.\pk)}$}                               &                                                                     & \centering{}{\Large $\xleftarrow{(c,\opad.\pk)}$} & Recall that $c\leftarrow\QFHE.\enc_{\sk}(C)$.\tabularnewline
                                                             &                                                                     &                                                                     & \centering{}                                      & \tabularnewline
                                                             & {\Large $\xrightarrow{(c_{\mathbf{a}},\hat{k}'',s')}$}              &                                                                     & \centering{}                                      & \tabularnewline
                                                             &                                                                     & $\tilde{k}\leftarrow K$                                             & \centering{}                                      & \tabularnewline
                                                             & {\Large $\overset{(\tilde{q},\tilde{k})}{\circlearrowleft}$}                 & Rewinds ${\cal A}'$ to learn the truth table $\tau$.                & \centering{}                                      & \tabularnewline
                                                             &                                                                     & Under the QFHE encryption, with probability $p_{\mathbf{a}|\tau C}$
          outputs $c_{\mathbf{a}'}\leftarrow\QFHE.\Eval_{\mathbf{a}}(c)$ where

          $\mathbf{a}'=\mathbf{a}$ if $\pred(\mathbf{a}[C],C)=1$,

          otherwise $\mathbf{a}'=\mathbf{a}_{C}$.            & \centering{}{\Large $\xrightarrow{(c_{\mathbf{a}'},\hat{k}'',s')}$} & \tabularnewline
                                                             &                                                                     &                                                                     & \centering{}{\Large $\xleftarrow{(q,k)}$}         & \tabularnewline
                                                             &                                                                     &                                                                     & \centering{}{\Large $\xrightarrow{a:=\tau(q)}$}   & (Recall:)

          Accept if both (1) and (2) hold:

          (1) ${\rm pred}(\mathbf{a}',C)=1$ and

          (2) $\mathbf{a}'[q]=a$.\tabularnewline
                                                            &           &           &           &       \tabularnewline
          \cline{1-1} \cline{3-3} \cline{5-5}
        \end{tabular}
        \par\end{centering}
      \caption{\label{alg:C1_red_for_lemma2}${\cal A}'$ uses ${\cal A}$ (a PPT
        algorithm---crucial because it is rewound), to play the compiled
        contextuality game $\protect\G'$. Its winning probability upper
        bounds that of ${\cal A}$ and can be computed in terms of $p_{\mathbf{a}\tau|C}$
        for ${\cal A}$.}
    \end{algorithm}

    Let $p_{\mathbf{a}\tau|c}$ be as in \Subsecref{C1_proof_step1} and
    recall that $p_{\mathbf{a}\tau|C}=p_{\mathbf{a}|\tau C}p_{\tau|C}$.
    For each $C\in\Call$, denote by $\mathbf{a}_{C}$ answers such that
    $\pred(\mathbf{a}_{C},C)=1$. Consider the adversary ${\cal A}'$
    as in \Algref{C1_red_for_lemma2}. Since this is just a way to compute
    an upper bound, the running time of ${\cal A}'$ does not matter.
    We show that
    \begin{align*}
      \Pr[\accept\leftarrow\left\langle {\cal A},{\cal C}'\right\rangle ]\le & \Pr[\accept\leftarrow\left\langle {\cal A}',{\cal C}'\right\rangle ]                                                     \\
      =                                                                      & \sum_{C}\Pr(C)\sum_{\tau}p_{\tau|C}\sum_{\mathbf{a}'}p'_{\mathbf{a}'|\tau C}\sum_{q\in C}\delta_{\mathbf{a}'[q],\tau[q]}
    \end{align*}
    where the first line holds because, by construction, ${\cal A}'$
    can do no worse than ${\cal A}$. More specifically, ${\cal A}'$
    behaves exactly like ${\cal A}$ when $\mathbf{a}$ satisfies the
    predicate, i.e. $\mathbf{a}'=\mathbf{a}$, and when $\mathbf{a}$
    fails the predicate, ${\cal A}'$ only increases its probability of
    success by responding with $\mathbf{a}'\neq\mathbf{a}$ such that
    $\pred(\mathbf{a}',C)=1$. The second line is straightforward. Denote
    by $p'_{\tau|C}$ the probability with which ${\cal A}'$ responds
    with $\tau$ on being asked $C$. Similarly, let $p'_{\mathbf{a}'|\tau C}$
    be the probability that ${\cal A}'$ responds with $\mathbf{a}'$
    given $\tau$ and $C$. Then, the challenger ${\cal C}'$ asks a context
    $C$ with probability $\Pr(C)$, to which the prover ${\cal A}'$
    responds with $\tau$ with probability $p'_{\tau|C}=p_{\tau|C}$ and
    given $\tau C$, it responds with $\mathbf{a}'$ with probability
    $p'_{\mathbf{a}'|\tau C}$. It follows that $p'_{\mathbf{a}'\tau|C}:=p'_{\mathbf{a}'|\tau C}p'_{\tau|C}$
    satisfies the asserted properties: it has no support over $(\mathbf{a}',C)$
    that are not feasible (don't satisfy the predicate for the corresponding
    context $C$) and whenever $(\mathbf{a},C)$ is feasible, $\mathbf{a}'=\mathbf{a}$
    by construction.
  \end{proof}
}

\newpage{}

\section{Construction of the $(|C|-1,1)$ Compiler}

{We define the compiler formally.}

\begin{algorithm}[h]
  \centering{}%
  \begin{tabular}{|>{\raggedright}p{5cm}|>{\centering}p{3cm}|>{\raggedright}m{5.5cm}|}
    \multicolumn{1}{>{\raggedright}p{5cm}}{Honest Prover (${\cal A}$)} & \multicolumn{1}{>{\centering}p{3cm}}{}                  & \multicolumn{1}{>{\raggedright}m{5.5cm}}{Challenger ($\calC$)}\tabularnewline
    \cline{1-1} \cline{3-3}
                                                                        &           &           \tabularnewline
                                                                       &                                                         & $\sk\leftarrow\QFHE.\gen(1^{\lambda})$

    $C\leftarrow\calD$

    $q_{\skp}\leftarrow C$

    $C':=C\backslash q_{\skp}$

    $c\leftarrow\QFHE.\enc_{\sk}(C')$

    \,

    $(\opad.\pk,\opad.\sk)\leftarrow\opad(1^{\lambda})$\tabularnewline
                                                                       & {\Large $\xleftarrow{(c,\opad.\pk)}$}                   & \tabularnewline
    Under the QFHE encryption, measures $\{O_{q'}\}_{q'\in C'}$ and obtains
    encrypted answers $c_{\mathbf{a}'}$ (where $\mathbf{a}'$ are the
    answers, indexed by $C'$)

    and the post-measurement state $(U_{k''}\left|\psi_{C'\mathbf{a}'}\right\rangle ,\hat{k}'')$.

    Applies an oblivious $\mathbf{U}$-pad to this state to obtain \,$\left(U_{k'}U_{k''}\left|\psi_{C'\mathbf{a}'}\right\rangle ,s'\right)$\,$\leftarrow$\,

    $\opad.\Enc(\opad.\pk,U_{k''}\left|\psi_{q'a'}\right\rangle ).$    &                                                         & \tabularnewline
                                                                       & {\Large $\xrightarrow{(c_{\mathbf{a}'},\hat{k}'',s')}$} & \tabularnewline
                                                                       &                                                         & Using the secret keys $\sk,\opad.\sk$, finds the $k$ such that $U_{k}=U_{k''}U_{k'}$,
    samples $q\leftarrow C$\tabularnewline
                                                                       & {\Large $\xleftarrow{(q,k)}$}                           & \tabularnewline
    Measures $U_{k}O_{q}U_{k}^{\dagger}$ and obtains $a$               &                                                         & \tabularnewline
                                                                       & {\Large $\overset{a}{\longrightarrow}$}                              & \tabularnewline
                                                                       &                                                         & Computes $\mathbf{a}'=\Dec_{\sk}(c_{\mathbf{a}'})$.

    If $q\in C'$,

    $\quad$accept if $a=\mathbf{a}'[q]$

    If $q\notin C'$ (i.e. $q=q_{\skp}$),

    $\quad$accept if ${\rm pred}(\mathbf{a}'\cup a,C)=1$ $\quad$where
    $\mathbf{a}'\cup a$ denotes answers indexed by $C$.\tabularnewline
    &           &       \tabularnewline
    \cline{1-1} \cline{3-3}
  \end{tabular}\caption{Game $\protect\G'$ produced by the $(1,1)$-compiler for any contextuality
    game $\protect\G$ with contexts of size two.\label{alg:The-Cminus11-compiler}}
\end{algorithm}

\subsection{Compiler Guarantees}
The compiler satisfies the following. Without loss of generality, we restrict to contextuality games where all contexts have the same size (see Remark \ref{rem:samesizecontexts}). 

\begin{thm}[Guarantees of the $(|C|-1,1)$ compiled contextuality game $\G'$]
  \label{thm:Cminus11_CompiledGameIsSecure}
  Suppose $\QFHE$ and $\opad$ are secure (as in \Defref{QFHEscheme,Oblivious-U-pad}), and compatible (as in \Defref{QFHEcomptableopad}).
  Let $\G$ be any contextuality
  game with $\valNC<1$ where all contexts are of the same size (i.e.
  $|C|=|C'|$ for all $C,C'\in\Call$). Let $\G'_{\lambda}$ be the compiled game
  produced by \Algref{The-Cminus11-compiler} on input $\G$ and a security parameter
  $\lambda$. Then, the following holds.
    \begin{itemize}
        \item  (Completeness) There is a negligible function $\negl$, such that, for all $\lambda \in \mathbb{N}$, the honest QPT prover from \Algref{The-Cminus11-compiler} wins $\G'_{\lambda}$ with probability at least
  \[
    c(\lambda):=1-\frac{1}{|C|}+\frac{\valQu}{|C|}-\negl(\lambda).
  \]
        \item (Soundness) For every PPT adversary $\cal A$, there
  is a negligible function $\negl'$ such that, for all $\lambda \in \mathbb{N}$, the probability
  that $\cal A$ wins $\G'_{\lambda}$ is at most
  \[
    s(\lambda):=1-\frac{1}{|C|}+\frac{\valNC}{|C|}+\negl'(\lambda) \,,
  \]
    \end{itemize}
 Furthermore, $G'_{\lambda}$ is faithful to $\G$ (as in \Defref{OperationalTestOfContextuality}) with parameters $s(\lambda)$ and $c(\lambda)$.
\end{thm}

Completeness is straightforward to verify.  The proof of faithfulness is analogous to that of \Thmref{CompiledGameIsSecure}. We prove soundness in Section \ref{sec:12}.

\section{Soundness Analysis of the $(|C|-1,1)$ compiler}
\label{sec:12}

\subsection{The Reduction\label{subsec:Cminus11,The-Reduction}}

Denote by $\Cprimeall:=\{C\backslash q_{\skp}\}_{C\in\Call,q_{\skp}\in C}$
the set consisting of contexts with exactly one question removed.

The reduction is very similar to that in \Subsecref{The-Reduction},
so we do not repeat the accompanying high-level explanations. Let ${\cal A}_{2}\cdot\truthtable$ be as defined in \Algref{Cminus11_A_2_reduction}. Define $p_{C'\tau}$ to be the probability that the procedure ${\cal A}_{2}\cdot\truthtable(C')$
outputs $\tau$, where $C'\in\Cprimeall$ is a context with exactly
one question removed.
\begin{algorithm}[h]
  \begin{centering}
    \begin{tabular}{|>{\centering}p{2cm}|c|>{\centering}p{5cm}|}
      \multicolumn{1}{>{\centering}p{2cm}}{${\cal A}$} & \multicolumn{1}{c}{}                                    & \multicolumn{1}{>{\centering}p{5cm}}{${\cal A}_{2}.\truthtable(C')$}\tabularnewline
      \cline{1-1} \cline{3-3}
                                                       &                                                         & \tabularnewline
                                                       &                                                         & $\sk\leftarrow\QFHE.\gen(1^{\lambda})$

      $c\leftarrow\QFHE.\Enc(C')$

      $(\opad.\pk,\opad.\sk)\leftarrow\opad(1^{\lambda})$\tabularnewline
                                                       & {\Large $\xleftarrow{(c,\opad.\pk)}$}                   & \tabularnewline
                                                       & {\Large $\xrightarrow{(c_{\mathbf{a}'},\hat{k}'',s')}$} & \tabularnewline
                                                       &                                                         & Samples $k\leftarrow K$\tabularnewline
                                                       & {\Large $\circlearrowleft$}                                      & For each $q\in Q$, ask $(q,k)$, receive $a$ and rewind until $\tau$
      is fully specified.\tabularnewline
                                                       &                                                         & \tabularnewline
                                                       &                                                         & Outputs $\tau$\tabularnewline
                                                       &            &       \tabularnewline
      \cline{1-1} \cline{3-3}
    \end{tabular}
    \par\end{centering}
  \caption{\label{alg:Cminus11_ATruthTable}The procedure ${\cal A}_{2}.\protect\truthtable$
    takes as input a context with one question excluded, $C'=C\backslash q_{\protect\skp}$.
    It produces a truth table $\tau$ corresponding to it. Note that this
    is a randomised procedure (depends on the $\protect\QFHE$ encryption
    procedure) so for the same $C'$ the procedure may output different
    $\tau$s. The goal is to learn the probabilities of different $\tau$s
    appearing for each question $C'$.}
\end{algorithm}

\begin{algorithm}
  \begin{centering}
    \begin{tabular}{|>{\centering}p{0.7cm}|c|>{\centering}p{5cm}|c|>{\centering}p{3.8cm}|}
      \multicolumn{1}{>{\centering}p{0.7cm}}{}           & \multicolumn{1}{c}{}                                    & \multicolumn{1}{>{\centering}p{5cm}}{}                 & \multicolumn{1}{c}{}                       & \multicolumn{1}{>{\centering}p{3.8cm}}{}\tabularnewline
      \multicolumn{1}{>{\centering}p{0.7cm}}{${\cal A}$} & \multicolumn{1}{c}{}                                    & \multicolumn{1}{>{\centering}p{5cm}}{${\cal A}_{2}$}   & \multicolumn{1}{c}{}                       & \multicolumn{1}{>{\centering}p{3.8cm}}{}\tabularnewline
      \cline{1-1} \cline{3-3}
                                                       &                                                         &                                                        & \multicolumn{1}{c}{}                       & \multicolumn{1}{>{\centering}p{3.8cm}}{}\tabularnewline
                                                       &                                                         & \textbf{Phase 1}                                       & \multicolumn{1}{c}{}                       & \multicolumn{1}{>{\centering}p{3.8cm}}{}\tabularnewline
                                                       & {\Large $\longleftrightarrow$}                                       & Compute $p_{C'\tau}$ (as in the description)

      for each $C'$ by

      running ${\cal A}_{2}.\truthtable(C')$.          & \multicolumn{1}{c}{}                                    & \multicolumn{1}{>{\centering}p{3.8cm}}{}\tabularnewline
      \cline{1-1}
      \multicolumn{1}{>{\centering}p{0.7cm}}{}           &                                                         &                                                        & \multicolumn{1}{c}{}                       & \multicolumn{1}{>{\centering}p{3.8cm}}{}\tabularnewline
      \multicolumn{1}{>{\centering}p{0.7cm}}{}           &                                                         & Use $p_{C'\tau}$ to learn

      the questions $C'_{*0},C'_{*1}$ and

      the disjoint sets $T_{0},T_{1}$ of truth tables. & \multicolumn{1}{c}{}                                    & \multicolumn{1}{>{\centering}p{3.8cm}}{}\tabularnewline
      \multicolumn{1}{>{\centering}p{0.7cm}}{}           &                                                         &                                                        & \multicolumn{1}{c}{}                       & \multicolumn{1}{>{\centering}p{3.8cm}}{}\tabularnewline
      \multicolumn{1}{>{\centering}p{0.7cm}}{${\cal A}$} &                                                         & \textbf{Phase 2}                                       & \multicolumn{1}{c}{}                       & \multicolumn{1}{>{\centering}p{3.8cm}}{$C_{2}$}\tabularnewline
      \cline{1-1} \cline{5-5}
                                                       &                                                         &                                                        & {\Large $\xrightarrow{(C'_{*0},C'_{*1})}$} & \tabularnewline
                                                       &                                                         & $(\opad.\pk,\opad.\sk)\leftarrow\opad(1^{\lambda})$    &                                            & $\sk\leftarrow\QFHE.\gen(1^{\lambda})$

      $b\leftarrow\{0,1\}$

      $c_{b}\leftarrow\QFHE.\Enc_{\sk}(q_{*b})$\tabularnewline
                                                       & {\Large $\xleftarrow{(c_{b},\opad.\pk)}$}               &                                                        & {\Large $\overset{c_{b}}{\longleftarrow}$}              & \tabularnewline
                                                       & {\Large $\xrightarrow{(c_{\mathbf{a}'},\hat{k}'',s')}$} &                                                        &                                            & \tabularnewline
                                                       &                                                         & Samples $k\leftarrow K$                                &                                            & \tabularnewline
                                                       & {\Large $\circlearrowleft$}                                      & For each $q\in Q$, repeats $(q,k)$ to determine $\tau$ &                                            & \tabularnewline
                                                       &                                                         &                                                        &                                            & \tabularnewline
                                                       &                                                         & Set $b'=0$ if $\tau\in T_{0}$,

      and $b'=1$ if $\tau\in T_{1}$                    & {\Large $\overset{b'}{\longrightarrow}$}                             & \tabularnewline
                                                       &                                                         &                                                        &                                            & Accept if $b'=b$\tabularnewline
                                                       &        &       &       &       \tabularnewline
      \cline{1-1} \cline{3-3} \cline{5-5}
    \end{tabular}
    \par\end{centering}
  \caption{\label{alg:Cminus11_A_2_reduction}The algorithm ${\cal A}_{2}$ uses
    the adversary ${\cal A}$ for the compiled contextuality game $\protect\G'$,
    to break the $2$-IND security game for the $\protect\QFHE$ scheme.}
\end{algorithm}

\subsection{Proof Strategy}
\begin{lem}[Uniformly random $k$ is equivalent to the correct $k$]
  \label{lem:Cminus11_1guessK=00003DgoodK}Let $B_{0}$ (resp. $B_{1}$)
  be a PPT algorithm that takes $C'\in\Cprimeall$ as an input, interacts
  with ${\cal A}$ and outputs a bit, as described in \Algref{Cminus11_random-k-or-not}.
  Then, there is a negligible function $\negl$ such that $\left|\Pr[0\leftarrow\left\langle B_{0},{\cal A}\right\rangle ]-\Pr[0\leftarrow\left\langle B_{1},{\cal A}\right\rangle ]\right|\le\mathsf{negl}$.
\end{lem}

\begin{algorithm}
  \begin{centering}
    \begin{tabular}{|>{\raggedright}p{5cm}|>{\centering}p{3cm}|>{\raggedright}m{5.5cm}|}
      \multicolumn{1}{>{\raggedright}p{5cm}}{${\cal A}$} & \multicolumn{1}{>{\centering}p{3cm}}{}                  & \multicolumn{1}{>{\raggedright}m{5.5cm}}{$B_{0}(C')$ (resp. $B_{1}(C')$)}\tabularnewline
      \cline{1-1} \cline{3-3}
                                                        &           &       \tabularnewline
                                                         &                                                         & $\sk\leftarrow\QFHE.\gen(1^{\lambda})$

      $c\leftarrow\QFHE.\enc_{\sk}(C')$

      \,

      $(\opad.\pk,\opad.\sk)\leftarrow\opad(1^{\lambda})$\tabularnewline
                                                         & {\Large $\xleftarrow{(c,\opad.\pk)}$}                   & \tabularnewline
                                                         &                                                         & \tabularnewline
                                                         & {\Large $\xrightarrow{(c_{\mathbf{a}'},\hat{k}'',s')}$} & \tabularnewline
                                                         &                                                         & $B_{0}$ computes $k':=\opad.\Dec(\opad.\sk,s')$ and uses the secret
      key $\sk$ to compute $k''$ and then finds the $k$ satisfying $U_{k}=U_{k''}U_{k'}$.

      (resp. \textbf{$B_{1}$} samples a uniform $k\leftarrow K$).\tabularnewline
                                                         & {\Large $\overset{(q,k)}{\circlearrowleft}$}                     & Potentially rewinds ${\cal A}$ to this step and queries with $(q,k)$
      for arbitrary $q\in Q$.\tabularnewline
                                                         &                                                         & Runs an arbitrary procedure to compute a bit $b'$.\tabularnewline
                                                         &          &       \tabularnewline
      \cline{1-1} \cline{3-3}
    \end{tabular}
    \par\end{centering}
  \caption{\label{alg:Cminus11_random-k-or-not}Whether a PPT adversary ${\cal A}$
    for the compiled contextuality game $\protect\G'$ is used with the
    correct $k$ or a uniformly random $k$, it makes no difference, if all algorithms
    involved are PPT.}
\end{algorithm}

The following we will check carefully at the end. It says that one
can treat ${\cal A}$ as though it is consistent.
\begin{lem}[Consistency only helps]
  \label{lem:Cminus11_2EstimateAssumingConsistent}Let
  \begin{itemize}
    \item ${\cal C}_{k\leftarrow K}$ be exactly the same as the challenger 
          ${\cal C}$ for $\G'$ except that it samples $k\leftarrow K$ uniformly,
          instead of computing it correctly, let
    \item ${\cal A}$ be any PPT algorithm that plays $\G'$ and $\Pr[\accept\leftarrow\left\langle {\cal A},{\cal C}_{k\leftarrow K}\right\rangle ]=p$
          and denote by
    \item $p_{C'\tau}$ be as described above.
  \end{itemize}
  Then
  \[
    p\le\left(1-\frac{1}{|C|}\right)+\sum_{C\in\Call}\Pr(C)\sum_{q_{\skp}\in C}\frac{1}{|C|}\sum_{\tau}p_{C',\tau}\pred(\tau[C],C)
  \]
  where $C'=C\backslash q_{\skp}$, and $\Pr(C)$ denotes the probability
  with which ${\cal C}$ samples the context $C$.
\end{lem}

The first term captures the probability that the consistency test
passes and the second one captures the probability that the predicate
test passes.
\begin{lem}
  \label{lem:Cminus11_3PrecisionNonIssue}Exactly the same as \Lemref{3PrecisionNonIssue}
  except that
  \begin{itemize}
    \item ${\cal C}$ is the challenger for the game produced by the $(|C|-1,1)$
          compiler,
    \item instead of $p_{q'\tau}$, use $p_{C'\tau}$, and
    \item the construction of ${\cal A}_{2}$ is as in \Algref{Cminus11_A_2_reduction}.
  \end{itemize}
  Let ${\cal A}$ be such that $\Pr[\accept\leftarrow\left\langle {\cal A},{\cal C}\right\rangle ]\ge1-\frac{1}{|C|}+\frac{\valNC}{|C|}+\epsilon$
  and let ${\cal A}_{2,\epsilon}$ be as in \Lemref{3PrecisionNonIssue}.
  Then, it holds that
  \[
    \Pr[\accept\leftarrow\left\langle {\cal A}_{2,\epsilon},{\cal C}_{2}\right\rangle ]\ge\Pr[\accept\leftarrow\left\langle {\cal A}_{2},{\cal C}_{2}\right\rangle ]-O(\epsilon^{3}).
  \]

\end{lem}

The following is the contrapositive of the soundness guarantee in \Thmref{Cminus11_CompiledGameIsSecure}.
\begin{thm}[Soundness condition restated from \Thmref{Cminus11_CompiledGameIsSecure}]
  \label{thm:Cminus11_MainSoundness} Suppose
  \begin{itemize}
    \item ${\cal A}$ is any PPT algorithm that wins with probability $\Pr[\accept\leftarrow\left\langle {\cal A},{\cal C}\right\rangle ]\ge1-\frac{1}{|C|}+\frac{\valNC}{|C|}+\epsilon$
          for some non-negligible function $\epsilon$ and
    \item the $\opad$ used is secure, then
  \end{itemize}
  there is a PPT algorithm ${\cal A}_{2,\epsilon}$ that wins the $2$-IND
  security game of the $\QFHE$ scheme with probability
  \[
    \Pr[\accept\leftarrow\left\langle {\cal A}_{2,\epsilon},C_{2}\right\rangle ]\ge\frac{1}{2}+\nonnegl,
  \]
  where $\nonnegl$ is a non-negligible function that depends on $\epsilon$.
\end{thm}

In \Subsecref{Cminus11_step1}, we prove \Thmref{Cminus11_MainSoundness} assuming \Lemref{Cminus11_1guessK=00003DgoodK, Cminus11_2EstimateAssumingConsistent, Cminus11_3PrecisionNonIssue}.
In \Subsecref{Cminus11_step2}, we prove the lemmas.

\subsection{Proof assuming the lemmas (Step 1 of 2)}\label{subsec:Cminus11_step1}

{\begin{proof}
    From \Lemref{Cminus11_1guessK=00003DgoodK}, we have that $|\Pr[\accept\leftarrow\left\langle {\cal A},{\cal C}\right\rangle ]-\Pr[\accept\leftarrow\left\langle {\cal A},{\cal C}_{k\leftarrow K}\right\rangle ]|\le\negl$.
    Recall the definition of $p_{C'\tau}$ and use \Lemref{Cminus11_2EstimateAssumingConsistent}
    to write

    \begin{equation}
      \Pr[\accept\leftarrow\left\langle {\cal A},{\cal C}\right\rangle ]-\negl\le\left(1-\frac{1}{|C|}\right)+\sum_{C\in\Call}\Pr(C)\sum_{q_{\skp}\in C}\frac{1}{|C|}\sum_{\tau}p_{C',\tau}\pred(\tau[C],C)\label{eq:Cminus11_Pr_acc_A_C}
    \end{equation}
    where recall that $C':=C\backslash q_{\skp}$. Observe also that there
    exists $C'_{*0}\neq C'_{*1}$ such that
    \begin{equation}
      \left\Vert p_{C'_{*0}}-p_{C'_{*1}}\right\Vert _{1}:=\sum_{\tau}\left|p_{C'_{*0}\tau}-p_{C'_{*1}\tau}\right|\ge\eta\label{eq:Cminus11_p_q_star0tau_minus_p_q_star1tau}
    \end{equation}
    for some non-negligible function $\eta$. This is a consequence of
    the assumption that
    \begin{equation}
      \Pr[\accept\leftarrow\left\langle {\cal A},{\cal C}\right\rangle ]\ge1-\frac{1}{|C|}+\frac{\valNC}{|C|}+\epsilon\label{eq:Cminus11_A_wins_against_C_w_non-negl}
    \end{equation}
    for some non-negligible function $\epsilon$. To see this, suppose for contradiction that for all $C'_{*0}\neq C'_{*1}$, it were the
    case that $\sum_{\tau}|p_{C'_{*0}\tau}-p_{C'_{*1}\tau}|\le\negl$
    for some negligible function, then one could write, using \Eqref{Cminus11_Pr_acc_A_C}
    and $p_{\tau}:=p_{C'_{*0}\tau}$
    \begin{align*}
      \Pr[\accept\leftarrow\left\langle {\cal A},{\cal C}\right\rangle ]\le & \left(1-\frac{1}{|C|}\right)+\sum_{C\in\Call}\Pr(C)\cancel{\sum_{q_{\skp}\in C}\frac{1}{|C|}}\sum_{\tau}p_{\tau}\pred(\tau[C],C)+\negl' \\
      =                                                                     & \left(1-\frac{1}{|C|}\right)+\sum_{\tau}p_{\tau}\sum_{C\in\Call}\Pr(C)\pred(\tau[C],C)+\negl'                                           \\
      \le                                                                   & 1-\frac{1}{|C|}+\frac{\valNC}{|C|}+\negl'.
    \end{align*}
    But this contradicts \Eqref{Cminus11_A_wins_against_C_w_non-negl}
    and thus \Eqref{Cminus11_p_q_star0tau_minus_p_q_star1tau} holds for
    some $C'_{*0}\neq C'_{*1}$ as claimed.

    The remaining analysis is the same as the $(1,1)$ case, briefly,
    note that
    \begin{align*}
      \Pr[\accept\leftarrow\left\langle {\cal A}_{2},{\cal C}_{2}\right\rangle ]= & \frac{1}{2}\cdot\sum_{\tau\in T_{0}}\Pr[{\cal A}\text{ outputs }\tau|C'_{*0}\text{ was encrypted}]+ \\
                                                                                  & \frac{1}{2}\cdot\sum_{\tau\in T_{1}}\Pr[{\cal A}\text{ outputs }\tau|C'_{*1}\text{ was encrypted}]  \\
      =                                                                           & \frac{1}{2}+\frac{1}{2}\sum_{\tau\in T_{0}}(p_{C'_{*0}\tau}-p_{C'_{*1}\tau})                        \\
      \ge                                                                         & \frac{1}{2}+\frac{\eta}{4}.
    \end{align*}
    Since $\eta$ is non-negligible, existence of a PPT ${\cal A}$ satisfying
    \Eqref{Cminus11_A_wins_against_C_w_non-negl} breaks the security
    of the underlying $\QFHE$ scheme. The precision issue is handled
    by invoking \Lemref{Cminus11_3PrecisionNonIssue}. This completes
    the proof.
  \end{proof}
}

\subsection{Proof of the lemmas (Step 2 of 2)}\label{subsec:Cminus11_step2}

We only prove \Lemref{Cminus11_2EstimateAssumingConsistent}. The proofs of
\Lemref{Cminus11_1guessK=00003DgoodK} and \Lemref{Cminus11_3PrecisionNonIssue}
are analogous to those of \Lemref{1guessK=00003DgoodK} and
\Lemref{3PrecisionNonIssue}, respectively.

{\begin{proof}[Proof of \Lemref{Cminus11_2EstimateAssumingConsistent}]
    Consider the adversary ${\cal A}'$ in \Algref{Cminus11_red_for_lemma2}
    that uses ${\cal A}$ to interact with ${\cal C}_{k\leftarrow K}=:{\cal C}'$
    (not to be confused with $C'$ which denotes a context with exactly
    one question removed).
    \begin{algorithm}[h]
      \begin{centering}
        \begin{tabular}{|>{\raggedright}p{1cm}|>{\centering}p{2cm}|>{\raggedright}m{4cm}|>{\raggedright}m{2cm}|>{\raggedright}m{5cm}|}
          \multicolumn{1}{>{\raggedright}p{1cm}}{${\cal A}$} & \multicolumn{1}{>{\centering}p{2cm}}{}                  & \multicolumn{1}{>{\raggedright}m{4cm}}{${\cal A}'$}     & \multicolumn{1}{>{\raggedright}m{2cm}}{}                         & \multicolumn{1}{>{\raggedright}m{5cm}}{${\cal C}_{k\leftarrow K}=:\calC'$}\tabularnewline
          \cline{1-1} \cline{3-3} \cline{5-5}
                                                             &                                                         &                                                         & \centering{}                                                     & \tabularnewline
                                                             & {\Large $\xleftarrow{(c,\opad.\pk)}$}                   &                                                         & \centering{}{\Large $\xleftarrow{(c,\opad.\pk)}$}                & Recall: $c\leftarrow\enc_{\sk}(C')$ for $C'=C\backslash q_{\skp}$.\tabularnewline
                                                             &                                                         &                                                         & \centering{}                                                     & \tabularnewline
                                                             & {\Large $\xrightarrow{(c_{\mathbf{a}'},\hat{k}'',s')}$} &                                                         & \centering{}                                                     & \tabularnewline
                                                             &                                                         & $\tilde{k}\leftarrow K$                                 & \centering{}                                                     & \tabularnewline
                                                             & {\Large $\overset{(\tilde{q},\tilde{k})}{\circlearrowleft}$}     & Rewinds ${\cal A}'$ to learn the truth table $\tau$.    & \centering{}                                                     & \tabularnewline
                                                             &                                                         & Evaluate $c_{\tau[C']}\leftarrow\QFHE.\Eval_{\tau}(c).$ & \centering{}{\Large $\xrightarrow{(c_{\tau(q')},\hat{k}'',s')}$} & \tabularnewline
                                                             &                                                         &                                                         & \centering{}{\Large $\xleftarrow{(q,k)}$}                        & \tabularnewline
                                                             &                                                         &                                                         & \centering{}{\Large $\xrightarrow{a:=\tau(q)}$}                  & Recall: $\mathbf{a}'=\dec_{\sk}(c_{\tau[C']})=\tau[C']$

          $\ $

          If $q\in C'$,

          $\quad$accept if $a=\mathbf{a}'[q]$

          If $q\notin C'$ (i.e. $q=q_{\skp}$),

          $\quad$accept if ${\rm pred}(\mathbf{a}'\cup a,C)=1$ $\quad$where
          $\mathbf{a}'\cup a$ denotes answers indexed by $C$.\tabularnewline
          &         &           &       &       \tabularnewline
          \cline{1-1} \cline{3-3} \cline{5-5}
        \end{tabular}
        \par\end{centering}
      \caption{\label{alg:Cminus11_red_for_lemma2}${\cal A}'$ uses ${\cal A}$
        (a PPT algorithm), to play the compiled contextuality game $\protect\G'$.
        Its winning probability upper bounds that of ${\cal A}$ and can
        be computed in terms of $p_{C'\tau}$ for ${\cal A}$.}
    \end{algorithm}
    We show that
    \begin{align}
      \Pr[\accept\leftarrow\left\langle {\cal A},{\cal C}'\right\rangle ] & \le\Pr[\accept\leftarrow\left\langle {\cal A}',{\cal C}'\right\rangle ]\label{eq:Cminus11_Aprime_just_as_well}                                                   \\
                                                                          & =\left(1-\frac{1}{|C|}\right)+\sum_{C\in\Call}\Pr(C)\sum_{q_{\skp}\in C}\frac{1}{|C|}\sum_{\tau}p_{C',\tau}\pred(\tau[C],C)\label{eq:Cminus11_Aprime_at_most}
    \end{align}
    where recall that $C'=C\backslash q_{\skp}$.

    Note that for any given $C'$, the challenger asks any specific $q$
    with probability $1/|C|$, which in particular means that $q=q_{\skp}$
    with probability $1/|C|$ and $q\neq q_{\skp}$ with probability $1-1/|C|$.

    Let's derive \Eqref{Cminus11_Aprime_just_as_well}. Consider the interactions
    $\left\langle {\cal A},{\cal C}'\right\rangle $ and $\left\langle {\cal A}',{\cal C}'\right\rangle $.
    Conditioned on $C'$, there are two cases: (1) ${\cal A}$ is consistent,
    in which case, the answers given by ${\cal A}'$ and ${\cal A}$ are
    identical, or (2) ${\cal A}$ is inconsistent in which case ${\cal A}$
    fails with probability \emph{at least} $1/|C|$ (because the challenger
    spots the inconsistency with probability at least $1/|C|$), while
    ${\cal A}'$ fails with probability \emph{at most} $1/|C|$ (because
    it at most (potentially) fails the predicate evaluation, which happens
    with probability exactly $1/|C|$).

    As for \Eqref{Cminus11_Aprime_at_most}, it follows because ${\cal C}'$
    selects a context $C$ with probability $\Pr(C)$, it asks $q\neq q_{\skp}$
    with probability $1-1/|C|$ which corresponds to doing a consistency
    test---which ${\cal A}'$ passes with probability $1$ by construction.
    Finally, note that $\mathbf{a}'\cup a=\tau[C]$. Now, ${\cal C}'$
    asks $q=q_{\skp}$ with probability $1/|C|$ and in this case, ${\cal A}'$
    responds with $\tau[C]$ with probability $p_{C',\tau}$. Thus, its
    success probability in this case is the weighted average of $\pred(\tau[C],C)=\pred(\mathbf{a}'\cup a,C)$
    where the weights are given by $p_{C',\tau}$ (recall that $C'=C\backslash q_{\skp}$).
    This completes the proof.
  \end{proof}
}

\newpage{}

\pagebreak{}

\bibliographystyle{alpha}
\bibliography{Vertical}

\end{document}